\documentclass[journal,10pt]{IEEEtran}
\usepackage{pgfplots} 
\usepackage{comment}
\usepackage{amsfonts,enumitem}
\usepackage{textcomp}
\def\BibTeX{{\rm B\kern-.05em{\sc i\kern-.025em b}\kern-.08em
		T\kern-.1667em\lower.7ex\hbox{E}\kern-.125emX}}
\usepackage{amsmath,amssymb,multirow}
\usepackage{setspace}
\usepackage{listings}
\usepackage{cite}
\usepackage{stackengine}
\usepackage{acronym}
\usepackage{mathtools}
\usepackage{subfigure}
\usepackage{algorithm}
\usepackage{algpseudocode}
\algdef{SE}[SUBALG]{Indent}{EndIndent}{}{\algorithmicend\ }%
\algtext*{Indent}
\algtext*{EndIndent}
\acrodef{RIS}{reconfigurable intelligent surface}
\acrodef{SNR}{signal-to-noise ratio}
\acrodef{ISAC}{integrated sensing and communication}
\acrodef{ISLAC}{integrated sensing, localization, and communication}
\acrodef{LoS}{line-of-sight}
\acrodef{AoA}{angle-of-arrival}
\acrodef{AoD}{angle-of-departure}
\acrodef{UE}{user equipment}
\acrodef{BS}{base station}
\acrodef{MCRB}{misspecified Cram\'{e}r-Rao bound}
\acrodef{CRB}{Cram\'{e}r-Rao bound}
\acrodef{LB}{lower bound}
\acrodef{MML}{mismatched maximum-likelihood}
\acrodef{ML}{maximum-likelihood}
\acrodef{MIMO}{multiple-input multiple-output}

\allowdisplaybreaks

\usepackage{amsthm}
\theoremstyle{plain}

\newtheorem{rem}{Remark}
\theoremstyle{plain}

\newtheorem{lemma}{Lemma}

\pgfplotsset{
	tick label style={font=\small},
	label style={font=\small},
	legend style={font=\small}
}

\usepackage{ellipsis}
\usetikzlibrary{calc}
\usetikzlibrary{decorations.pathreplacing,decorations.markings,shapes.geometric}
\tikzset{naming/.style={align=center,font=\small}}
\tikzset{antenna/.style={insert path={-- coordinate (ant#1) ++(0,0.25) -- +(135:0.25) + (0,0) -- +(45:0.25)}}}
\tikzset{station/.style={naming,draw,shape=dart,shape border rotate=90, minimum width=25mm, minimum height=25mm,outer sep=0pt,inner sep=3pt}}
\tikzset{mobile/.style={naming,shape=circle,draw, minimum width=4mm,minimum height=1mm, outer sep=0pt,inner sep=3pt}}
\tikzset{risblockage/.style={naming,fill = red!30, very thick, shape=rectangle,minimum width=5mm,minimum height=18mm, outer sep=0pt,inner sep=3pt}}
\tikzset{radiation/.style={{decorate,decoration={expanding waves,angle=90,segment length=4pt}}}}

\newcommand{\RISS}[1]{%
	\begin{tikzpicture}
		\begin{scope}[
			yshift=-85,every node/.append style={
				yslant-=0.1,xslant=-1},yslant=-0.5,xslant=-1
			]
			\draw [step = 0.25, ultra thick, draw=blue, fill=black!20!white] (0,0) grid  (2,2) rectangle (0,0);
		\end{scope}
	\end{tikzpicture}
}

\newcommand{\rev}[1]{\textcolor{black}{#1}} % final submission after acceptance
\DeclareMathOperator*{\argmax}{arg\,max}
\DeclareMathOperator*{\argmin}{arg\,min}
\newcommand\abs[1]{\left|#1\right|}
\newcommand\norm[1]{\left\lVert #1\right\rVert}
\newcommand{\bp}{\boldsymbol{p}}
\newcommand{\bk}{\boldsymbol{k}}
\newcommand{\bcc}{\boldsymbol{c}}
\newcommand{\by}{\boldsymbol{y}}
\newcommand{\bz}{\boldsymbol{z}}

\newcommand{\bx}{\boldsymbol{x}}
\newcommand{\bmu}{\boldsymbol{\mu}}
\newcommand{\tmu}{\tilde{{\mu}}}
\newcommand{\tbmu}{\tilde{\boldsymbol{\mu}}}
\newcommand{\balp}{{\boldsymbol {\alpha}}}
\newcommand{\bbalp}{\overline{\boldsymbol {\alpha}}}
\newcommand{\bbp}{\overline{\boldsymbol{p}}}
\newcommand{\bb}{\boldsymbol{b}}
\newcommand{\bw}{\boldsymbol{w}}
\newcommand{\bwtilde}{\tilde{\bw}}
\newcommand{\bomeg}{\boldsymbol{\omega}}

\newcommand{\bbet}{\overline{\bet}}
\newcommand{\bet}{\boldsymbol{\eta}}
\newcommand{\bA}{\boldsymbol{A}}
\newcommand{\bB}{\boldsymbol{B}}

\newcommand{\dmax}{d_{\rm{max}}}
\newcommand{\qmax}{q_{\rm{max}}}

\newcommand{\bJ}{\boldsymbol{J}}
\newcommand{\bep}{\boldsymbol{\epsilon}}
\newcommand{\rmx}{\mathrm{x}}
\newcommand{\rmy}{\mathrm{y}}
\newcommand{\rmz}{\mathrm{z}}
\newcommand{\ba}{\boldsymbol{a}}
\newcommand{\bu}{\boldsymbol{u}}
\newcommand{\bh}{\boldsymbol{h}}
\newcommand{\bbf}{\boldsymbol{f}}
\newcommand{\bg}{\boldsymbol{g}}

\newcommand{\bR}{\boldsymbol{R}}
\newcommand{\bS}{\boldsymbol{S}}

\newcommand{\btheta}{\boldsymbol{\theta}}

\newcommand{\bG}{\boldsymbol{G}}
\newcommand{\bX}{\boldsymbol{X}}
\newcommand{\tbQ}{{\boldsymbol{Q}}} 
\DeclareMathOperator{\Tr}{Tr}
\newcommand{\projrange}[1]{\boldsymbol{\Pi}_{#1}}
\newcommand{\projnull}[1]{\boldsymbol{\Pi}^{\perp}_{#1}}
\newcommand{\Imatrix}{{ \boldsymbol{\mathrm{I}} }}

\newcommand{\Tthr}{T_{\rm{thr}}}
\newcommand{\Imax}{I_{\rm{max}}}
\newcommand{\Jmax}{J_{\rm{max}}}
\newcommand{\kappamax}{\kappa_{\rm{max}}}

\newcommand{\yy}{\boldsymbol{y}}
\newcommand{\aaa}{\boldsymbol{a}}
\newcommand{\ww}{\boldsymbol{w}}
\newcommand{\zetab}{\boldsymbol{\zeta}}
\newcommand{\bW}{\boldsymbol{W}}
\newcommand{\Gammab}{\boldsymbol{\Gamma}}
\newcommand{\Thetab}{\boldsymbol{\Theta}}
\newcommand{\pp}{\boldsymbol{p}}
\newcommand{\nn}{\boldsymbol{n}}

\newcommand{\Gammabt}{\widetilde{\Gammab}}

\newcommand{\kappahat}{\widehat{\kappa}}
\newcommand{\pphat}{\widehat{\pp}}
\newcommand{\zetabhat}{\widehat{\zetab}}
\newcommand{\alphahat}{\widehat{\alpha}}
\newcommand{\phihat}{\widehat{\phi}}
\newcommand{\bUps}{\boldsymbol{\Upsilon}}

\newcommand{\hermit}{\mathsf{H}}
\newcommand{\trpose}{\mathsf{T}}

\newcommand{\bpris}{\bp_{\text{RIS}}}
\newcommand{\bpbs}{\bp_{\text{BS}}}

\newcommand{\complexset}[2]{ \mathbb{C}^{#1 \times #2}  }

\newcommand{\realset}[2]{ \mathbb{R}^{#1 \times #2}  }

\newcommand{\bbetamin}{\overline{\beta}_{{\rm{min}}}}

\newcommand{\bkappa}{\overline{\kappa}}
\newcommand{\bphi}{\overline{\phi}}

\newcommand{\betamin}{\beta_{{\rm{min}}}}
\newcommand{\betaminhat}{\widehat{\beta}_{{\rm{min}}}}

\newcommand{\boldone}{{ {\boldsymbol{1}} }}

\newcommand{\BSt}{\text{BS}}
\newcommand{\RISt}{\text{RIS}}

\newcommand{\mcrb}{{\rm{MCRB}}}
\newcommand{\crb}{{\rm{CRB}}}
\newcommand{\lb}{{\rm{LB}}}
\newcommand{\bias}{{\rm{Bias}}}

\newcommand{\ppbs}{\pp_{\BSt}}
\newcommand{\ppris}{\pp_{\RISt}}

\newcommand{\BS}[1]{%
	\begin{tikzpicture}
		\node[station] (base) {#1};
		
		\draw[line join=bevel] (base.100) -- (base.80) -- (base.110) -- (base.70) -- (base.north west) -- (base.north east);
		\draw[line join=bevel] (base.100) -- (base.70) (base.110) -- (base.north east);
		
		\draw[line cap=rect] ([yshift=0pt]base.north) [antenna=1];
	\end{tikzpicture}
}

\newdimen\bpt
\def\mobile#1{\leavevmode 
	\bpt=#1bp \hbox to7\bpt{\kern1\bpt \lower1\bpt\vbox to12\bpt{}%
		\pdfliteral{q #1 0 0 #1 0 0 cm 1 j 2 w 0 0 5 10 re B 
			1 g 1 G  1 w .3 1.8 4.4 7 re B 
			1.5 w 2.5 .2 0 .1 re B .3 w 1.7 10 1.6 0 re B Q}%
		\hss}}

\newcommand{\UEE}[1]{%
	\mobile{2}
}

\newcommand{\RIS}[1]{%
	\begin{tikzpicture}
		\draw[step=0.125,thick, draw = blue] (0,0) grid (1.5,1.5);
	\end{tikzpicture}
}

\begin{document}
\bstctlcite{IEEEexample:BSTcontrol}

%%%%%%%%%%%%%%%%%% title page information %%%%%%%%%%%%%%%%%%
%\renewcommand{\thepage}{}
\title{RIS-aided Near-Field Localization under Phase-Dependent Amplitude Variations}

% 	\author{C\"{u}neyd \"{O}zt\"{u}rk\IEEEauthorrefmark{1}, Musa Furkan Keskin\IEEEauthorrefmark{2}, Henk Wymeersch\IEEEauthorrefmark{2}, Sinan Gezici\IEEEauthorrefmark{1}\\
% 	\IEEEauthorrefmark{1}Department of Electrical
% and Electronics Engineering, Bilkent University, Turkey\\
% \IEEEauthorrefmark{2}Department of Electrical
%  Engineering, Chalmers University of Technology, Sweden}
 
 	\author{Cuneyd Ozturk,\thanks{C. Ozturk and S. Gezici are with the Department of Electrical and Electronics Engineering, Bilkent University, Ankara, 06800, Turkey, E-mails:	\{cuneyd,gezici\}@ee.bilkent.edu.tr} \emph{Student Member, IEEE}, Musa Furkan Keskin, \emph{Member, IEEE},
     \\Henk Wymeersch,\thanks{M. F. Keskin and H. Wymeersch are with the Department of Electrical Engineering, Chalmers University of Technology, Sweden, E-mails: \{furkan,henkw\}@chalmers.se} \emph{Senior Member, IEEE}, and Sinan Gezici, \emph{Senior Member, IEEE}
     %\thanks{This work significantly extends our previous conference publication in \cite{cuneyd_ICC_RIS_2022} in the following aspects. First, the problem of RIS-aided near-field localization under model mismatch is investigated by providing complete details of all theoretical results regarding the pseudo-true parameter and the \ac{MCRB}. Second, we consider two additional scenarios where the UE has the knowledge of the RIS amplitude model with unknown and known model parameters. For these new scenarios, the theoretical performance bounds are derived, and novel localization and online RIS calibration algorithms are proposed. Moreover, extensive simulations are carried out to explore the performances of the proposed algorithms under a wide variety of operating conditions.}
     %Part of this work will be presented at IEEE International Conference on Communications (ICC), June 2022 . 
     %In the conference version of this paper \cite{cuneyd_ICC_RIS_2022}, RIS aided near-field localization is performed only when the knowledge of the RIS amplitude model is missing, and the details of theoretical results are not presented completely. In this paper, . In addition, when  is known,  are derived. Moreover,  a novel localization and online RIS calibration algorithms are proposed for the cases with and without knowledge of the RIS amplitude model. In all cases, performances of the proposed algorithms are investigated with performing .}
     \thanks{This work was supported, in part, by the EU H2020 RISE-6G project under grant 101017011 and by the MSCA-IF grant 888913 (OTFS-RADCOM).}
     \vspace{-0.5cm}}
 
	\maketitle
	\begin{abstract}
		We investigate the problem of reconfigurable intelligent surface (RIS)-aided near-field localization of a user equipment (UE) served by a base station (BS) under phase-dependent amplitude variations at each RIS element. Through a misspecified Cram\'{e}r-Rao bound (MCRB) analysis and a resulting lower bound (LB) on localization, we show that when the UE is unaware of amplitude variations (i.e., assumes unit-amplitude responses), severe performance penalties can arise, especially at high signal-to-noise ratios (SNRs). Leveraging Jacobi-Anger expansion to decouple range-azimuth-elevation dimensions, we develop a low-complexity approximated mismatched maximum likelihood (AMML) estimator, which is asymptotically tight to the LB. To mitigate performance loss due to model mismatch, we propose 
	to jointly estimate the UE location and the RIS amplitude model parameters. The corresponding Cram\'{e}r-Rao bound (CRB) is derived, as well as an iterative refinement algorithm, which employs the AMML method as a subroutine and alternatingly updates individual parameters of the RIS amplitude model. 
	Simulation results indicate fast convergence and performance close to the CRB. The proposed method can successfully recover the performance loss of the AMML under a wide range of RIS parameters and effectively calibrate the RIS amplitude model online with the help of a user that has an a-priori unknown location.

\end{abstract}

\section{Introduction}\label{sec:Intro}
Among the envisioned technological enablers for 6G, three stand out as being truly disruptive: the transition from 30 GHz to beyond 100 GHz (the so-called higher mmWave and lower THz bands) \cite{saad2019vision,rappaport2019wireless,tataria20216g}, the convergence of communication, localization, and sensing (referred to as \ac{ISAC} or \ac{ISLAC}) \cite{chiriyath2017radar,liu2021integrated,de2021convergent,wymeersch2021integration,JCS_2021}, and the introduction of \acp{RIS} \cite{RIS_tutorial_2021,elzanaty20216g,wymeersch2020radio}. RISs are large passive metasurfaces, comprising arrays of programmable reflective unit cells, and have the ability to shape the propagation environment by judiciously adjusting the phase shifts at each reflecting element, thus locally boosting the \ac{SNR} to improve communication quality \cite{RIS_Access_2019,RIS_commag_2021,RIS_WCM_2021,RIS_EE_TWC_2019}. This is especially relevant in beyond 100 GHz to overcome sudden drops in rate caused by temporary blockage of the \ac{LoS} path \cite{RIS_WCM_2021,LOS_NLOS_NearField_2021}. In order to provide enhanced performance in downlink single- and multi-user systems, passive reflect beamforming at the RIS can be optimized, potentially together with active beamforming at the \ac{BS}, to maximize energy efficiency \cite{RIS_EE_TWC_2019,distRIS_EE_2022}, sum-rate \cite{hybridBF_RIS_JSAC_2020,DRL_RIS_JSAC_2020,RIS_sumrate_2020} and mutual information \cite{jointDecomp_TCOM_2021}, as well as to minimize total transmit power at the \ac{BS} \cite{JointActivePassiveRIS_TWC_2019,outage_RIS_TSP_2021}.

In parallel with the benefits for communications, \acp{RIS} can similarly improve localization performance \cite{RIS_2018_TSP}. Stronger even, \acp{RIS} with known position and orientation enjoy the ability to enable localization in scenarios where it would otherwise be impossible \cite{Keykhosravi2020_SisoRIS}. In this respect, the large aperture of the \ac{RIS} has several interesting properties. First of all, the \ac{SNR}-boosting provides accurate delay measurements when wideband signals are used \cite{RIS_bounds_TSP_2021,rappaport2019wireless}. Secondly, the large number of elements provides high resolution in \ac{AoA} (for uplink localization) or \ac{AoD} (for downlink localization) \cite{rappaport2019wireless}. Third, when the \ac{UE} is close to the \ac{RIS} (in the sense that the distance to the \ac{RIS} is of similar order as the physical size of the \ac{RIS}), wavefront curvature effects (so-called geometric near-field) can be harnessed to localize the user \cite{RIS_2018_TSP,Shaban2021,nearFieldRIS_LOSBlock_2022,EM_wavefront,rahal2021risenabled,LOS_NLOS_NearField_2021,nearfieldTrack_TSP_2021}, even when the \ac{LoS} path between the \ac{BS} and \ac{UE} is blocked, irrespective of whether wideband or narrowband signals are used. Moreover, closed-form RIS phase profile designs taking into account the spherical wavefront can be employed to improve localization accuracy under near-field conditions \cite{RIS_bounds_TSP_2021,rahalNearFieldProfileDesign_2022}. As a step further, joint benefits in \ac{ISLAC} applications can be reaped via RIS phase profile adjustment by simultaneous optimization of localization and communications metrics \cite{RIS_Location_Win_2022}.
	%
% 	\begin{figure}
% 	\centering
% 	    \begin{tikzpicture}[every path/.append style={thick}]
%             \matrix[column sep=.75cm,row sep=.75cm]
%             {
%             & \node(C) [label={below:{\textbf{RIS}}}]{\RIS{}}; &  &  \\
%             \node[every path/.append style={thick},inner sep=0pt](A){\BS{\textbf{BS}}}; & \node(D){\RISBlockage{LoS \\Blockage}}; &   &\node[thick](B){\UE{\textbf{UE}}}; \\
%             };
%             \draw[thick,radiation,decoration={angle=45}] (A.north) -- +(45:0.5);
%             \draw[thick] (A.north)--(C.center);
%             \draw[thick] (C.center)--([xshift = -.1768cm, yshift = -.1768cm]B.north);
%             \draw[thick,radiation,decoration={angle=45}] ([xshift = -.6cm, yshift = .2768cm]B.north)-- + (-45:0.5);
%         \end{tikzpicture}
%         \caption{Configuration of a RIS-aided localization system with LoS blockage.}
%         \vspace{-0.5cm}
%          \label{fig:RISconfiguration}
% 	\end{figure}

%     \begin{figure}
%         \centering
%         \includegraphics[width = 1.1\linewidth]{RISfigure.pdf}
%  \caption{Configuration of a RIS-aided localization system with LoS blockage.}
%         \label{fig:RISconfiguration}
%     \end{figure}

    	\begin{figure}
	\centering
	    \begin{tikzpicture}[every path/.append style={thick},auto]

            \matrix[column sep=.4cm,row sep=.75cm]
            {
            & \node(C) [label={[xshift = 0.6cm, yshift = 0.6cm]below:{\textbf{RIS}}}]{\RISS{}}; &   &  \\
            \node[every path/.append style={thick},inner sep=0pt](A){\BS{\textbf{BS}}}; & \node(D){ \includegraphics[width=.2\textwidth]{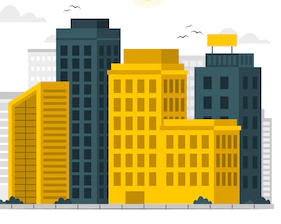}}; &   &\node[thick](B)[label={below:{\textbf{UE}}}]{\UEE{}}; &\\
            };

            \draw[thick,radiation,decoration={angle=45}] ([xshift = .1cm, yshift = .1cm]A.north) -- +(45:0.5);
        
             \draw[thick,radiation,decoration={angle=45}] ([xshift = -.3cm, yshift = .35cm]B.north)-- + (-45:0.5);
            
            \draw[-stealth, ultra thick, draw = red] (A.north)--(C.center);
            
             \draw[-stealth, ultra thick, draw = red] (C.center) --(B.north);
             
             \draw [-stealth, dashed, draw = red] (A.north) to  node [above, sloped] (TextNode1) {LoS Path} ([xshift = -1.2cm, yshift = -0.8cm]D.north);

          \draw[-stealth,ultra thick, draw = green] (0,2.5)  node[left, xshift = -1cm, text width=3cm, text centered] {
        \begin{minipage}{\textwidth}
            \begin{align*}
               w_{t,m} &= \beta(\theta_{t,m}) e^{j\theta_{t,m}}            \end{align*}
        \end{minipage}
    }to[out=15,in=45](-3.5,2.5);

        \end{tikzpicture}
        \caption{Configuration of a RIS-aided localization system with LoS blockage.}
         \label{fig:RISconfiguration}
	\end{figure}
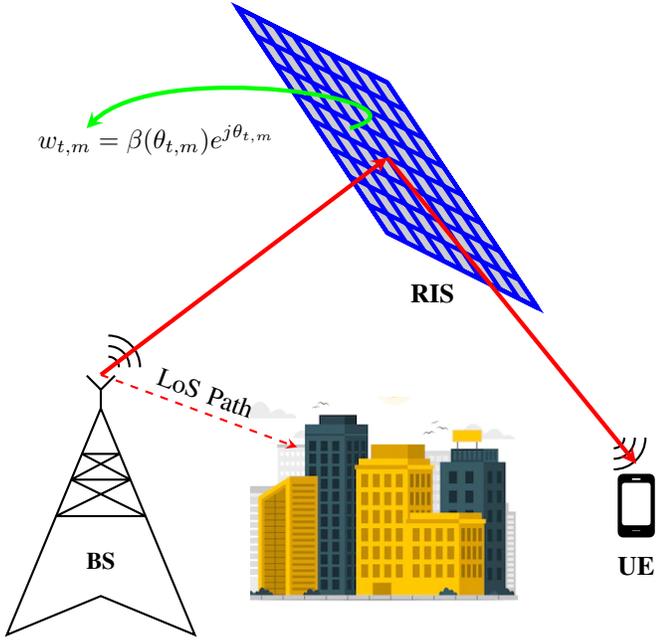
	
	In ISLAC scenarios, critical to the effective utilization of \ac{RIS} is the control of individual \ac{RIS} elements, commonly through phase shifters, which provide element-by-element control with a certain resolution and allow coherent combination of  paths to/from the \ac{RIS} \cite{RIS_sumrate_2020,RIS_commag_2021}. For localization, in contrast to communication, the receiver should be equipped with the knowledge of the \ac{RIS} phase profiles to apply suitable high-accuracy processing methods \cite{wymeersch2020radio}. The ability to modulate the \ac{RIS} phase profiles brings additional benefits, such as separating the controlled and uncontrolled multipath through temporal coding \cite{keykhosravi2021multiris}. Hence, the ability to control the RIS in a precise and known manner is essential for \ac{ISLAC} applications, which necessitates the availability of accurate and simple \ac{RIS} phase control models. Such models should ideally account for the per-element response \cite{abeywickrama2020intelligent}, the finite quantization of the control \cite{RIS_phase_quantization_2021,RIS_commag_2021}, mutual coupling \cite{RIS_coupling_2021}, calibration effects, and power losses. Most studies on \ac{RIS} localization have considered ideal phase shifters (e.g. \cite{nearFieldRIS_LOSBlock_2022,Keykhosravi2020_SisoRIS,RIS_bounds_TSP_2021,RIS_loc_2021_TWC,Shaban2021,rahal2021risenabled,LOS_NLOS_NearField_2021}), omitting the listed impairments. How these proposed localization approaches fare under these impairments is both unknown and important. 
	
	In this manuscript, we investigate the problem of RIS-aided geometric near-field localization of a single-antenna \ac{UE} served by a single-antenna \ac{BS} under \ac{LoS} blockage \cite{nearFieldRIS_LOSBlock_2022,RIS_ANM_2021,rahal2021risenabled,LOS_NLOS_NearField_2021}, considering a realistic \ac{RIS} amplitude model \cite{abeywickrama2020intelligent}, which relies on equivalent circuit models of individual reflecting elements. Specifically, we quantify degradation in localization performance due to mismatch between an \textit{ideal model} with unit-amplitude RIS element responses and a \textit{realistic model} with phase-dependent amplitude variations \cite{abeywickrama2020intelligent,PDAV_TVT_2021}, by resorting to the \ac{MCRB} analysis \cite{Fortunati2017}. In addition, we develop novel localization and online RIS calibration algorithms for cases with and without the knowledge of the underlying RIS amplitude model. The main contributions and novelty of our manuscript can be summarized as follows:
	\begin{itemize}
	    \item \textbf{MCRB Analysis of Near-Field Localization under RIS Non-Idealities:} Employing the \ac{MCRB} \cite{Ricmond2015MCRB,Fortunati2017} as a tool to assess the accuracy loss under model mismatch, we provide a simple expression to find the \textit{pseudo-true parameter} for the considered scenario and derive the \ac{MCRB} of the pseudo-true parameter and the \ac{LB} of the true parameter. The MCRB analysis over a wide range of RIS model parameters reveals an order-of-magnitude localization performance degradation due to model misspecification at high \acp{SNR}, both in terms of the \ac{LB} and the \ac{MML} estimator \cite{Fortunati2017}. In contrast, when the true phase control model is available, localization performance is relatively stable, for all considered model parameter settings.  
	    \item \textbf{Low-Complexity Near-Field Localization via Jacobi-Anger Expansion:} Building upon the ideas in our recent conference papers \cite{cuneyd_ICC_RIS_2022,Shaban2021}, we develop a novel low-complexity near-field localization algorithm using Jacobi-Anger expansion, which enables decoupling of range, azimuth and elevation dimensions. The resulting algorithm, named approximated MML (AMML), avoids the costly 3-D search over the UE position by performing three 1-D searches and attains the corresponding theoretical limits.
	    \item \textbf{Joint Localization and Online RIS Calibration:} Under the assumption of known RIS amplitude model with unknown parameters, we propose an efficient approximate \ac{ML} (AML) algorithm for joint localization and online RIS calibration. The proposed approach iteratively updates the RIS model parameters based on an initial UE location estimate from the output of the model-unaware AMML method, and refines the UE location using the updated RIS model. The AML algorithm is shown to significantly outperform the AMML estimator at high SNRs (where degradation due to model mismatch is most evident), closing the performance gap with respect to the case with known model parameters, and converges quickly to the corresponding model-aware CRB in few iterations.
	\end{itemize}

\section{System Model}\label{sec:System}
	
In this section, we describe the system geometry, present the signal model and the RIS  model, and formulate the problem of interest.

\subsection{Geometry and Signal Model}\label{sec_sig_mod}
	
	%Based on the system model in \cite{Keykhosravi2020_SisoRIS}, 
We consider an RIS-aided localization system (see Fig.~\ref{fig:RISconfiguration}) with a single-antenna BS, an $M$-element RIS, and a single-antenna UE having the following three-dimensional (3-D) locations: $\bp_{\text{BS}}$ denotes the known BS location, $\bp_{\text{RIS}}= [\rmx_{\text{R}}\, \rmy_{\text{R}}\, \rmz_{\text{R}}]^{\trpose}$ is the known center of the RIS, $\bp_{m} = [\rmx_m\, \rmy_m \, \rmz_m]^{\trpose}$ represents the known location of the $m$th RIS element for $1\leq m\leq M$, and $\bbp = [\overline{\rmx} ~ \overline{\rmy} ~ \overline{\rmz}]^{\trpose}$ is the unknown UE location. For convenience, following the notation given in \cite{Fortunati2017, Fortunati2018Chapter4}, we use $\overline{(\cdot)}$ for the true values of the parameters of interest throughout the manuscript. 
	
In the considered setting, the BS broadcasts a narrowband signal $s_t$ over $T$ transmissions under the constraint of $\mathbb{E}\{\abs{s_t}^2\} = E_s$. For simplicity, we assume that $s_t = \sqrt{E_s}$ for any $t$. Assuming \ac{LoS} blockage \cite{nearFieldRIS_LOSBlock_2022,LOS_NLOS_NearField_2021} and the absence of uncontrolled multipath\rev{\footnote{\label{fn_multipath}\rev{In RIS-aided communications and localization, the uncontrolled multipath can usually be ignored \cite{RIS_bounds_TSP_2021,zhang2020towards,RIS_Location_Win_2022} or treated as an extra additive disturbance \cite{RIS_loc_2021_TWC} due to the spatial filtering capability of the RIS with tunable phase shifts, which makes the BS-RIS-UE path much stronger than the uncontrollable paths.}}}, the signal received by the UE involves only reflections from the RIS and can be expressed at transmission $t$ as 
\begin{equation}\label{eq_yt}
		y_t  = \overline{\alpha}\underbrace{\ba^{\trpose}(\bbp) \text{diag}(\bw_t) \ba(\bp_{\text{BS}})}_{\triangleq \bb^{\trpose}(\bbp)\bw_t}s_t  + n_t\, ,
\end{equation}
where $\overline{\alpha}$ is the unknown channel gain \rev{including the effects of path loss, directivity of the RIS elements and of the antennas of the BS and UE \cite[Eq.~(3)]{LOS_NLOS_NearField_2021}, and polarization mismatch\footnote{\label{fn_global}\rev{Since the RIS aperture size is much smaller than the distances from the RIS to the BS and UE in the considered geometry, the channel gain appears as a global factor that has the same value for each RIS element \cite[Eq.~(3)]{EM_wavefront}.}} \cite[Ch.~1-11]{milligan2005modern}}, $\bw_t = [w_{t,1}\, \ldots\, w_{t,M}]^{\trpose}$ is the RIS profile at transmission $t$, and $n_t$ is uncorrelated zero-mean additive Gaussian noise with variance $N_0/2$ per real dimension. %It is assumed that $n_t$'s are independent and identically distributed. 
Moreover, 
\begin{align} \label{eq_bp}
    \bb(\bp) = \ba(\bp) \odot \ba(\bp_{\text{BS}}) ~,    
\end{align}
where \rev{$\odot$ represents the Hadamard (element-wise) product and} $\ba(\bp)$ is the near-field RIS steering vector for a given position $\bp$, defined as
	\begin{align} \label{eq_ap_nearfield}
			[\ba(\bp)]_{m} &= \exp\left(-j \frac{2\pi}{\lambda}\left(\norm{\bp-\bp_m}-\norm{\bp-\bp_{\text{RIS}}}\right)\right) 
		\end{align}
for $m\in\{1, \ldots, M\}$, in which $\lambda$ denotes the signal wavelength \rev{and the RIS center $\ppris$ is chosen as the reference location \cite{EM_wavefront}, \cite{nearfieldTrack_TSP_2021}. As the distance to the RIS becomes sufficiently large with respect to the RIS size, i.e., when $\norm{\bp-\bpris} \gg \norm{\bp_m - \bpris}$, the near-field steering vector in \eqref{eq_ap_nearfield} reverts to its standard far-field counterpart as follows \cite[Eq.~(9)]{nearfield_Friedlander_2019} (see Appendix~\ref{sec_app_farfield} for details): 
	\begin{equation}
	     [\ba(\bp)]_m \rev{\approx [\ba(\vartheta, \varphi)]_m} = \exp \left(-j (\pp_m \rev{- \ppris})^\trpose \boldsymbol{k}(\vartheta,\varphi)\right)~, \label{eq:CM1}
	\end{equation}
	where
	\begin{align} \label{eq_k_wavevector}
	    \bk(\vartheta,\varphi) = - \frac{2 \pi}{\lambda} [\sin \vartheta \cos \varphi  ~ \sin \vartheta \sin \varphi ~ \cos \vartheta  ]^{\trpose}
	\end{align}
	is the wavevector, $\vartheta\in[0, \, \pi/2]$ denote the elevation angle between the $z$-axis and $\bp - \rev{\bpris}$ (assuming the RIS orientation aligns with the $x$-$y$ plane \cite[Fig.~1]{Shaban2021}), and $\varphi\in [0, \, 2\pi]$ represent the azimuth angle between the projection of $\bp - \rev{\bpris}$ on the $x$-$y$ plane and the $x$-axis, measured counter-clockwise.}
	
	\rev{For the considered geometry, the UE and the BS are assumed to be located in the Fresnel (radiative) near-field region with respect to the RIS, which corresponds to the interval \cite{Fresnel_2011,Fresnel_2016,nearfieldTrack_TSP_2021}
    \begin{align} \label{eq_dmax}
        0.62 \sqrt{\frac{D^3}{\lambda}} \leq d \leq \frac{2 D^2}{\lambda} ~,
    \end{align}
    where $D$ is the RIS aperture size (i.e., the largest distance between any two RIS elements) and $d = \norm{\pp - \ppris}$ with $\pp$ denoting the UE or BS location. Operating in the near-field enables us to localize the UE under \ac{LoS} blockage by harnessing the \textit{spherical wavefront} of the impinging signal at the UE over the BS-RIS-UE path \cite{EM_wavefront,LOS_NLOS_NearField_2021,nearfieldTrack_TSP_2021}. As seen from the structure in \eqref{eq_ap_nearfield}, the unknown UE location $\bbp$ can be directly estimated from the received signal in \eqref{eq_yt} by exploiting the phase shifts across the RIS elements (i.e., $m = 1, \ldots, M$) over multiple observations (i.e., $t = 1, \ldots, T$). In far-field, it is not possible to localize the UE with a single RIS under \ac{LoS} blockage since the UE can only estimate the \ac{AoD} from the RIS (the far-field steering vector in \eqref{eq:CM1} depends only on the angles), which is not sufficient for localization. Note that the UE cannot extract any delay information in far-field in the considered setup due to narrowband transmission and asynchronism between the UE and the BS.}

	\subsection{Model for RIS Elements}\label{sec_model_ris}
% 	\rev{The proposed model in \cite{abeywickrama2020intelligent} is more realistic than the naive model. We use that model as its accuracy is illustrated for a practical reflecting element. Our purpose is to show that localization is sensitive to minor model mismatches. In the references 
% \cite{RIS_Circuit_2020,RIS_Circuit_Gradient_2012,RIS_Circuit_Mag_2012}, the amplitude model of the intelligent surfaces is not presented in a precise manner so that we can exploit it in our analyses. On the other hand, \cite{abeywickrama2020intelligent} provides a closed-form expression for the 
% RIS amplitude model which allows us to make some theoretical analyses as we do throughout the manuscript.}

	Following the practical model in \cite{abeywickrama2020intelligent}\rev{\footnote{\label{fn_practical}\rev{The RIS phase-amplitude model in \cite[Eq.~(5)]{abeywickrama2020intelligent} covers a wide range of semiconductor devices employed in RIS implementation, such as a positive-intrinsic-negative (PIN) diode and a variable capacitance diode, and agrees well with the experimental results in the literature \cite{abeywickrama2020intelligent,zhu2013active,RIS_Circuit_2020,RIS_Circuit_Gradient_2012}. It is possible to derive more accurate RIS phase-amplitude relations than in \cite[Eq.~(5)]{abeywickrama2020intelligent} using a generalized sheet transition conditions (GSTC) based model \cite{RIS_Circuit_Mag_2012} instead of an equivalent circuit model. However, the goal of this paper is not to derive a perfectly accurate model of RIS element responses for all types of RISs (which is not practically possible), but rather to quantify the sensitivity of RIS-aided localization to deviations from the naive unit-amplitude model commonly used in the literature \cite{nearFieldRIS_LOSBlock_2022,Keykhosravi2020_SisoRIS,RIS_bounds_TSP_2021,RIS_loc_2021_TWC,Shaban2021,rahal2021risenabled,LOS_NLOS_NearField_2021} by adopting an already available, fairly realistic model with a closed-form phase-amplitude relation as in  \cite[Eq.~(5)]{abeywickrama2020intelligent}. As will be discussed in Sec.~\ref{sec_mcrb_analysis}, the proposed MCRB analysis is generic and applicable to any RIS model with closed-form mapping between element phases and amplitudes.}}}, we consider \textit{phase-dependent amplitude variations} of the RIS elements given by %$w_{t,m}$ is given by 
	\begin{align}\label{eq_wtm_ris}
		w_{t,m} = \beta(\theta_{t,m}) e^{j\theta_{t,m}},
	\end{align}
	with $\theta_{t,m}\in [-\pi, \pi)$ and $\beta(\theta_{t,m})\in[0,1]$ denoting the phase shift\rev{\footnote{\label{fn_phase}\rev{Given the near-optimal performance of RIS-aided systems using phase shifters with few quantization bits \cite{hybridBF_RIS_JSAC_2020,RIS_phase_quantization_2021,uplink_PN_WCL_2021} and the availability of such low-resolution RIS phase shifters in practice \cite{RIS_Circuit_2020,RIS_phase_quantization_2021}, we assume continuous phase shifts in our model.}}} and the corresponding amplitude, respectively. In particular, $\beta(\theta_{t,m})$ is expressed as
	\begin{equation}\label{eq_beta_model}
		\beta(\theta_{t,m}) = (1-\bbetamin)\left(\frac{\sin(\theta_{t,m}-\bphi) + 1}{2}\right)^{\bkappa} + \bbetamin,
	\end{equation}
	where $\bbetamin\geq 0$, $\bphi\geq 0$, and $\bkappa \geq 0$ are the constants related to the specific circuit implementation \cite{abeywickrama2020intelligent}. To illustrate amplitude variations in \eqref{eq_beta_model}, Fig.~\ref{fig:1b} plots $	\beta(\theta)$ as a function of the applied phase shift $\theta$ for various values of $\bbetamin$ when $\bkappa = 1.5$ and $\bphi = 0$. As seen from the figure, larger amplitude fluctuations occur as $\bbetamin$ approaches $0$. The resulting performance penalties in location estimation will be quantified through the \ac{MCRB} analysis in Section~\ref{sec_mcrb_analysis}.

		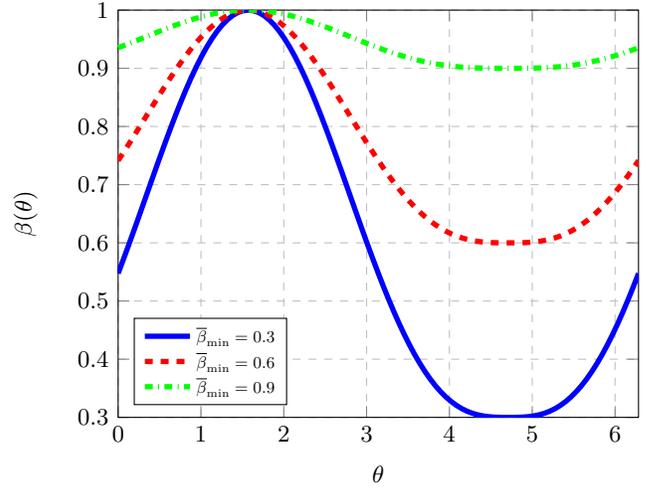
\begin{figure}%[H]
		\centering
		\begin{tikzpicture}
			\begin{axis}[
				width=8.5cm,
                height=7cm,
				legend style={nodes={scale= 0.7, transform shape},at={(1,1)},anchor=north east}, 
				legend cell align={left},
				legend image post style={mark indices={}},
				xlabel={$\theta$},
				ylabel={$\beta(\theta)$},
				xmin=0, xmax=2*pi,
				ymin=0.3, ymax=1,
				xtick={0, 1, 2, 3, 4, 5, 6},
				ytick={0.3,0.4,0.5, 0.6, 0.7, 0.8, 0.9, 1},
				legend pos=south west,
				ymajorgrids=true,
				xmajorgrids=true,
				grid style=dashed,
				]
				
				\addplot[thick,
				color=blue,
				line width = 2,
				mark = ,
				mark indices={    1,5,9,13,17,21,25,29,33,37,41,45,49,53,57,61,65,69,73,77,81,85,89,93,97,101},
				mark options={solid},			
				mark size = 4pt,
				]
				coordinates {
               (0,0.54749)(0.062832,0.57116)(0.12566,0.59544)(0.1885,0.62021)(0.25133,0.64533)(0.31416,0.67066)(0.37699,0.69604)(0.43982,0.72134)(0.50265,0.74639)(0.56549,0.77105)(0.62832,0.79515)(0.69115,0.81855)(0.75398,0.8411)(0.81681,0.86264)(0.87965,0.88304)(0.94248,0.90217)(1.0053,0.91988)(1.0681,0.93608)(1.131,0.95063)(1.1938,0.96346)(1.2566,0.97446)(1.3195,0.98357)(1.3823,0.99072)(1.4451,0.99586)(1.508,0.99896)(1.5708,1)(1.6336,0.99896)(1.6965,0.99586)(1.7593,0.99072)(1.8221,0.98357)(1.885,0.97446)(1.9478,0.96346)(2.0106,0.95063)(2.0735,0.93608)(2.1363,0.91988)(2.1991,0.90217)(2.2619,0.88304)(2.3248,0.86264)(2.3876,0.8411)(2.4504,0.81855)(2.5133,0.79515)(2.5761,0.77105)(2.6389,0.74639)(2.7018,0.72134)(2.7646,0.69604)(2.8274,0.67066)(2.8903,0.64533)(2.9531,0.62021)(3.0159,0.59544)(3.0788,0.57116)(3.1416,0.54749)(3.2044,0.52455)(3.2673,0.50245)(3.3301,0.48129)(3.3929,0.46117)(3.4558,0.44215)(3.5186,0.42431)(3.5814,0.40769)(3.6442,0.39233)(3.7071,0.37827)(3.7699,0.3655)(3.8327,0.35403)(3.8956,0.34385)(3.9584,0.33492)(4.0212,0.32721)(4.0841,0.32066)(4.1469,0.3152)(4.2097,0.31077)(4.2726,0.30727)(4.3354,0.30461)(4.3982,0.30268)(4.4611,0.30138)(4.5239,0.30058)(4.5867,0.30017)(4.6496,0.30002)(4.7124,0.3)(4.7752,0.30002)(4.8381,0.30017)(4.9009,0.30058)(4.9637,0.30138)(5.0265,0.30268)(5.0894,0.30461)(5.1522,0.30727)(5.215,0.31077)(5.2779,0.3152)(5.3407,0.32066)(5.4035,0.32721)(5.4664,0.33492)(5.5292,0.34385)(5.592,0.35403)(5.6549,0.3655)(5.7177,0.37827)(5.7805,0.39233)(5.8434,0.40769)(5.9062,0.42431)(5.969,0.44215)(6.0319,0.46117)(6.0947,0.48129)(6.1575,0.50245)(6.2204,0.52455)(6.2832,0.54749)
				};

				\addplot[dashed,
        		color=red,
        		line width = 2pt,
				mark = ,
				mark indices={    1,5,9,13,17,21,25,29,33,37,41,45,49,53,57,61,65,69,73,77,81,85,89,93,97,101},
				mark options={solid},			
				mark size = 3pt,
				]
				coordinates {
                 (0,0.74142)(0.062832,0.75495)(0.12566,0.76882)(0.1885,0.78298)(0.25133,0.79733)(0.31416,0.8118)(0.37699,0.82631)(0.43982,0.84077)(0.50265,0.85508)(0.56549,0.86917)(0.62832,0.88295)(0.69115,0.89632)(0.75398,0.9092)(0.81681,0.92151)(0.87965,0.93317)(0.94248,0.9441)(1.0053,0.95422)(1.0681,0.96347)(1.131,0.97179)(1.1938,0.97912)(1.2566,0.98541)(1.3195,0.99061)(1.3823,0.9947)(1.4451,0.99764)(1.508,0.99941)(1.5708,1)(1.6336,0.99941)(1.6965,0.99764)(1.7593,0.9947)(1.8221,0.99061)(1.885,0.98541)(1.9478,0.97912)(2.0106,0.97179)(2.0735,0.96347)(2.1363,0.95422)(2.1991,0.9441)(2.2619,0.93317)(2.3248,0.92151)(2.3876,0.9092)(2.4504,0.89632)(2.5133,0.88295)(2.5761,0.86917)(2.6389,0.85508)(2.7018,0.84077)(2.7646,0.82631)(2.8274,0.8118)(2.8903,0.79733)(2.9531,0.78298)(3.0159,0.76882)(3.0788,0.75495)(3.1416,0.74142)(3.2044,0.72831)(3.2673,0.71569)(3.3301,0.7036)(3.3929,0.6921)(3.4558,0.68123)(3.5186,0.67103)(3.5814,0.66154)(3.6442,0.65276)(3.7071,0.64472)(3.7699,0.63743)(3.8327,0.63088)(3.8956,0.62506)(3.9584,0.61995)(4.0212,0.61555)(4.0841,0.6118)(4.1469,0.60869)(4.2097,0.60615)(4.2726,0.60415)(4.3354,0.60263)(4.3982,0.60153)(4.4611,0.60079)(4.5239,0.60033)(4.5867,0.6001)(4.6496,0.60001)(4.7124,0.6)(4.7752,0.60001)(4.8381,0.6001)(4.9009,0.60033)(4.9637,0.60079)(5.0265,0.60153)(5.0894,0.60263)(5.1522,0.60415)(5.215,0.60615)(5.2779,0.60869)(5.3407,0.6118)(5.4035,0.61555)(5.4664,0.61995)(5.5292,0.62506)(5.592,0.63088)(5.6549,0.63743)(5.7177,0.64472)(5.7805,0.65276)(5.8434,0.66154)(5.9062,0.67103)(5.969,0.68123)(6.0319,0.6921)(6.0947,0.7036)(6.1575,0.71569)(6.2204,0.72831)(6.2832,0.74142)
				};

        	\addplot[dashdotted,
				color=green,
				line width = 2pt,
				mark = ,
				mark indices={    1,5,9,13,17,21,25,29,33,37,41,45,49,53,57,61,65,69,73,77,81,85,89,93,97,101},
				mark options={solid},			
				mark size = 3pt,
				]
				coordinates {
                 (0,0.93536)(0.062832,0.93874)(0.12566,0.94221)(0.1885,0.94574)(0.25133,0.94933)(0.31416,0.95295)(0.37699,0.95658)(0.43982,0.96019)(0.50265,0.96377)(0.56549,0.96729)(0.62832,0.97074)(0.69115,0.97408)(0.75398,0.9773)(0.81681,0.98038)(0.87965,0.98329)(0.94248,0.98602)(1.0053,0.98855)(1.0681,0.99087)(1.131,0.99295)(1.1938,0.99478)(1.2566,0.99635)(1.3195,0.99765)(1.3823,0.99867)(1.4451,0.99941)(1.508,0.99985)(1.5708,1)(1.6336,0.99985)(1.6965,0.99941)(1.7593,0.99867)(1.8221,0.99765)(1.885,0.99635)(1.9478,0.99478)(2.0106,0.99295)(2.0735,0.99087)(2.1363,0.98855)(2.1991,0.98602)(2.2619,0.98329)(2.3248,0.98038)(2.3876,0.9773)(2.4504,0.97408)(2.5133,0.97074)(2.5761,0.96729)(2.6389,0.96377)(2.7018,0.96019)(2.7646,0.95658)(2.8274,0.95295)(2.8903,0.94933)(2.9531,0.94574)(3.0159,0.94221)(3.0788,0.93874)(3.1416,0.93536)(3.2044,0.93208)(3.2673,0.92892)(3.3301,0.9259)(3.3929,0.92302)(3.4558,0.92031)(3.5186,0.91776)(3.5814,0.91538)(3.6442,0.91319)(3.7071,0.91118)(3.7699,0.90936)(3.8327,0.90772)(3.8956,0.90626)(3.9584,0.90499)(4.0212,0.90389)(4.0841,0.90295)(4.1469,0.90217)(4.2097,0.90154)(4.2726,0.90104)(4.3354,0.90066)(4.3982,0.90038)(4.4611,0.9002)(4.5239,0.90008)(4.5867,0.90002)(4.6496,0.9)(4.7124,0.9)(4.7752,0.9)(4.8381,0.90002)(4.9009,0.90008)(4.9637,0.9002)(5.0265,0.90038)(5.0894,0.90066)(5.1522,0.90104)(5.215,0.90154)(5.2779,0.90217)(5.3407,0.90295)(5.4035,0.90389)(5.4664,0.90499)(5.5292,0.90626)(5.592,0.90772)(5.6549,0.90936)(5.7177,0.91118)(5.7805,0.91319)(5.8434,0.91538)(5.9062,0.91776)(5.969,0.92031)(6.0319,0.92302)(6.0947,0.9259)(6.1575,0.92892)(6.2204,0.93208)(6.2832,0.93536)
				};

				\legend{$\bbetamin = 0.3$, $\bbetamin = 0.6$, $\bbetamin = 0.9$ }
				
			\end{axis}
		\end{tikzpicture}
		\caption{$\beta(\theta)$ in \eqref{eq_beta_model} for three different values of $\bbetamin$, when $\bkappa = 1.5$ and $\bphi  =0$.}
		\label{fig:1b}
	\end{figure}
	
	\subsection{Problem Description}\label{sec_prob_des}
		Given the observations $y_t$ in \eqref{eq_yt} over $T$ transmission instances, our goal  is to derive theoretical performance bounds and develop low-complexity algorithms for estimating the position of the UE $\bbp$ (and the channel gain $\bbalp$ as an unknown nuisance parameter) under three different scenarios:
		\begin{itemize}
		    \item \textit{Scenario-I}: There is a mismatch between the assumption and the reality in this scenario. It is assumed that the amplitudes of the RIS elements are equal to $1$ (which is equivalent to assuming $\bbetamin=1$ or $\overline{\kappa} = 0$); however, the \textit{true model} is as in \eqref{eq_wtm_ris}.
		    \item \textit{Scenario-II}: The  \textit{true model} in \eqref{eq_wtm_ris} is known, but the RIS related parameters, $\bbetamin, \overline{\kappa},$ and $\overline{\phi}$, are assumed to be unknown.
		    \item \textit{Scenario-III}: Both the  \textit{true model} in \eqref{eq_wtm_ris} and the RIS related parameters, $\bbetamin, \overline{\kappa},$ and $\overline{\phi}$, are known. 
		\end{itemize}
		\rev{Note that the amplitude variations occur in the RIS elements through the RIS control model $w_{t,m}$ in \eqref{eq_wtm_ris}, while the near-field effects manifest itself through the steering vector $\ba(\bp)$ in \eqref{eq_ap_nearfield}. Hence, the considered scenarios differ in their knowledge about $w_{t,m}$, but they all consider the same steering vector $\ba(\bp)$.}
        
        \rev{The motivation for considering Scenario~I is to determine the conditions (regarding, e.g., SNR, RIS size, RIS model parameters $\bbetamin, \overline{\kappa},$ and $\overline{\phi}$) under which the conventional unit-amplitude model can be employed for localization without significant performance degradation, when the true model is given by \eqref{eq_wtm_ris}. On the other hand, we consider Scenario~II to investigate how the effects of RIS amplitude variations on accuracy can be counteracted by designing powerful algorithms for localization and RIS model calibration, when the UE is aware of the true model in \eqref{eq_wtm_ris}. Scenario~III is mainly used for benchmarking purposes, i.e., to provide an upper bound on the performance of the algorithms developed under Scenario~II.}
		
		To handle the different scenarios, the remainder of the manuscript is organized as follows. In Section~\ref{sec_mcrb_analysis}, the \ac{MCRB} analysis of near-field localization under Scenario~I is performed while Section~\ref{sec_mml_scenario1} focuses on the estimator design for Scenario~I. Then, the localization algorithms and the theoretical bounds for Scenario~II and Scenario~III are presented in Section~\ref{sec_scen23}. Finally, numerical examples for all the three scenarios are provided in Section~\ref{sec:Nume}.
	
\section{Scenario-I: Misspecified Cram\'{e}r-Rao Bound (MCRB)  Analysis}\label{sec_mcrb_analysis}	
		
In this scenario, we aim to quantify the localization performance loss due to the model mismatch resulting from the phase-dependent amplitude variations specified in \eqref{eq_wtm_ris}. To that end, we will resort to the MCRB analysis \cite{Fortunati2017, Fortunati2018Chapter4,MCRB_TSP_2021,MCRB_delay_ICASSP_2020}. In the following, we first describe the \textit{true model}, which corresponds to the realistic RIS response model in \eqref{eq_wtm_ris}, and the \textit{assumed model}, which is the ideal unit-amplitude RIS model commonly employed in the literature. Then, we provide theoretical background on the \ac{MCRB}, propose a method to find the \textit{pseudo-true parameter}, and derive the \ac{MCRB} and the corresponding \ac{LB}.
	 
\subsection{True and Assumed RIS Amplitude Models} 
\subsubsection{True Model}
The true parameter vector $\bbet$ is given by $\bbet \triangleq [\text{Re} (\overline{\alpha}) \,\text{Im} (\overline{\alpha}) \, \overline{\bp}^{\trpose}]^\trpose $.
	We  define $\mu_t$  as 
	\begin{gather}\label{eq:mut}
	\mu_t \triangleq \overline{\alpha} \sum_{m=1}^{M} [\bb(\bbp)]_m w_{t,m} s_t 
	\end{gather}
	for $t = 1, \ldots, T$. Then, the probability density function (pdf) of the true observation, $p(\by)$,  can be expressed as
	\begin{equation}\label{eq:truedpf}
	p(\by)  = \left(\frac{1}{\pi N_0}\right)^{T} \exp\left(-\frac{1}{N_0} \norm{\by-\bmu}^2\right),
	\end{equation}
	where $\by \triangleq [y_1\, \ldots \,y_T]^{\trpose}$ and $\bmu\triangleq [\mu_1\, \ldots \, \mu_T]^{\trpose}$. 
	
	\subsubsection{Assumed Model}
	
In practice, the knowledge of the exact amplitude model in \eqref{eq_beta_model}, which is hardware dependent, may not be available. In that case, the ideal model of $\beta(\theta_{t,m})=1$  can be used. For this \textit{assumed model}, we represent $w_{t,m}$ as $\tilde{w}_{t,m}$, which is given by 
	\begin{align}\label{eq_wt_assumed}
	\tilde{w}_{t,m} = e^{j\theta_{t,m}}
	\end{align}
	for any $t$ and $m$. 
	Therefore, the misspecified parametric pdf for $\bet \triangleq [\text{Re}(\alpha)\, \text{Im}(\alpha)\, \bp^{\trpose}]^{\trpose}$ under the assumed model, denoted by $\tilde{p}(\by|\bet)$, can be expressed as 
	\begin{equation}\label{eq:assumedpdf}
	\tilde{p}(\by|\bet)  = \left(\frac{1}{\pi N_0}\right)^{T} \exp\left(-\frac{1}{N_0} \norm{\by-\tilde{\bmu}(\bet)}^2\right),
	\end{equation}
	where $\tbmu(\bet) \triangleq [\tmu_1(\bet)\, \ldots \, \tmu_T(\bet)]^{\trpose}$, and 
	\begin{gather}\label{eq:mut_til}
	\tmu_t(\bet) \triangleq  \alpha \sum_{m=1}^{M} [\bb(\bp)]_m \tilde{w}_{t,m} s_t    
	\end{gather}
	for $t = 1, \ldots, T$. It is noted that when $\overline{\beta}_{\text{min}} = 1$ or $\overline{\kappa} = 0$,  $p(\by)$ and $\tilde{p}(\by| \bbet)$ coincide with each other for any $\by$.
	
	\subsection{MCRB Definition}
% 	Let $\bbet = [\text{Re} (\bar{\alpha}) ~ \text{Im} (\bar{\alpha}) ~ \bbp^\trpose]^\trpose $ be the vector of true parameters.
	We introduce the pseudo-true parameter \cite{Fortunati2017}, which minimizes the Kullback-Leibler (KL) divergence between the true pdf in \eqref{eq:truedpf} and the misspecified parametric pdf in \eqref{eq:assumedpdf}; namely, 
	%$\bet_{0}\in\mathbb{R}^{5}$:
	\begin{equation}\label{eq:eta0}
		\bet_0 = \argmin_{\bet\in\mathbb{R}^{5}} D\left(p(\by)\Vert \tilde{p}(\by| \bet)\right),
	\end{equation}
	where $ D\left(p(\by)\Vert \tilde{p}(\by| \bet)\right)$ denotes the KL divergence between the densities $p(\by)$ and $\tilde{p}(\by|\bet)$. %It is noted that if there is no mismatch (i.e., if $\tilde{p}(\by|\bet)=p(\by|\bet)$ for all $\by$ and $\bet$), $\bet_0$ becomes equal to the true parameter vector $\bbet$.

	Next, let  $\hat{\bet}(\by)$ be a misspecified-unbiased (MS-unbiased) estimator of $\bbet$, i.e., the mean of the estimator $\hat{\bet}(\by)$ under the true model is equal to $\bet_0$. The \ac{MCRB} is a lower bound for the covariance matrix of any MS-unbiased estimator of $\bbet$, $\hat{\bet}(\by)$ \cite{Fortunati2017, Fortunati2018Chapter4, Ricmond2015MCRB}:
	\begin{align}
		&\mathbb{E}_p\{(\hat{\bet}(\by)-\bet_0)(\hat{\bet}(\by)-\bet_0)^{\trpose}\} \succeq  \mcrb(\bet_0),
		\label{eq:MCRB}
		%&{}\hspace{2cm}\geq \bA_{\bet_0}^{-1} \bB_{\bet_0}\bA_{\bet_0}^{-1}  \triangleq \mcrb(\bet_0),
	\end{align}
	where $\mathbb{E}_p\{\cdot\}$ denotes the expectation operator under the true model $p(\by)$ and 
	\begin{align} \label{eq_mcrb_def}
		\mcrb(\bet_0) \triangleq \bA_{\bet_0}^{-1} \bB_{\bet_0}\bA_{\bet_0}^{-1},
	\end{align}
	in which the $(i,j)$-th elements of the matrices $\bA_{\bet_0}$ and  $\bB_{\bet_0}$ are calculated as 
	\begin{align}\label{eq:Aeta0}
		[\bA_{\bet_0}]_{ij} &= \mathbb{E}_p\left\{\frac{\partial^2}{\partial \eta_i \partial \eta_j} \log \tilde{p}(\by|\bet)\Big|_{\bet = \bet_0}\right\}, \\\label{eq:Beta0}
		[\bB_{\bet_0}]_{ij} &= \mathbb{E}_p\left\{\frac{\partial}{\partial \eta_i } \log \tilde{p}(\by|\bet) \frac{\partial}{\partial \eta_j } \log \tilde{p}(\by|\bet)\Big|_{\bet =  \bet_0}\right\},
	\end{align} 
	for $1\leq i, j\leq 5$, with $\eta_i$ denoting the $i$th element of $\bet$.  
	
	Since the value of the pseudo-true parameter is generally not of interest, the \ac{MCRB} is used to establish the \ac{LB} of any MS-unbiased estimator with respect to the true parameter value \cite{Fortunati2017}
	\begin{align}
		\mathbb{E}_p\{(\hat{\bet}(\by)-\bbet)(\hat{\bet}(\by)-\bbet)^{\trpose}\} \succeq \lb(\bet_0),\label{eq:LB}
	\end{align}
	where 
	\begin{align}\label{eq_lb}
	    \lb(\bet_0)\triangleq \mcrb(\bet_0)  + (\bbet-\bet_0)(\bbet-\bet_0)^{\trpose}~.
	\end{align}
	%a bound for the mean-squared error of any MS-unbiased estimator $\hat{\bet}(\by)$ as follows \cite{Fortunati2017}:
	%\begin{equation}\label{eq:LB}
	%	\mathbb{E}_p\{(\hat{\bet}(\by)-\bbet)(\hat{\bet}(\by)-\bbet)^{\trpose}\} \geq \mcrb(\bet_0)  + (\bbet-\bet_0)(\bbet-\bet_0)^{\trpose}.
	%\end{equation} 
	%We denote the expression at the right-hand-side (RHS) of the inequality in \eqref{eq:LB} as $\lb(\bet_0)$, i.e., $\lb(\bet_0)\triangleq \mcrb(\bet_0)  + (\bbet-\bet_0)(\bbet-\bet_0)^{\trpose}$. 
	The last term \rev{in \eqref{eq_lb}} is a bias term; that is, $\bias(\bet_0) \triangleq (\bbet-\bet_0)(\bbet-\bet_0)^{\trpose}$, and it is independent of the SNR. Hence, as the SNR tends to infinity, the MCRB term goes to zero, and the bias term becomes a tight bound for the MSE of any MS-unbiased estimator.

	\subsection{MCRB Derivation for RIS-aided Localization}\label{sec:MCRBder}

	\subsubsection{Determining the Pseudo-True Parameter}\label{sec_pseudo}
	To derive the MCRB for estimating the UE position under mismatch between the amplitude models for the RIS elements, we should first calculate the $\bet_0$ parameter in \eqref{eq:eta0} for the system model described in Section~\ref{sec:System}; that is, we should find the value of $\bet$ that minimizes the KL divergence between $p(\by)$ in \eqref{eq:truedpf} and $\tilde{p}(\by|\bet)$ in \eqref{eq:assumedpdf}. The following lemma characterizes $\bet_0$ for the considered system model.
	\begin{lemma}~\label{lemma_pseudo}
	The value of $\bet$ that minimizes the KL divergence between $p(\by)$ in \eqref{eq:truedpf} and $\tilde{p}(\by|\bet)$ in \eqref{eq:assumedpdf}
	%, which is a parameterized version of \eqref{eq:assumedpdf}, 
	can be expressed as
	\begin{equation}
		\bet_0 = \argmin_{\bet\in\mathbb{R}^5} \norm{\bep(\bet)} \label{eq_betz}
	\end{equation}
	where  $\bep(\bet) \triangleq [\epsilon_1(\bet)\, \ldots \, \epsilon_T(\bet)]^{\trpose}$ and $\epsilon_t(\bet) \triangleq \mu_t-\tmu_t(\bet)$ for $t=1,\ldots,T$. 
	\end{lemma}
	\begin{proof}
		See Appendix~\ref{sec:AppA}.
	\end{proof}
	\rev{
	\begin{rem}[Applicability of Lemma~\ref{lemma_pseudo} to Different RIS Models]\label{rem_lemma1}
	The proof of Lemma~\ref{lemma_pseudo} does not exploit the specific functional form of the adopted (true) RIS model in \eqref{eq_beta_model} and thus can be applied to different RIS models. Hence, Lemma~\ref{lemma_pseudo} can be used to compute the MCRB and LB in \eqref{eq_lb} for any choice of the RIS element model $\beta(\theta_{t,m})$ in \eqref{eq_beta_model} by simply plugging the corresponding $w_{t,m}$ in \eqref{eq_wtm_ris} into the true model expressions in \eqref{eq:mut} and \eqref{eq:truedpf}.
	%That is, if RIS amplitudes do not obey the formula given in equations \eqref{eq_wtm_ris},\eqref{eq_beta_model}, we can still use Lemma~\ref{lemma_pseudo} for computing the MCRB. Our analyses are generic and can be easily modified to handle the different RIS models.
	\end{rem}
	}
    Lemma~\ref{lemma_pseudo} states that the pseudo-true parameter minimizes the Euclidean distance between the noise-free observations under the true and assumed models.
	
	Let $\gamma(\bet)\triangleq \norm{\bep(\bet)} = \norm{\bmu-\tbmu(\bet)}$. It is noted from \eqref{eq:mut} and \eqref{eq:mut_til} that $\gamma(\bet)$ is non-convex with respect to $\bet$; hence, it is challenging to solve \eqref{eq_betz} in its current form. Based on \eqref{eq:assumedpdf} and \eqref{eq:mut_til}, we can re-write \eqref{eq_betz} as
	\begin{align}\label{eq_etabar2}
		(\balp_0, \bp_0) = \argmin_{(\balp, \bp)} \norm{\bmu- \alpha \, \bcc(\bp) } ,
	\end{align}
	where $[\bcc(\bp)]_t \triangleq \sum_{m=1}^{M} [\bb(\bp)]_m \tilde{w}_{t,m} s_t$.
	%\begin{align}
	%	\balp = [\text{Re}(\alpha)~ \text{Im}(\alpha)]^{\trpose}, ~~
	%	[\bcc(\bp)]_t \triangleq \sum_{m=1}^{M} [\bb(\bp)]_m \tilde{w}_{t,m} s_t \,.
	%\end{align}
	%For the minimization problem in \eqref{eq_etabar2}, 
	The optimal complex-valued $\alpha$ for any given $\bp$ can be expressed 
	%as a function of the position $\bp$ 
	as 
	\begin{gather}\label{eq:alpOpt}
		%\alpha = \frac{ \bcc^{\mathsf{H}}(\bp) \bmu }{ \bcc^{\mathsf{H}}(\bp) \bcc(\bp)  } \,\cdot
		\alpha = \rev{\bcc(\bp)^{\dagger} \bmu} ~,
	\end{gather}
	\rev{where $\bX^{\dagger} = (\bX^\hermit \bX)^{-1} \bX^\hermit$ denotes pseudo-inverse.}
	Inserting \eqref{eq:alpOpt} into \eqref{eq_etabar2}, the problem can be reduced to a 3-D search as follows:
	\begin{gather}\label{eq_etabar3}
		%\bp_0 = \argmax_{\bp} \rev{ \frac{\abs{ \bmu^{\mathsf{H}}\bcc(\bp)}^2}{\norm{\bcc(\bp)}^2} } ~.
		%\bmu^{\mathsf{H}} \projrange{\bcc(\bp)} \bmu~,
		\bp_0 = \rev{\argmin_{\bp} \norm{\projnull{\bcc(\bp)} \bmu }~,}
	\end{gather}
	\rev{where $\projnull{\bX} = \Imatrix - \projrange{\bX}$ and $\projrange{\bX} = \bX \bX^{\dagger}$ denotes the orthogonal projection matrix onto the column space of $\bX$.}
	%where $\projrange{\boldsymbol{\rmx}} \triangleq {\boldsymbol{\rmx} \boldsymbol{\rmx}^{\mathsf{H}}}/{\boldsymbol{\rmx}^{\mathsf{H}} \boldsymbol{\rmx}}$. 
	Therefore, $\bet_0=[\balp_0^{\trpose}~\bp_0^{\trpose}]^{\trpose}$ can be found by first performing a 3-D optimization as in \eqref{eq_etabar3}, and then calculating $\alpha_0$ via \eqref{eq:alpOpt} and obtaining $\balp_0$ as $\balp_0 = [\text{Re}(\alpha_0)~ \text{Im}(\alpha_0)]^{\trpose}$.
	
	\begin{rem}[\rev{Initialization of \eqref{eq_etabar3}}]
	In order to determine the pseudo-true parameter or equivalently to find $\bp_0$ in \eqref{eq_etabar3}, the true location $\bbp$ can be used for initialization, which reduces the computational complexity of MCRB calculation significantly.
	\end{rem}
	
	\rev{\begin{rem}[Special Cases of MCRB] \label{rem:MCRB}
	As $\bbetamin$ approaches $1$, or $\bkappa$ approaches $0$, the RIS amplitude $\beta(\theta_{t,m})$ in \eqref{eq_beta_model} converges to $1$, implying that the phase-dependent amplitude variations model (true model) in \eqref{eq_wtm_ris} converges to the standard unit-amplitude model (assumed model) in \eqref{eq_wt_assumed}. In that case, the mismatch between the reality and the assumed model vanishes, which suggests that the pseudo-true parameter $\bet_0$ in \eqref{eq_betz} approaches the true parameter $\bbet$. As a result, the bias term in \eqref{eq_lb} disappears and the \ac{MCRB} in \eqref{eq_mcrb_def} reverts to the classical \ac{CRB} expression since $\bA_{\bet_0} = \bA_{\bbet} = -\bB_{\bet_0} = -\bB_{\bbet}$ \cite{Fortunati2018Chapter4}.
% 	As $\bbetamin$ approaches $1$ or $\bkappa$ approaches $0$, the model mismatch between the assumed and the true model vanishes.
% 	Note that without model mismatch, 
% 	$\bA_{\bet_0} = \bA_{\bbet} = -\bB_{\bet_0} = -\bB_{\bbet}$ so that the \ac{MCRB} reverts to the classical Cram\'{e}r-Rao bound (CRB) \cite{Fortunati2018Chapter4}. 
	\end{rem}}

	\subsubsection{Deriving the MCRB}
	After finding $\bet_0$, we compute the matrices  $\bA_{\bet_0}$ from \eqref{eq:Aeta0} and  $\bB_{\bet_0}$ from \eqref{eq:Beta0} for evaluating the MCRB in \eqref{eq:MCRB}. 	Based on the pdf expressions in \eqref{eq:truedpf}--\eqref{eq:assumedpdf}, \eqref{eq:Aeta0} becomes
    \begin{align}
		&[\bA_{{\bet_0}}]_{ij} =  -\frac{1}{N_0}\left( \int\frac{\partial^2}{\partial \eta_i \partial \eta_j} \norm{\by-\tbmu(\bet)}^2 p(\by)\,d\by\right)\Bigg|_{\bet = \bet_0} \\
		& = -\frac{1}{N_0} \left(\sum_{t=1}^{T}  \int \frac{\partial^2}{\partial \eta_i \partial \eta_j} \abs{y_t-\tmu_t(\bet)}^2 p(\by) \, d\by \right)\Bigg|_{\bet = \bet_0} \\
		& =  -\frac{1}{N_0} \left(\sum_{t=1}^{T}  \int \frac{\partial^2}{\partial \eta_i \partial \eta_j} \abs{y_t-\tmu_t(\bet)}^2 p(y_t) \, dy_t \right)\Bigg|_{\bet = \bet_0} \\
		& = \frac{2}{N_0} \text{Re}\left\{\sum_{t=1}^{T} \epsilon_t(\bet)^{*}\frac{\partial^2 \tmu_t(\bet)}{\partial \eta_i \eta_j} - \frac{\partial \tmu_t^{*}(\bet)}{\partial \eta_i} \frac{\partial \tmu_t(\bet)}{\partial \eta_j}\right\}\Bigg|_{\bet = \bet_0} ~.
	\end{align}

    In addition, after some algebraic manipulation, the $(i,j)$th entry of matrix $\bB_{\bet_0}$ in \eqref{eq:Beta0} can be written as the sum of four terms as $[\bB_{\bet_0}]_{ij} = T_1 + T_2 + T_3 + T_4$, where
	\begin{align*}
		T_1  = \frac{1}{N_0^2} \left(\sum_{t=1}^{T}\frac{\partial \tmu_t(\bet)}{\partial \eta_i}\epsilon_t(\bet)^{*}\right)\left(\sum_{l=1}^{T} \frac{\partial \tmu_l(\bet)}{\partial \eta_j}\epsilon_l(\bet)^{*} \right)\Bigg|_{\bet = \bet_0}
	\end{align*}
	\begin{align*}
		&T_2  = \Bigg[\frac{1}{N_0^2} \left(\sum_{t=1}^{T}\frac{\partial \tmu_t(\bet)}{\partial \eta_i}\epsilon_t(\bet)^{*}\right) \\
		&\times \left(\sum_{l=1}^{T} \frac{\partial \tmu_l^{*}(\bet)}{\partial \eta_j}\epsilon_l(\bet) \right)  + \frac{1}{N_0} \sum_{t=1}^{T} \frac{\partial \tmu_t(\bet)}{\partial \eta_i}\frac{\partial \tmu_t^{*}(\bet)}{\partial \eta_j}\Bigg] \Bigg|_{\bet = \bet_0}
	\end{align*}
	\begin{align*}	T_3  = \frac{1}{N_0^2} \left(\sum_{t=1}^{T}\frac{\partial \tmu_t^{*}(\bet)}{\partial \eta_i}\epsilon_t(\bet)\right) \left(\sum_{t=1}^{T}\frac{\partial \tmu_l^{*}(\bet)}{\partial \eta_j}\epsilon_l(\bet)\right)\Bigg|_{\bet = \bet_0}
	\end{align*}
	\begin{align*}
		&T_4 =\Bigg[\frac{1}{N_0^2} \left(\sum_{t=1}^{T} \frac{\partial \tmu_t^{*}(\bet)}{\partial \eta_i} \epsilon_t(\bet) \right)\\
		&\times \left(\sum_{l=1}^{T} \frac{\partial \tmu_l(\bet)}{\partial \eta_j} \epsilon_l(\bet)^{*} \right) + \frac{1}{N_0} \sum_{t=1}^{T} \frac{\partial \tmu_t^{*}(\bet)}{\partial \eta_i}\frac{\partial \tmu_t(\bet)}{\partial \eta_j}\Bigg]\Bigg|_{\bet = \bet_0} ~.
	\end{align*}
   Hence, $[\bA_{{\bet_0}}]_{ij}$ and  $[\bB_{\bet_0}]_{ij}$ can be written in more compact forms as follows:
    \begin{align}\label{eq:Aeta0_1}
		&[\bA_{\bet_0}]_{ij}   = \frac{2}{N_0} \text{Re}\left\{\bep(\bet)^{\mathsf{H}} \frac{\partial^2\tbmu(\bet)}{\partial\eta_i \partial\eta_j}-\Big(\frac{\partial\tbmu(\bet)}{\partial\eta_i}\Big)^{\mathsf{H}}\frac{\partial\tbmu(\bet)}{\partial\eta_j}\right\}\Bigg|_{\bet = \bet_0}
		\end{align}
		\begin{align}
		&[\bB_{\bet_0}]_{ij} = \frac{2}{N_0}\Bigg[ \frac{2}{N_0}\text{Re}\left\{\bep(\bet)^{\mathsf{H}}\frac{\partial\tbmu(\bet)}{\partial\eta_i}\right\} \text{Re}\left\{\bep(\bet)^{\mathsf{H}}\frac{\partial\tbmu(\bet)}{\partial\eta_j}\right\}  \notag \\
		&+ \text{Re} \left\{\left(\frac{\partial\tbmu(\bet)}{\partial\eta_i}\right)^{\mathsf{H}}\frac{\partial\tbmu(\bet)}{\partial\eta_j}\right\}\Bigg]\Bigg|_{\bet = \bet_0} \label{eq:Beta0_1} ~,
	\end{align}
	where
	\begin{equation}
	    \frac{\partial^2\tbmu(\bet)}{\partial\eta_i \partial\eta_j}\triangleq \left[\frac{\partial^2 \tmu_1(\bet)}{\partial \eta_i \partial \eta_j} \, \ldots \, \frac{\partial^2 \tmu_T(\bet)}{\partial \eta_i \partial \eta_j}\right]^{\trpose}.
	\end{equation}
	Therefore, once we compute the first and the second derivatives of $\tmu_t(\bet)$ with respect to $\bet$, we can easily compute the matrices $\bA_{\bet_0}$ and $\bB_{\bet_0}$ as specified above. The derivatives are presented in Appendix \ref{sec:AppB}. Based on $\bA_{\bet_0}$ and $\bB_{\bet_0}$, the MCRB in \eqref{eq:MCRB} and the lower bound in \eqref{eq:LB} can be evaluated in a straightforward manner. %Overall, we provide a generic lower bound for the considered system model, which takes into account the imperfect knowledge of the amplitude models for the RIS elements.

	\subsection{\rev{Localization Accuracy vs. Number of Transmissions}}\label{subsec:MCRBvsT}
	
	\rev{In order to illustrate the relationship between the localization accuracy and the number of transmissions, we analyze how the trace of the MCRB or the LB changes with respect to the number of transmissions, $T$.  We consider two sets of observations for the same values of the true parameter vector $\bbet$, the RIS size $M$, and the RIS related parameters $\betamin, \kappa, \phi$. }

\rev{We prove that if repetitive RIS phase profiles are used, the MCRB and LB do not decay with $1/T$ whereas the CRB does. To be more specific, we state the following lemma.}

\begin{lemma}\label{lemma:MCRB}
\rev{
For the $k$-th set of the observations, let the number of observations, the phase profile and the corresponding pseudo-true parameter be given by $T^{(k)}$, $\btheta^{(k)} =\{\theta^{(k)}_{t,m}\}_{\{t=1,m=1\}}^{\{T^{(k)},M\}}$, and $\bet_0^{(k)}$, respectively,  where $k\in\{1,2\}$.}
\rev{
\begin{itemize}
    \item We assume that $T^{(2)} = K T^{(1)}$, where $K\in\mathbb{N}$. 
    \item Moreover, for the phase profiles, for any $1\leq m\leq M$, assume that
\begin{equation}\label{eq:phaseprofile_r}
    \theta^{(2)}_{t,m} = \theta^{(1)}_{ ( t\text{ mod } T^{(1)}), m}
\end{equation}
where $\big( t\text{ mod } T^{(1)}\big)$ denotes the remainder of the division of $t$ by $T^{(1)}$.
\end{itemize}
Let $\mcrb\big(T^{(k)},\bet_0^{(k)}\big)$, $\lb\big(T^{(k)},\bet_0^{(k)}\big)$ and $\crb\big(T^{(k)}\big)$  denote the corresponding MCRB, LB and CRB for the $k$-th set of observations. Then, 
\begin{align*}
    &\frac{1}{K} < \frac{\Tr\big\{\mcrb\big(T^{(2)},\bet_0^{(2)}\big)\big\}}{\Tr\big\{\mcrb\big(T^{(1)},\bet_0^{(1)}\big)\big\}} < 1, \\
    &~ \frac{1}{K} < \frac{\Tr\big\{\lb\big(T^{(2)},\bet_0^{(2)}\big)\big\}}{\Tr\big\{\lb\big(T^{(1)},\bet_0^{(1)}\big)\big\}} < 1~\text{and}~\frac{\Tr\big\{\crb\big(T^{(2)}\big)\big\}}{\Tr\big\{\crb\big(T^{(1)}\big)\big\}} = \frac{1}{K}.
\end{align*}}
\end{lemma}
\begin{proof}
\rev{\ See Appendix~\ref{sec:AppMCRB}.}
\end{proof}
\rev{Lemma~\ref{lemma:MCRB} reveals that the trace of the CRB decays with $1/T$ while that of the MCRB or LB decays slower than $1/T$ due to model misspecification.}

	\section{Scenario-I: Mismatched Estimator}\label{sec_mml_scenario1}
	
	In this section, we focus on estimator design for Scenario-I. First, we derive the plain \ac{MML} estimator, which entails computationally prohibitive high-dimensional non-convex optimization. To circumvent this, we then propose a low-complexity estimator capitalizing on the Jacobi-Anger expansion, which reduces the problem to a series of line searches over range, azimuth, and elevation domains.
	
	\label{Sec:CaseIMML}
	%Following the classical definition of the maximum likelihood (ML) estimator, 
	\subsection{Mismatched Maximum Likelihood (MML) Estimator}
	
	\rev{Given the pdf of the assumed model in \eqref{eq:assumedpdf},} the \ac{MML} estimator is given by \cite{Fortunati2017}
	\begin{equation}\label{eq:MML1}
		\hat{\bet}_{\text{MML}}(\by) = \argmax_{\bet\in\mathbb{R}^5} \log \tilde{p}(\by|\bet). 
	\end{equation}
	Under some regularity conditions, it can be shown that $\hat{\bet}_{\text{MML}}(\by)$ is asymptotically MS-unbiased and its error covariance matrix is asymptotically equal to the $\mcrb(\bet_0)$ \cite[Thm. 2]{Fortunati2017}. Hence, the covariance matrix of the MML estimator is asymptotically tight to the MCRB.  
	
	From \eqref{eq:assumedpdf} and \eqref{eq:MML1}, the MML estimator based on the received signal $\by$ in \eqref{eq_yt} can be expressed as
	\begin{align} \label{eq_mml_est}
		\hat{\bet}_{\text{MML}}(\by) &= \argmax_{\bet\in\mathbb{R}^5} \log \tilde{p}(\by\lvert\bet) = \argmin_{\bet\in\mathbb{R}^5} \norm{\by-\tbmu(\bet)}.
	\end{align}
    Since this problem is in the same form as in \eqref{eq_betz}, it can be reduced to a 3-D optimization problem as discussed in Section~\ref{sec:MCRBder}. In order to solve the resulting problem, initialization can be very critical as we are faced with a non-convex optimization problem. During the estimation process, we do not have access to the true position $\bbp$. Hence, we cannot use the true position vector $\bbp$ for the initialization. If an arbitrarily chosen position vector is used for the initialization, the global optimal solution of \eqref{eq:MML1} cannot always be obtained. To find a remedy for this issue, we next propose an approximated version of the MML estimator \rev{in \eqref{eq_mml_est}}, namely, the approximate mismatched maximum likelihood (AMML) estimator. \rev{Additionally, to reduce the complexity of the proposed algorithm, we propose a} Jacobi-Anger expansion \rev{based} approach \rev{when the number of transmissions is sufficiently large} \cite{Shaban2021,cuneyd_ICC_RIS_2022},\rev{\cite{JacobiAnger_MIMO_2019}}.

	\subsection{\rev{AMML Estimator}}\label{sec_amml}
% 	 \rev{In this part, we develop a multi-step approach to estimate the UE position in \eqref{eq_mml_est} as follows:
%     \begin{enumerate}
%         \item To estimate azimuth and elevation separately from the distance, we approximate the near-field field RIS steering vector in \eqref{eq_ap_nearfield} as its far-field counterpart, which depends only on azimuth and elevation \cite[Eq.~(11)]{LOS_NLOS_NearField_2021}.
%         \item We estimate azimuth and elevation via two dimensional search by using the far-field approximation.
%         \item We switch back to the near-field model in \eqref{eq_ap_nearfield}. Starting from the azimuth and elevation estimates obtained from the far-field approximation, we alternate between the estimates of the azimuth, elevation and the estimate of the distance.
%     \end{enumerate}
%     These steps will now be described in detail.}
    \rev{%In this part, we develop a multi-step approach to estimate the UE position in \eqref{eq_mml_est}. 
    Let us first re-write $\tbmu(\bet)$ in \eqref{eq:mut_til} using \eqref{eq_bp} as
    \begin{align} \nonumber
    \tbmu(\bet) &= \alpha \underbrace{\begin{bmatrix} \left(\bwtilde_1 \odot \ba(\bpbs) \right)^\trpose
    	\\ \vdots
    	\\ \left(\bwtilde_T \odot \ba(\bpbs) \right)^\trpose
    	\end{bmatrix}}_{\triangleq \tbQ \in \complexset{T}{M} } \ba(\bp) \sqrt{E_s}  ~, \\ \label{eq_y_orig} 
    	& = \alpha\tbQ \ba(\bp) \sqrt{E_s}   ~,
    \end{align}
    where $\bwtilde_t \triangleq [\tilde{w}_{t,1} \, \ldots \, \tilde{w}_{t,M}]^\trpose$ and $s_t = \sqrt{E_s}$ for any $t$. Inserting \eqref{eq_y_orig} into \eqref{eq_mml_est}, the MML estimator becomes
    \begin{align} \label{eq_orig_optimization}
         (\widehat{\alpha},\widehat{d},\widehat{\vartheta},\widehat{\varphi}) = \argmin_{\alpha, d, \vartheta, \varphi} \norm{\by- \alpha \, \tbQ \ba\left(\pp(d,\vartheta,\varphi)\right) \sqrt{E_s} } ~,
    \end{align}
    where
    \begin{align} \label{eq_pos_sph}
        \bp(d, \vartheta, \varphi) \triangleq \bpris + d \,[\sin \vartheta \cos \varphi  ~ \sin \vartheta \sin \varphi ~ \cos \vartheta  ]^{\trpose}
    \end{align}
    denotes the position vector with respect to $\bpris$ in spherical
    coordinates according to the definitions after \eqref{eq_k_wavevector}. Since $\widehat{\alpha}$ can be obtained in closed-form as a function of the other unknowns as in \eqref{eq:alpOpt}, \eqref{eq_orig_optimization} can be solved via 3-D optimization. In the following, we describe the steps to solve \eqref{eq_orig_optimization} in an efficient manner.}
    
	 \subsubsection{\rev{Approximation of Near-Field as Far-Field for Initial Azimuth and Elevation Estimation}}\label{sec_step1}
	 \rev{To avoid 3-D search in \eqref{eq_orig_optimization}, we approximate the near-field steering vector as its far-field counterpart via \eqref{eq:CM1}:
	 \begin{equation}\label{eq_farfield_approx}
	     [\ba(\bp)]_m \rev{\approx [\ba(\vartheta, \varphi)]_m} ~,
	\end{equation}
	which reduces \eqref{eq_orig_optimization} to
\begin{align} \label{eq_2D_varthetaphi}
         (\widehat{\alpha},\widehat{\vartheta},\widehat{\varphi}) = \argmin_{\alpha, \vartheta, \varphi} \norm{\by- \alpha \, \tbQ \ba(\vartheta,\varphi) \sqrt{E_s} } ~.
    \end{align}
    Similar to \eqref{eq:alpOpt}, $\alphahat$ can be estimated in closed-form as a function of $\vartheta$ and $\varphi$ in \eqref{eq_2D_varthetaphi}, leading to
    \begin{align}\label{eq_2D_varthetaphi2}
		(\widehat{\vartheta},\widehat{\varphi}) = \argmin_{\substack{\vartheta \in [0, \, \pi/2] \\ \varphi \in [0, \, 2\pi]}}  \norm{\projnull{\tbQ \ba(\vartheta,\varphi) } \by } ~.
	\end{align}
    In \eqref{eq_2D_varthetaphi2}, initial azimuth and elevation estimates can be found by conducting 2-D search over $\vartheta$ and $\varphi$.}
     
    \subsubsection{\rev{Switching Back to Near-Field for Location Estimation}}
    
    \rev{Starting from the initial estimates $\widehat{\vartheta}$ and $\widehat{\varphi}$ in the approximated MML estimator \eqref{eq_2D_varthetaphi2}, location estimation can be performed by switching back to the original MML estimator \eqref{eq_orig_optimization} based on the near-field model via  the following alternating iterations. 
    \begin{itemize}
        \item \textit{Update $d$:} Given $\widehat{\vartheta}$ and $\widehat{\varphi}$, we estimate distance $d$ in \eqref{eq_orig_optimization} as
         \begin{align} \label{eq_1D_distance}
         (\widehat{\alpha},\widehat{d}) = \argmin_{\alpha, d} \norm{\by- \alpha \, \tbQ \ba\big(\bp(d,\widehat{\vartheta},\widehat{\varphi})\big) \sqrt{E_s} } ~.
    \end{align}
    Similar to the solution of \eqref{eq_2D_varthetaphi}, distance can be estimated from \eqref{eq_1D_distance} via
    \begin{equation} \label{eq_1D_distance2}
	    \widehat{d} = \argmin_{d \in (0,\dmax)} \norm{\projnull{\tbQ \ba(\bp(d, \widehat{\vartheta}, \widehat{\varphi})) } \by } ~,
	\end{equation}
    which entails a simple line search for a given maximum distance $\dmax$ imposed by near-field conditions in \eqref{eq_dmax}.
    \item \textit{Update $\vartheta$ and  $\varphi$:} 
    Given $\widehat{d}$ from \eqref{eq_1D_distance2}, azimuth and elevation estimates can be updated in \eqref{eq_orig_optimization} via 
       \begin{align} \label{eq_2D_varthetaphi_update}
         (\widehat{\alpha},\widehat{\vartheta},\widehat{\varphi}) = \argmin_{\alpha, \vartheta, \varphi} \norm{\by- \alpha \, \tbQ \ba\big(\bp(\widehat{d}, \vartheta, \varphi)\big) \sqrt{E_s} } ~,
    \end{align}
    leading to
    \begin{align}\label{eq_2D_varthetaphi_update2}
		(\widehat{\vartheta},\widehat{\varphi}) = \argmin_{\substack{\vartheta \in [0, \, \pi/2] \\ \varphi \in [0, \, 2\pi]}}  \norm{\projnull{\tbQ\ba(\bp(\widehat{d}, \vartheta, \varphi)) } \by } ~.
	\end{align}
    \end{itemize}
    After a predetermined number of iterations $\Imax$, the resulting estimates of elevation, azimuth and distance from \eqref{eq_1D_distance2} and \eqref{eq_2D_varthetaphi_update2} are inserted into $\bp(\widehat{d}, \widehat{\vartheta}, \widehat{\varphi})$ in \eqref{eq_pos_sph} to find the AMML estimate for the UE location.}
    
    \rev{To further reduce the complexity, we next propose a Jacobi-Anger expansion based estimator that avoids 2-D searches in \eqref{eq_2D_varthetaphi2} and \eqref{eq_2D_varthetaphi_update2} for large number of transmissions $T$.}
    
    % \rev{As the estimator involves two dimensional searches over the azimuth and elevation angles, we next propose the following Jacobi-Anger approximation based low-complexity estimator for large values of the number of transmissions, $T$}.
    
    \subsection{\rev{Jacobi-Anger Expansion Based AMML Estimator for Large Number of Transmissions}}\label{sec_jacobi_amml}
	\rev{In this part, we develop a low-complexity multi-step approach to estimate the UE position in the MML estimator \eqref{eq_orig_optimization} by leveraging the Jacobi-Anger expansion \cite{Cuyt_2008_HandbookContFrac}, as detailed in the following.}
% 	   \begin{enumerate}
%         \item To estimate azimuth and elevation separately from the distance, we will use the steps described in Sec.~\ref{sec_step1}.
%         \item To decouple the estimation of azimuth and elevation in far-field, we approximate the far-field steering vector using finite number of terms in the Jacobi-Anger expansion.
%         \item Employing an unstructured \ac{ML} approach \cite{ML_array_EXIP_98}, we estimate azimuth and elevation via two separate line searches.
%         \item To estimate the distance, we switch back to the near-field model in \eqref{eq_ap_nearfield} and plug in the azimuth and elevation estimates obtained from the far-field Jacobi-Anger approximation.
%     \end{enumerate}
%     These steps will now be described in detail.}
% 	\rev{After approximating the near-field as far-field as in Sec.~\ref{sec_step1}, we use Jacobi-Anger approximation to decouple azimuth and elevation angles.}

	\subsubsection{\rev{Jacobi-Anger Approximation to Decouple Azimuth and Elevation Parameters in Far-Field}}\label{sec_jacobi_step2}
	\rev{Similar to Sec.~\ref{sec_amml}, we begin by invoking the far-field approximation in \eqref{eq:CM1} and considering the resulting approximated estimator in \eqref{eq_2D_varthetaphi}. Using \eqref{eq_dist_approx}, the far-field steering vector in \eqref{eq:CM1} can be expressed as
	\begin{align} \label{eq_ap_nearfield_Jacobi}
	    [\ba(\vartheta, \varphi)]_m = \exp\left( j \frac{2 \pi}{\lambda} q_m \sin(\vartheta) \cos(\varphi - \psi_m) \right) ~,
	\end{align}
	where $q_m$ and $\psi_m$ are defined in \eqref{eq_pm_pris}.} By employing the Jacobi-Anger expansion, \rev{$[\ba(\vartheta, \varphi)]_m$ in \eqref{eq_ap_nearfield_Jacobi}} can be \rev{expressed} as \rev{\cite[Eq.~(17.1.7)]{Cuyt_2008_HandbookContFrac}}
	\begin{equation} \label{eq_jacobi_expansion}
	    \rev{[\ba(\vartheta, \varphi)]_m} = \sum_{n=-\infty}^{\infty} j^{n} J_n\left(\frac{2\pi}{\lambda} \rev{q_m} \sin \vartheta \right) e^{jn(\varphi-\psi_m)} ~,
	\end{equation}
	where $J_n(\cdot)$ is the $n$th order Bessel function of the first kind. As $\abs{J_n(\cdot)}$ decays to $0$ as $\abs{n}$ increases, by neglecting the terms with $\abs{n} > N$, for some $N\in\mathbb{N}$, \rev{\eqref{eq_jacobi_expansion}} can be approximated by\rev{\cite{JacobiAnger_MIMO_2019}}
	\begin{equation}\label{eq:JacobiAnger}
	    \rev{[\ba(\vartheta, \varphi)]_m} \approx \sum_{n=-N}^{N} j^{n} J_n\left(\frac{2\pi}{\lambda} \rev{q_m} \sin \vartheta \right) e^{jn(\varphi-\psi_m)} ~,
	\end{equation}
    \rev{where 
    \begin{align}\label{eq_Nmin}
        N > \frac{2\pi}{\lambda} \qmax \sin \vartheta    
    \end{align}
    enables sufficient precision for high-quality approximation \cite{JacobiAnger_MIMO_2019}, with $\qmax$ denoting the maximum distance of any RIS element from the RIS centre. To provide a compact expression for \eqref{eq:JacobiAnger},} we define \rev{$\bg_m(\vartheta) \in \complexset{(2N+1)}{1}$} and \rev{$\bh(\varphi) \in \complexset{(2N+1)}{1}$} as
    \begin{align} \label{eq_gm}
    	[\bg_m(\vartheta)]_n &= j^n J_n\left( \frac{2\pi}{\lambda} \rev{q_m} \sin \vartheta\right)e^{-jn\psi_m} ~, \\ \label{eq_hphi}
    	[\bh(\varphi)]_n & = e^{jn\varphi} 
    \end{align}
    for $n \in \{-N, \ldots, N\}$. Then, inserting \eqref{eq_gm} and \eqref{eq_hphi} into \eqref{eq:JacobiAnger} yields
    \begin{equation}\label{eq_jacobi}
    	\rev{\ba(\vartheta, \varphi)} \approx \bG^{\trpose}(\vartheta)\bh(\varphi),
    \end{equation}
    where
    \begin{equation}
    \bG(\vartheta) = [\bg_1(\vartheta) \, \ldots \, \bg_M(\vartheta)] \in \complexset{(2N+1)}{M}.
    \end{equation}
    \rev{As seen from \eqref{eq_jacobi}, the proposed Jacobi-Anger expansion of the far-field RIS steering vector in \eqref{eq_ap_nearfield_Jacobi} enables decoupling azimuth and elevation parameters in a matrix-vector multiplication form.} 
    %according to the assumed model in \eqref{eq:mut_til}:
    % \begin{align} \label{eq_y_orig_near}
    % 	\by &= \oalpha \tbQ \ba(\bbp) \sqrt{E_s} + \bn ~,
    % 	\\ \label{eq_jacobi_1st_approx}
    % 	&\approx \oalpha \, \tbQ \ba(\ovartheta, \ovarphi) \sqrt{E_s} + \bn ~,
    % 	\\ \label{eq_jacobi_2nd_approx}
    % 	&\approx \oalpha \, \tbQ \bG^{\trpose}(\ovartheta)\bh(\ovarphi) \sqrt{E_s} + \bn ~,
    % \end{align}
    % where $\bwtilde_t \triangleq [\tilde{w}_{t,1} \, \ldots \, \tilde{w}_{t,M}]^\trpose$, $s_t = \sqrt{E_s}$ for any $t$, $\bn \triangleq [n_1 \, \ldots \, n_T]^{\trpose}$, and the approximations in \eqref{eq_jacobi_1st_approx} and \eqref{eq_jacobi_2nd_approx} result from \eqref{eq:CM1} and \eqref{eq_jacobi}, respectively.}

    \subsubsection{\rev{Azimuth and Elevation Estimation via Unstructured \ac{ML}}}\label{sec_jacobi_step3}
    \rev{We can now exploit the structure in \eqref{eq_jacobi} to recast the MML estimator in \eqref{eq_2D_varthetaphi} as
    \begin{align} \label{eq_ml_uns}
         (\widehat{\alpha},\widehat{\vartheta},\widehat{\varphi}) = \argmin_{\alpha, \vartheta, \varphi} \norm{\by- \alpha \, \tbQ \bG^{\trpose}(\vartheta)\bh(\varphi) \sqrt{E_s} } ~.
    \end{align}
    To enable decoupled estimation of azimuth and elevation parameters in \eqref{eq_ml_uns}, we resort to unstructured \ac{ML} techniques \cite{ML_array_EXIP_98}. To this end, let
    \begin{align} \label{eq_bbf_uns}
        \bbf \triangleq \alpha \, \bh(\varphi) \sqrt{E_s} \in \complexset{(2N+1)}{1}    
    \end{align}
    denote the unstructured vector where the dependency on $\alpha$ and $\varphi$ is dropped. Substituting \eqref{eq_bbf_uns} into \eqref{eq_ml_uns} yields the unstructured ML optimization problem:
    \begin{align} \label{eq_ml_uns2}
         (\widehat{\bbf},\widehat{\vartheta}) = \argmin_{\bbf, \vartheta} \norm{\by- \tbQ \bG^{\trpose}(\vartheta)\bbf } ~.
    \end{align}
    In \eqref{eq_ml_uns2}, $\widehat{\bbf}$ can be readily obtained as a function of $\vartheta$
    \begin{align} \label{eq_bbf_est}
        \widehat{\bbf} = \left(\tbQ \bG^{\trpose}(\vartheta)\right)^{\dagger} \by
    \end{align}
    under the condition $T \geq 2N+1$, noting that $\tbQ \bG^{\trpose}(\vartheta) \in \complexset{T}{(2N+1)}$. By plugging \eqref{eq_bbf_est} back into \eqref{eq_ml_uns2}, elevation can be estimated via a simple line search:
    \begin{equation}\label{eq_est_vartheta}
	    \widehat{\vartheta} = \argmin_{\vartheta \in [0, \, \pi/2]} \norm{\projnull{\tbQ \bG^{\trpose}(\vartheta)} \by }~.
	\end{equation}}
    
    \rev{Using the estimated elevation $\widehat{\vartheta} $ in \eqref{eq_est_vartheta}, the ML problem in \eqref{eq_ml_uns} can be recast as
    \begin{align} \label{eq_ml_uns_azimuth}
         (\widehat{\alpha},\widehat{\varphi}) = \argmin_{\alpha,  \varphi} \norm{\by- \alpha \, \tbQ \bG^{\trpose}(\widehat{\vartheta})\bh(\varphi) \sqrt{E_s} } ~,
    \end{align}
    %Following a similar approach to that in \eqref{eq_ml_uns2} to estimate $\alpha$ and substituting the resulting estimate back into \eqref{eq_ml_uns_azimuth}, 
    where the task of azimuth estimation boils down to a line search:
    \begin{equation}\label{eq_est_varphi}
	    \widehat{\varphi} = \argmin_{\varphi \in [0, \, 2 \pi]} \norm{\projnull{\tbQ \bG^{\trpose}(\widehat{\vartheta})\bh(\varphi) } \by }~.
	\end{equation}    }
    
    \subsubsection{\rev{Distance Estimation Using the Near-Field Model}}\label{sec_jacobi_step4}
    \rev{Having estimated the elevation and azimuth in \eqref{eq_est_vartheta} and \eqref{eq_est_varphi}, respectively, using the far-field approximation, we can now switch back to the original near-field based MML estimator in \eqref{eq_orig_optimization} and retrieve distance information via \eqref{eq_1D_distance2}.} %Given the estimates of elevation and azimuth $(\widehat{\vartheta}, \widehat{\varphi})$, the \ac{MML} estimator \eqref{eq_mml_est} for the original observation model in \eqref{eq_y_orig_near} (based on the near-field steering vector $\ba(\bp)$ in \eqref{eq_ap_nearfield}) is given by 
%     \begin{align} \label{eq_dist_nearfield}
%          (\widehat{\alpha},\widehat{d}) = \argmin_{\alpha, d} \norm{ \by - \alpha \, \tbQ \ba\big(\bp(d, \widehat{\vartheta}, \widehat{\varphi})\big) \sqrt{E_s} } ~.
%     \end{align}
%  Similar to the solution of \eqref{eq_ml_uns_azimuth}, distance can be estimated from \eqref{eq_dist_nearfield} via
%     \begin{equation} \label{eq_est_d}
% 	    \widehat{d} = \argmin_{d \in (0,\dmax)} \norm{\projnull{\tbQ \ba\big(\bp(d, \widehat{\vartheta}, \widehat{\varphi})\big) } \by } ~,
% 	\end{equation}
%     which entails a simple line search for a given maximum distance $\dmax$ imposed by near-field conditions in \eqref{eq_dmax}. After obtaining the estimates of \rev{elevation, azimuth and distance via \eqref{eq_est_vartheta}, \eqref{eq_est_varphi} and \eqref{eq_est_d}, respectively, the AMML estimate for the UE location is given by \eqref{eq_pos_sph}}.}

    \subsection{\rev{Summary of AMML Based Algorithm}}\label{sec_summary}
    \rev{In Sec.~\ref{sec_amml} and Sec.~\ref{sec_jacobi_amml}, we propose two approaches to implement the MML estimator in \eqref{eq_mml_est}. The approach in Sec.~\ref{sec_amml} can be used for any number of transmissions $T$, while the one in Sec.~\ref{sec_jacobi_amml} represents a low-complexity alternative that can be employed for $T \geq \Tthr = 2N+1$ due to the condition imposed in \eqref{eq_bbf_est}. Here, $N$ can be set according to \eqref{eq_Nmin}. Therefore, we can decide on the preferred approach depending on whether $T \geq \Tthr$ holds. The entire algorithm is summarized in Algorithm~\ref{alg_jacobi}.}

	%\subsubsection{\rev{Location Estimation}}\label{sec_step5}
	
	%Based on the UE location estimate, the search intervals in \rev{\eqref{eq_est_vartheta}, \eqref{eq_est_varphi} and \eqref{eq_est_d}} can be made narrower, and \rev{Algorithm~\ref{alg_jacobi}} \rev{can be re-run to refine the estimates}. 
	
% 
	
% 	As discussed in \cite{Shaban2021}, these three-step simple line searches have a low computational complexity. 
	
	%%%%%%%%%%%%%%%%%%%%%%%%%%%%%%%%%%%%%%%%%%%%%%%%%%%%%%%%
	%%%%%%%%%%%%%%%%%%%%%%%%%%%%%%%%%%%%%%%%%%%%%%%%%%%%%%%%
	% Jacob-Anger expansion based efficient localization algorithm
	\begin{algorithm}[t]
	\caption{AMML Algorithm for RIS-aided Near-Field Localization via Jacobi-Anger Expansion}
	\label{alg_jacobi}
% 	\footnotesize
	\begin{algorithmic}[1]
	    \State \textbf{Input:} Observation $\yy$ in \eqref{eq_yt}\rev{, $\Tthr$, $\Imax$.}
	    \State \textbf{Output:} Location estimate $\pphat$.
	    \If{\rev{$T\geq \Tthr$}}
	       % \item Set $\betamin = 1$.
	        \State Estimate the \rev{elevation} $\widehat{\vartheta}$ by solving \eqref{eq_est_vartheta}.
	        \State Using $\widehat{\vartheta}$, estimate the \rev{azimuth} $\widehat{\varphi}$ by solving \eqref{eq_est_varphi}.
	        \State Using $\widehat{\vartheta}$ and $\widehat{\varphi}$, estimate the distance $\widehat{d}$ by solving \eqref{eq_1D_distance2}.
	    \Else
	       % \item Set $\betamin = 1$.
	        \State \rev{Find the initial estimates $\widehat{\vartheta}$ and $\widehat{\varphi}$  by solving \eqref{eq_2D_varthetaphi2}.}
	        \State \rev{Set $i = 0$.}
	         \While {\rev{the objective function in \eqref{eq_orig_optimization} does not converge and $i < \Imax$} }
	            \State \rev{Update $\widehat{d}$ by solving  \eqref{eq_1D_distance2}.}
	            \State \rev{Update $\widehat{\vartheta}$ and $\widehat{\varphi}$ by solving  \eqref{eq_2D_varthetaphi_update2}.}
	            \State \rev{Set $i = i + 1$.} 
	        \EndWhile
	    \EndIf
	   \State Compute the location estimate via \rev{\eqref{eq_pos_sph} as $\pphat = \bp(\widehat{d}, \widehat{\vartheta}, \widehat{\varphi})$}.
	\end{algorithmic}
	\normalsize
\end{algorithm}
	%%%%%%%%%%%%%%%%%%%%%%%%%%%%%%%%%%%%%%%%%%%%%%%%%%%%%%%%
	%%%%%%%%%%%%%%%%%%%%%%%%%%%%%%%%%%%%%%%%%%%%%%%%%%%%%%%%
	\subsection{\rev{Complexity and Convergence of Algorithm~\ref{alg_jacobi}}}\label{sec:alg1_comp_conv}
	\rev{To evaluate the computational complexity of Algorithm~\ref{alg_jacobi}, suppose that the search intervals for distance, azimuth and elevation are discretized into grids of size $K$ each. We assume that $K, M >T$, which is reasonable given the large RIS size and a fine search granularity required for high-quality estimates. As shown in detail in Appendix~\ref{sec:AppCompCostAlg1}, the overall cost of Algorithm~\ref{alg_jacobi} for $T< \Tthr = 2N+1$ and $T \geq \Tthr$ is given by $\mathcal{O}(TK^2M)$ and $\mathcal{O} \left( T M (2N+1)K\right)$, respectively. For $T \geq \Tthr$, we have $K > T \geq 2N+1$, leading to $T M (2N+1) K\leq TK^2 M$. Thus, it can be concluded that, when $T$ is sufficiently large, by using the Jacobi-Anger expansion based estimator in Sec.~\ref{sec_jacobi_amml}, we decrease the computational complexity. Thus, for $T \geq 2N+1$ where $N$ is sufficiently large (e.g., \eqref{eq_Nmin}) to provide a good approximation in \eqref{eq:JacobiAnger}, the Jacobi-Anger based approach can be employed to reduce the computational burden.}
	
	\rev{Regarding the convergence of Algorithm~\ref{alg_jacobi}, it follows from \cite[Prop.~1]{seqOpt_TSP_2018} that if the solutions to \eqref{eq_1D_distance2} and \eqref{eq_2D_varthetaphi_update2} are unique, then Algorithm~\ref{alg_jacobi} converges to a stationary point of the original MML estimation problem in \eqref{eq_orig_optimization}. Note that the uniqueness of the solutions to \eqref{eq_1D_distance2} and \eqref{eq_2D_varthetaphi_update2} can be achieved with high probability using random RIS phase profiles \cite[Sec.~IV-B]{RIS_loc_2021_TWC}.}
	
% 	\rev{Regarding the convergence of Algorithm~\ref{alg_jacobi}, we need to discuss about the alternating iterations given in \eqref{eq_1D_distance2} and \eqref{eq_2D_varthetaphi_update2} or \eqref{eq_1D_distance2}, \eqref{eq_est_vartheta} and \eqref{eq_est_varphi}. From \cite[Prop.~1]{seqOpt_TSP_2018}, we can claim that if the optimal solution of \eqref{eq_orig_optimization} is unique, then via alternating iterations, the Algorithm 1 converges to the unique stationary point. In addition, as discussed in \cite{RIS_loc_2021_TWC}, with high probability the optimal solution should be unique. Hence, we effectively prove that the Algorithm 1 converges to the unique stationary point.}

% 	In this case, the plain \ac{MML} estimator in \eqref{eq_mml_est} requires $\mathcal{O}(K^3)$ operations to perform 3-D search, while Algorithm~\ref{alg_jacobi} involves three line searches, leading to $\mathcal{O}(K)$ complexity. Hence, the proposed solution has significantly lower complexity than the standard \ac{MML} one. Although being a heuristic approach, Algorithm~\ref{alg_jacobi} exhibits near-optimal estimation performance (as will be shown via numerical simulations in Sec.~\ref{sec_results_disc}), converging to the theoretical bounds with increasing \ac{SNR}, which indicates that low-complexity is achieved without sacrificing localization accuracy.

     Next, we analyze how the performance bounds and the estimator structures are affected when the true RIS amplitude model given in \eqref{eq_wtm_ris} is known.

    % \section{RIS aided Localization when the Amplitude Model is Known}
    \section{Scenario-II \& Scenario-III: RIS-aided Localization Under Known RIS Amplitude Model}\label{sec_scen23}
    
    In this section, we investigate RIS-aided localization under Scenario-II and Scenario-III, where the UE is assumed to be aware of the RIS amplitude model in \eqref{eq_wtm_ris}.

    \subsection{Scenario-II: Known RIS Amplitude Model with Unknown Parameters}\label{sec_scen2_crb}
  
    In this scenario, in order to parameterize the unknown system parameters, we will use $\bet = [\text{Re}(\alpha)\, \text{Im}(\alpha) \, \bp^{\trpose} \, \beta_{\text{min}}\, \kappa \, \phi]^{\trpose}$, which consists of both the channel parameters and the RIS model parameters. Then, by \cite[Eq.~ 9] {Shaban2021}, the Fisher Information matrix (FIM), $\bJ(\bet)\in\mathbb{R}^{8\times 8}$, can be expressed as
    \begin{equation}\label{eq:FIMCaseII}
\bJ(\bet) =  \frac{2}{N_0} \text{Re}\left\{\left(\frac{\partial \bmu}{\partial \bet}\right)^{\mathsf{H}} \frac{\partial \bmu}{\partial \bet}\right\}.
    \end{equation}
	In order to compute the derivatives of $\bmu$ with respect to the first five entries of $\bet$, we can use the derivatives given in Appendix \ref{sec:AppB} by replacing the $\tilde{w}_{t,m}$ terms with $w_{t,m}$. For the last three entries of $\bet$, i.e., for the RIS related parameters, the \rev{derivatives used in \eqref{eq:FIMCaseII}} are \rev{provided} in Appendix~\ref{sec:AppC}.

After obtaining the FIM, by computing $\Tr\{\bJ^{-1}(\bet)\}_{3:5,3:5}$, we can obtain the CRB for estimating the UE position. Moreover, the ML estimator can be stated as
\begin{equation}\label{eq_ml_case2}
    \hat{\bet}_{\text{ML}}(\by) = \argmax_{\bet\in\mathbb{R}^8} \log p(\by).
\end{equation}
As discussed in Section~\ref{sec:MCRBder}, the estimate for $\overline{\alpha}$ can  uniquely be determined for given estimates of $\bbp$, $\bbetamin$, $\overline{\kappa}$, and $\overline{\phi}$. That is, this problem can be reduced to a 6-dimensional optimization problem.  Similar to the mismatched scenario (Scenario-I), the initialization is an important issue for this scenario, as well. For practical implementations, we  propose an approximate version of the ML estimator, called the approximate maximum likelihood (AML) estimator, in the next section.

%%%%%%%%%%%%%%%%%%%%%%%%%%%%%%%%%%%%%%%%%%%%%%%%%%
%%%%%%%%%%%%%%%%%%%%%%%%%%%%%%%%%%%%%%%%%%%%%%%%%%

\subsection{Scenario-II: AML Estimator}\label{sec:AML}
In this part, our goal is to solve the problem of \textit{joint localization and online RIS calibration} in \eqref{eq_ml_case2}, which involves estimating the UE location $\bbp$ and the RIS model parameters $\bbetamin$, $\overline{\kappa}$, and $\overline{\phi}$ simultaneously. To accomplish this in an efficient manner, we propose a low-complexity estimator as an alternative to the high-dimensional non-convex optimization in \eqref{eq_ml_case2}. 

Let us write the observations in \eqref{eq_yt} in a vector form as
\begin{align}\label{eq_y_vec}
    \yy = \bar{\alpha} \, \bW^\trpose(\overline{\zetab}) \bb(\bbp) \sqrt{E_s} + \nn ,
\end{align}
where $\bW(\overline{\zetab}) = [ \ww_1(\overline{\zetab}) \, \ldots \, \ww_T(\overline{\zetab}) ] \in \complexset{M}{T}$ is the matrix of RIS profiles as a function of the parameters $\overline{\zetab} = [\bbetamin ~ \overline{\kappa} ~ \overline{\phi}]^\trpose$ of the RIS amplitude model in \eqref{eq_wtm_ris} and \eqref{eq_beta_model}, and $\nn = [n_1 \, \ldots \, n_T]^\trpose$ is the additive noise vector. From \eqref{eq_wtm_ris} and \eqref{eq_beta_model}, $\bW(\overline{\zetab})$ can be expressed as
\begin{align} \label{eq_bw}
    \bW(\overline{\zetab}) = \left( \bbetamin \Gammab_1(\overline{\kappa}, \overline{\phi}) + \Gammab_2(\overline{\kappa}, \overline{\phi}) \right) \odot e^{j \Thetab} ,
\end{align}
where $\Thetab \in \realset{M}{T}$ denotes the RIS phase shifts with $[\Thetab]_{m,t} = \theta_{m,t}$, 
\begin{equation}
    \Gammab_1(\overline{\kappa}, \overline{\phi}) \triangleq \boldone_M \boldone_T^\trpose - \left(\frac{\sin(\Thetab-\overline{\phi})+1}{2}\right)^{\overline{\kappa}} ~,
\end{equation}
\begin{equation}
    \Gammab_2(\overline{\kappa}, \overline{\phi}) \triangleq \left(\frac{\sin(\Thetab-\overline{\phi})+1}{2}\right)^{\overline{\kappa}} ~,
\end{equation}
$\boldone_M$ is an all-ones vector of size $M$, and $\sin(\cdot)$ and $(\cdot)^{\overline{\kappa}}$ operations are element-wise.
Plugging $\bb(\bbp) = \aaa(\bbp) \odot \aaa(\ppbs)$ and \eqref{eq_bw} into \eqref{eq_y_vec} yields
\begin{align}\label{eq_y_new}
    \yy = \overline{\alpha}  \left( \bbetamin \Gammabt_1(\overline{\kappa}, \overline{\phi}) + \Gammabt_2(\overline{\kappa}, \overline{\phi}) \right)^\trpose \aaa(\bbp) \sqrt{E_s}+ \nn ,
\end{align}
where
\begin{align}
    \Gammabt_1(\overline{\kappa}, \overline{\phi}) &\triangleq  \Gammab_1(\overline{\kappa}, \overline{\phi}) \odot e^{j \Thetab} \odot \aaa(\ppbs) \boldone_T^\trpose  ,
    \\ \nonumber
    \Gammabt_2(\overline{\kappa}, \overline{\phi}) &\triangleq   \Gammab_2(\overline{\kappa}, \overline{\phi}) \odot e^{j \Thetab} \odot \aaa(\ppbs) \boldone_T^\trpose .
\end{align}
The ML estimator corresponding to the observation model in \eqref{eq_y_new} is given by
\begin{align}\label{eq_ml_prop}
    \min_{\alpha, \pp, \zetab} ~ \norm{\yy - \alpha  \left( \betamin \Gammabt_1(\kappa, \phi) + \Gammabt_2(\kappa, \phi) \right)^\trpose \aaa(\pp)\sqrt{E_s}}^2 .
\end{align}

To solve the ML optimization problem in \eqref{eq_ml_prop}, we now propose a three-step procedure consisting of location initialization, online RIS calibration and location refinement.

\subsubsection{Step 0: Initialization of UE Location}
First, by assuming $\betamin = 1$, the initial location estimate $\pphat$ is obtained by using the Jacobi-Anger expansion based algorithm in Algorithm~\ref{alg_jacobi}. Inserting $\pphat$ into the ML estimator in  \eqref{eq_ml_prop}, the optimization problem becomes
\begin{align}\label{eq_ml_prop2}
    \min_{\alpha, \zetab} ~ \norm{\yy -  \alpha \left( \betamin \Gammabt_1(\kappa, \phi) + \Gammabt_2(\kappa, \phi) \right)^\trpose \ba(\pphat)\sqrt{E_s}}^2 .
\end{align}

\subsubsection{Step 1: Online RIS Calibration}
As noted from \eqref{eq_ml_prop2}, $\overline{\alpha}$ and $\bbetamin$ can now be estimated in closed-from as a function of the remaining unknown parameters, while the estimates of $\overline{\kappa}$ and $\overline{\phi}$ can be found via a 2-D search. This motivates an alternating optimization algorithm where we alternate among updates of $\alpha$, $\betamin$, and $(\kappa, \phi)$:
\begin{itemize}
    \item \textit{Update $\alpha$}: Given $\zetab$, a closed-form estimate of $\alpha$ in \eqref{eq_ml_prop2} is given by
    \begin{align}\label{eq_hhhat}
        \alphahat = \left(\bUps^\hermit(\zetab) \bUps(\zetab)\right)^{-1} \bUps^\hermit(\zetab) \yy,
    \end{align}
    where
    \begin{align}
        \bUps(\zetab) &\triangleq \left( \betamin \Gammabt_1(\kappa, \phi) + \Gammabt_2(\kappa, \phi) \right)^\trpose \ba(\pphat)\sqrt{E_s}.
    \end{align}
    \item \textit{Update $\betamin$}: Given $\alpha$, $\kappa$, and $\phi$, under the constraint of $0\leq \betamin \leq 1$, the Lagrangian can be expressed as follows:
    \begin{equation*}
        L(\betamin,\lambda_1,\lambda_2) =  \norm{\bz- \betamin \bomeg }^2 + \lambda_1 (\betamin-1) -\lambda_2 \betamin
    \end{equation*}
    where $\bz = \by-  \alpha \Gammabt_2(\kappa, \phi) ^\trpose \ba(\pphat)\sqrt{E_s}$ and $\bomeg =  \alpha\Gammabt_1(\kappa, \phi) ^\trpose \ba(\pphat)\sqrt{E_s}$. From Karush-Kuhn-Tucker conditions, the closed-form estimate of $\betamin$ in \eqref{eq_ml_prop2} can be obtained from the following set of equations:
    \begin{align}\label{eq_betaminhat}
        \lambda_1 -\lambda_2 + 2\betaminhat \bomeg^\hermit \bomeg &= \bz^\hermit \bomeg + \bomeg^\hermit \bz \\ \label{eq_betaminhat2}
        \lambda_1(\betaminhat-1) = 0, \,~ &\lambda_2 \betaminhat = 0.
    \end{align}
    Hence, we can conclude that $\betaminhat$ admits one of the three alternative forms: $\betaminhat = 1$,  $\betaminhat = 0$, or $\betaminhat = \text{Re}\{\bz^\hermit \bomeg \}/(\bomeg^\hermit \bomeg)$.
    Among these three solutions, the one which yields the smallest objective is chosen.
    \item \textit{Update $\kappa$ and $\phi$}: Given $\alpha$ and $\betamin$, we can estimate $\kappa$ and $\phi$ via a 2-D search:
    \begin{align}\label{eq_ml_prop3}
&(\kappahat, \phihat) = \nonumber \\ 
   &\argmin_{\substack{\kappa\in [0,\rev{\kappamax}) \\ \phi \in [0,2\pi)}} \norm{\yy -  \alpha\left( \betamin \Gammabt_1(\kappa, \phi) + \Gammabt_2(\kappa, \phi) \right)^\trpose \ba(\pphat)\sqrt{E_s}}
\end{align}
\rev{for some finite $\kappamax$.}
\end{itemize}

\subsubsection{Step 2: Refinement of UE Location}
As the output of this alternating procedure, we obtain the estimates of $\overline{\alpha}$, $\bbetamin$, $\overline{\kappa}$, and $\overline{\phi}$. By plugging the estimates of $\bbetamin$, $\overline{\kappa}$, and $\overline{\phi}$ (representing the \textit{calibrated} RIS parameters) back into the ML estimator in \eqref{eq_ml_prop}, the estimate of the UE location $\bbp$ can be refined
% \begin{align}\label{eq_dh_alpha}
%     \min_{\alpha, \pp} \norm{\bD \hhhat - \alpha \, \aaa(\pp)}^2 ~,
% \end{align}
via Algorithm~\ref{alg_jacobi}. The overall algorithm for joint UE localization and online RIS calibration is summarized in Algorithm~\ref{alg_loc_est}.

\begin{algorithm}[t]
	\caption{AML Algorithm for Joint UE Localization and Online RIS Calibration}
	\label{alg_loc_est}
% 	\footnotesize
	\begin{algorithmic}[1]
	    \State \textbf{Input:} Observation $\yy$ in \eqref{eq_y_new}, \rev{$\Tthr,\Jmax$}.
	    \State \textbf{Output:} Location estimate $\pphat$, channel gain estimate $\alphahat$ and estimates of RIS amplitude model parameters $\zetabhat = [\betaminhat, \kappahat, \phihat]^\trpose$.
	    \State \textbf{Step~0:} \textit{Initialization of UE Location} 
	    \begin{enumerate}[label=(\alph*)]
	        \item Set $\betamin = 1$.
	        \item Compute the initial location estimate $\pphat$ using Algorithm~\ref{alg_jacobi}.
	    \end{enumerate}
	    
	    \State  \textbf{Step~1:} \textit{Alternating Iterations for Online RIS Calibration} 
	    \State \rev{Set $m = 0$}.
	    \While {the objective function in \eqref{eq_ml_prop2} does not converge \rev{and $m< \Jmax$}}
	   
	    \State Update the channel estimate $\alphahat$ via \eqref{eq_hhhat}.
	    \State Update $\betaminhat$ via \eqref{eq_betaminhat} \rev{and \eqref{eq_betaminhat2}}.
	    \State Update $\kappahat$ and $\phihat$ via \eqref{eq_ml_prop3}.
	    \State \rev{Set $m = m + 1$}.
	    \EndWhile
	    \State  \textbf{Step~2:} \textit{Refinement of UE Location with Calibrated RIS Model}
	    \begin{enumerate}[label=(\alph*)]
	    \item Use Algorithm~\ref{alg_jacobi} to estimate the UE location, $\pphat$, and the channel gain, $\alphahat$, from \eqref{eq_ml_prop} by plugging the estimates $\betaminhat$, $\kappahat$ and $\phihat$ obtained at the output of \textbf{Step~1}.
	    \end{enumerate}
	\end{algorithmic}
	\normalsize
\end{algorithm}

\begin{rem}
If we skip Step 1 of Algorithm~\ref{alg_loc_est} (i.e., online RIS calibration) and choose $\betaminhat = 1$ (i.e., uncalibrated, unit-amplitude RIS model in \eqref{eq_wt_assumed}), the AML algorithm gives the same estimate as the AMML algorithm.
\end{rem}

\subsection{\rev{Complexity and Convergence of Algorithm~\ref{alg_loc_est}}}\label{sec:alg2_comp_conv}
\rev{We evaluate the complexity of Algorithm~\ref{alg_loc_est} by assuming that the search intervals for $\kappa$ and $\phi$ are discretized into grids of size $L$ each. If $T<\Tthr$, the overall cost of Algorithm~\ref{alg_loc_est} is simply equal to $\mathcal{O}\left(TM \max\{K,L\}^2\right)$, where $K$ is defined as in Sec.~\ref{sec:alg1_comp_conv}. If $T\geq \Tthr$, the computational cost becomes $\mathcal{O}\left(TM \max\{K(2N+1),L^2\}\right)$. The details can be found in Appendix~\ref{sec:AppCompCostAlg2}}.

\rev{The convergence of Algorithm~\ref{alg_loc_est} can be analyzed by noting that if the solutions to the subproblems of \eqref{eq_ml_prop2} via the alternating iterations in Step~1 are unique (i.e., $\alpha$, $\betamin$ and $(\kappa, \phi)$ are obtained uniquely in \eqref{eq_hhhat}, \eqref{eq_betaminhat}-\eqref{eq_betaminhat2} and \eqref{eq_ml_prop3}, respectively), then Step~1 converges to a stationary point of \eqref{eq_ml_prop2} \cite[Prop.~1]{seqOpt_TSP_2018}. Due to the closed-form solution in \eqref{eq_hhhat}, the uniqueness of $\alpha$ is guaranteed, while for $\betamin$ and $(\kappa, \phi)$ the uniqueness is satisfied with high probability via random RIS phase configurations, following the arguments in \cite[Sec.~IV-B]{RIS_loc_2021_TWC}.}

% \rev{In addition, to analyze the convergence behavior of the Algorithm 2, by using the same ideas given in Sec.~\ref{sec:alg1_comp_conv}, we can claim that the Algorithm 2 converges to the unique stationary point.}

\subsection{ Scenario-III: Known RIS Amplitude Model with Known Parameters}

In this scenario, the unknown system parameters are given by $\bet = [\text{Re}(\alpha)\, \text{Im}(\alpha) \, \bp^{\trpose} \,]^{\trpose}$ and the  values of the RIS related parameters are perfectly known. For this scenario, the FIM, $\bJ(\bet)\in\mathbb{R}^{5\times 5}$, can be computed as in \eqref{eq:FIMCaseII}. In addition, the derivatives presented in Appendix \ref{sec:AppB} can be used by replacing $\tilde{w}_{t,m}$'s with $w_{t,m}$'s. 

To estimate the UE location, we can employ the same AML algorithm, Algorithm~\ref{alg_loc_est}, as used in Scenario-II.
%again approximate the ML estimator. 
As the  values of $\bbetamin, \overline{\kappa}$, and $\overline{\phi}$ are available, we can skip Step~0 and Step~1 of Algorithm~\ref{alg_loc_est} and run Step~2 with the known values of the RIS amplitude model parameters.

%%%%%%%%%%%%%%%%%%%%%%%%%%%%%%%%%%%%%%%%%%%%%%%%%%
%%%%%%%%%%%%%%%%%%%%%%%%%%%%%%%%%%%%%%%%%%%%%%%%%%

\section{Numerical Results}\label{sec:Nume}
	
	In this section, we first present numerical examples for evaluating the theoretical %lower  and Cram\'{e}r-Rao 
	bounds in three different scenarios, and then compare the performance of the AMML and AML estimators against the theoretical bounds.
	%lower and  Cram\'{e}r-Rao bounds.  

	\subsection{Simulation Setup}\label{sec_sim}
	We consider an RIS with $M = 50\times 50$ elements, where the inter-element spacing is $\lambda/2$ and the area of each element is $A = \lambda^2/4$ \cite{Shaban2021}. The carrier frequency is equal to $f_c = 28\,$ GHz \rev{and the bandwidth is set to $1$ MHz}. The RIS is modeled to lie in the X-Y plane with $\bp_{\text{RIS}}=[0\, 0\, 0]^{\trpose}$. Moreover, for the RIS elements, $\theta_{t,m}$ values are generated uniformly and independently between $-\pi$ and $\pi$. \rev{In accordance with the Fresnel near-field region defined in \eqref{eq_dmax}, whose boundaries are $1.40$ and $26.79$ meters,} the BS and \rev{UE are} located at $\bp_{\text{BS}}=5.77\times 
	[-1\, 1\, 1]^{\trpose}$ and \rev{$\bbp=2.89\times 
	[1\, 1\, 1]^{\trpose}$ meters, respectively, leading to $10$ and $5$ meters of distance to the RIS}. We set the number of transmissions to $T = 200$, \rev{yielding a total duration of $0.2$ ms}. For simplicity, we assume that $s_t = \sqrt{E_s}$ for any $t$. Also, the  SNR is defined as
	\begin{gather}\label{eq_snr}
		\text{SNR} = \frac{E_s \abs{\overline{\alpha}}^2}{T N_0} \sum_{t=1}^{T}  \abs{\bb^{\trpose}(\bbp) \bw_t}^2\,.
	\end{gather}
	To solve \eqref{eq_etabar3} for the LB computation, we employ the GlobalSearch algorithm of MATLAB by providing $\bbp$ as the initial vector. \rev{For the Jacobi-Anger approximation in \eqref{eq:JacobiAnger}, $\frac{2\pi}{\lambda} \qmax \sin \vartheta $ in \eqref{eq_Nmin} evaluates to $90.7$ for the considered setup, while we set $N=50$, which turns out to be sufficient for the proposed algorithms to converge to their respective theoretical bounds. Furthermore, we set $\Tthr = 2N+1$, $\kappamax = 5$ and $\Imax = \Jmax = 5$ in Algorithm~\ref{alg_jacobi} and Algorithm~\ref{alg_loc_est}. Since $T\geq \Tthr$, Algorithm~\ref{alg_jacobi} chooses the Jacobi-Anger based branch}. \rev{To evaluate the performance of the branch $T < \Tthr$, we also show simulation results for $T=10$.}
	
	\subsection{Results and Discussions}\label{sec_results_disc}
	
		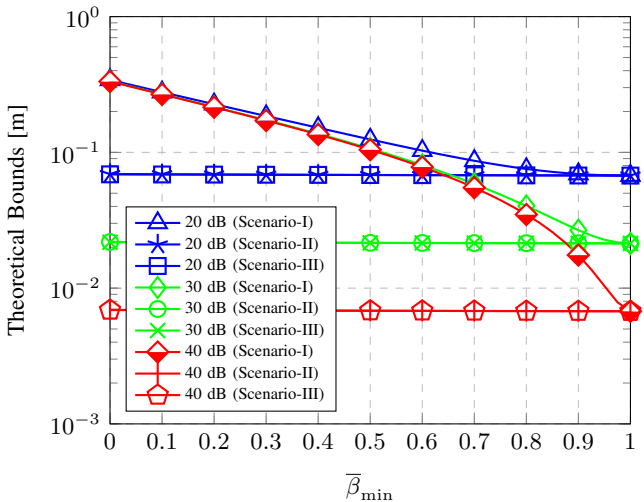
\begin{figure}%[H]
		\centering
		\begin{tikzpicture}
			\begin{semilogyaxis}[
				width=8.5cm,
                height=7cm,
				legend style={nodes={scale= 0.7, transform shape},at={(1,1)},anchor=north east}, 
				legend cell align={left},
				legend image post style={mark indices={}},
				xlabel={$\bbetamin$},
				ylabel={Theoretical Bounds [m]},
				xmin=0, xmax=1,
				ymin=0.001, ymax=1,
				xtick={0, 0.1, 0.2, 0.3, 0.4, 0.5, 0.6, 0.7, 0.8, 0.9, 1},
				ytick={0.001,0.01,0.1,1},
				legend pos=south west,
				ymajorgrids=true,
				xmajorgrids=true,
				grid style=dashed,
				]
				
				\addplot[thick,
				color=blue,
				mark = triangle,
				mark indices={1,5,9,13,17,21,25,29,33,37,41},
				mark options={solid},			
				mark size = 4pt,
				]
				coordinates {
            (1,0.067478)(0.975,0.067562)(0.95,0.067883)(0.925,0.068448)(0.9,0.069297)(0.875,0.070429)(0.85,0.071845)(0.825,0.073569)(0.8,0.075556)(0.775,0.077848)(0.75,0.080514)(0.725,0.083446)(0.7,0.086717)(0.675,0.090356)(0.65,0.094245)(0.625,0.098448)(0.6,0.10303)(0.575,0.10795)(0.55,0.11313)(0.525,0.11873)(0.5,0.12456)(0.475,0.13078)(0.45,0.13747)(0.425,0.14442)(0.4,0.15173)(0.375,0.15945)(0.35,0.16766)(0.325,0.17629)(0.3,0.18527)(0.275,0.1948)(0.25,0.2049)(0.225,0.2155)(0.2,0.22658)(0.175,0.23838)(0.15,0.25077)(0.125,0.26397)(0.1,0.27772)(0.075,0.29248)(0.05,0.30799)(0.025,0.32461)(0,0.34219)
				};
				
				\addplot[thick, 
				color=blue,
				mark = star,
				mark indices={1,5,9,13,17,21,25,29,33,37,41},
				mark options={solid},
				mark size = 4pt
				]
				coordinates {
             (1,0.067529)(0.975,0.067646)(0.95,0.067666)(0.925,0.067686)(0.9,0.067708)(0.875,0.06773)(0.85,0.067753)(0.825,0.067777)(0.8,0.067802)(0.775,0.067828)(0.75,0.067854)(0.725,0.067882)(0.7,0.06791)(0.675,0.067939)(0.65,0.06797)(0.625,0.068001)(0.6,0.068033)(0.575,0.068066)(0.55,0.0681)(0.525,0.068135)(0.5,0.068172)(0.475,0.068209)(0.45,0.068247)(0.425,0.068287)(0.4,0.068327)(0.375,0.068368)(0.35,0.068411)(0.325,0.068455)(0.3,0.0685)(0.275,0.068547)(0.25,0.068595)(0.225,0.068645)(0.2,0.068696)(0.175,0.06875)(0.15,0.068805)(0.125,0.068864)(0.1,0.068925)(0.075,0.068989)(0.05,0.069056)(0.025,0.069128)(0,0.069204)
				};
				
				\addplot[thick, 
				color=blue,
				mark = square,
				mark indices={1,5,9,13,17,21,25,29,33,37,41},
				mark options={solid},
				mark size = 3pt,
				]
				coordinates {
					    (1,0.067478)(0.975,0.067503)(0.95,0.06753)(0.925,0.067557)(0.9,0.067584)(0.875,0.067613)(0.85,0.067642)(0.825,0.067673)(0.8,0.067704)(0.775,0.067735)(0.75,0.067768)(0.725,0.067801)(0.7,0.067836)(0.675,0.06787)(0.65,0.067906)(0.625,0.067942)(0.6,0.067979)(0.575,0.068017)(0.55,0.068055)(0.525,0.068094)(0.5,0.068134)(0.475,0.068173)(0.45,0.068214)(0.425,0.068254)(0.4,0.068295)(0.375,0.068336)(0.35,0.068376)(0.325,0.068417)(0.3,0.068458)(0.275,0.068499)(0.25,0.068539)(0.225,0.068579)(0.2,0.068619)(0.175,0.068657)(0.15,0.068696)(0.125,0.068733)(0.1,0.06877)(0.075,0.068806)(0.05,0.068841)(0.025,0.068874)(0,0.068907)

				};

				\addplot[thick,
				color=green,
				mark = diamond,
				mark indices={1,5,9,13,17,21,25,29,33,37,41},
				mark options={solid},
				mark size = 4pt,
				]
				coordinates {
			    (1,0.021338)(0.975,0.021712)(0.95,0.022735)(0.925,0.024413)(0.9,0.026748)(0.875,0.029589)(0.85,0.032832)(0.825,0.036458)(0.8,0.040306)(0.775,0.044436)(0.75,0.04894)(0.725,0.053598)(0.7,0.058528)(0.675,0.06376)(0.65,0.069116)(0.625,0.074694)(0.6,0.080588)(0.575,0.086725)(0.55,0.093035)(0.525,0.099699)(0.5,0.10649)(0.475,0.11363)(0.45,0.12118)(0.425,0.12892)(0.4,0.13696)(0.375,0.14537)(0.35,0.15421)(0.325,0.16344)(0.3,0.17296)(0.275,0.183)(0.25,0.19358)(0.225,0.20461)(0.2,0.21609)(0.175,0.22827)(0.15,0.241)(0.125,0.25452)(0.1,0.26855)(0.075,0.28357)(0.05,0.29931)(0.025,0.31614)(0,0.33391)
				};
				
				\addplot[thick,
				color=green,
				mark = o,
				mark indices={1,5,9,13,17,21,25,29,33,37,41},
				mark options={solid},
				mark size = 3pt,
				]
				coordinates {
			      (1,0.021354)(0.975,0.021391)(0.95,0.021398)(0.925,0.021404)(0.9,0.021411)(0.875,0.021418)(0.85,0.021425)(0.825,0.021433)(0.8,0.021441)(0.775,0.021449)(0.75,0.021457)(0.725,0.021466)(0.7,0.021475)(0.675,0.021484)(0.65,0.021494)(0.625,0.021504)(0.6,0.021514)(0.575,0.021524)(0.55,0.021535)(0.525,0.021546)(0.5,0.021558)(0.475,0.02157)(0.45,0.021582)(0.425,0.021594)(0.4,0.021607)(0.375,0.02162)(0.35,0.021634)(0.325,0.021647)(0.3,0.021662)(0.275,0.021676)(0.25,0.021692)(0.225,0.021707)(0.2,0.021724)(0.175,0.021741)(0.15,0.021758)(0.125,0.021777)(0.1,0.021796)(0.075,0.021816)(0.05,0.021838)(0.025,0.02186)(0,0.021884)

				};

				\addplot[thick,
				color=green,
				mark = x,
				mark indices={1,5,9,13,17,21,25,29,33,37,41},
				mark options={solid},
				mark size = 4pt,
				]
				coordinates {
				 (1,0.021338)(0.975,0.021346)(0.95,0.021355)(0.925,0.021363)(0.9,0.021372)(0.875,0.021381)(0.85,0.02139)(0.825,0.0214)(0.8,0.02141)(0.775,0.02142)(0.75,0.02143)(0.725,0.021441)(0.7,0.021451)(0.675,0.021463)(0.65,0.021474)(0.625,0.021485)(0.6,0.021497)(0.575,0.021509)(0.55,0.021521)(0.525,0.021533)(0.5,0.021546)(0.475,0.021558)(0.45,0.021571)(0.425,0.021584)(0.4,0.021597)(0.375,0.02161)(0.35,0.021623)(0.325,0.021635)(0.3,0.021648)(0.275,0.021661)(0.25,0.021674)(0.225,0.021687)(0.2,0.021699)(0.175,0.021711)(0.15,0.021724)(0.125,0.021735)(0.1,0.021747)(0.075,0.021758)(0.05,0.021769)(0.025,0.02178)(0,0.02179)

				};

				\addplot[thick,
				color=red,
				mark = halfsquare*,
				mark options={solid},
				mark indices={1,5,9,13,17,21,25,29,33,37,41},
				mark size = 4pt,
				]
				coordinates {
				   (1,0.0067478)(0.975,0.0078792)(0.95,0.01038)(0.925,0.013679)(0.9,0.017515)(0.875,0.021611)(0.85,0.025876)(0.825,0.030346)(0.8,0.034874)(0.775,0.039572)(0.75,0.044569)(0.725,0.049636)(0.7,0.054919)(0.675,0.060461)(0.65,0.066079)(0.625,0.071888)(0.6,0.077989)(0.575,0.084309)(0.55,0.09078)(0.525,0.097592)(0.5,0.10451)(0.475,0.11177)(0.45,0.11943)(0.425,0.12727)(0.4,0.1354)(0.375,0.14388)(0.35,0.1528)(0.325,0.1621)(0.3,0.17168)(0.275,0.18178)(0.25,0.19241)(0.225,0.20349)(0.2,0.21501)(0.175,0.22723)(0.15,0.24001)(0.125,0.25356)(0.1,0.26762)(0.075,0.28267)(0.05,0.29843)(0.025,0.31528)(0,0.33307)
				};
				
				\addplot[thick,
				color=red,
				mark = |,
				mark options={solid},
				mark indices={1,5,9,13,17,21,25,29,33,37,41},
				mark size = 4pt,
				]
				coordinates {
					  (1,0.0067529)(0.975,0.0067646)(0.95,0.0067666)(0.925,0.0067686)(0.9,0.0067708)(0.875,0.006773)(0.85,0.0067753)(0.825,0.0067777)(0.8,0.0067802)(0.775,0.0067828)(0.75,0.0067854)(0.725,0.0067882)(0.7,0.006791)(0.675,0.0067939)(0.65,0.006797)(0.625,0.0068001)(0.6,0.0068033)(0.575,0.0068066)(0.55,0.00681)(0.525,0.0068135)(0.5,0.0068172)(0.475,0.0068209)(0.45,0.0068247)(0.425,0.0068287)(0.4,0.0068327)(0.375,0.0068368)(0.35,0.0068411)(0.325,0.0068455)(0.3,0.00685)(0.275,0.0068547)(0.25,0.0068595)(0.225,0.0068645)(0.2,0.0068696)(0.175,0.006875)(0.15,0.0068805)(0.125,0.0068864)(0.1,0.0068925)(0.075,0.0068989)(0.05,0.0069056)(0.025,0.0069128)(0,0.0069204)

				};

				\addplot[thick,
				color=red,
				mark = pentagon,
				mark options={solid},
				mark indices={1,5,9,13,17,21,25,29,33,37,41},
				mark size = 4pt,
				]
				coordinates {
					     (1,0.0067478)(0.975,0.0067503)(0.95,0.006753)(0.925,0.0067557)(0.9,0.0067584)(0.875,0.0067613)(0.85,0.0067642)(0.825,0.0067673)(0.8,0.0067704)(0.775,0.0067735)(0.75,0.0067768)(0.725,0.0067801)(0.7,0.0067836)(0.675,0.006787)(0.65,0.0067906)(0.625,0.0067942)(0.6,0.0067979)(0.575,0.0068017)(0.55,0.0068055)(0.525,0.0068094)(0.5,0.0068134)(0.475,0.0068173)(0.45,0.0068214)(0.425,0.0068254)(0.4,0.0068295)(0.375,0.0068336)(0.35,0.0068376)(0.325,0.0068417)(0.3,0.0068458)(0.275,0.0068499)(0.25,0.0068539)(0.225,0.0068579)(0.2,0.0068619)(0.175,0.0068657)(0.15,0.0068696)(0.125,0.0068733)(0.1,0.006877)(0.075,0.0068806)(0.05,0.0068841)(0.025,0.0068874)(0,0.0068907)

				};

				\legend{20 dB (Scenario-I), 20 dB (Scenario-II), 20 dB (Scenario-III), 30 dB (Scenario-I), 30 dB (Scenario-II), 30 dB (Scenario-III), 40 dB (Scenario-I), 40 dB (Scenario-II), 40 dB (Scenario-III) }
				
			\end{semilogyaxis}
		\end{tikzpicture}
		\caption{Theoretical bounds versus $\bbetamin$ for SNR = $20\,$dB, $30\,$dB and $40\,$dB when the UE distance is $5$ meters, $\overline{\kappa} = 1.5$ and $\overline{\phi} = 0$.}
		\label{fig:2}
	\end{figure}	
	
	\subsubsection{Theoretical Limits vs. RIS Model Parameters}
	In Fig.~\ref{fig:2}, for all the three scenarios, we show the theoretical bounds as a function of $\bbetamin$  for SNRs of $20$, $30$, and $40$ dB when the UE distance is $5$ meters from the center of the RIS, $\overline{\kappa} = 1.5$, and $\overline{\phi} = 0$.  We observe from the figure that as $\bbetamin$ decreases, i.e., as the mismatch between the true and the assumed models increases, the LB increases and raising the SNR level does not improve the LB values significantly. In addition, the sensitivity to the model mismatch is more pronounced at higher SNRs, while for an SNR of  $20$ dB, the performance is relatively insensitive for $\bbetamin>0.7$. This shows that being unaware of the true RIS amplitude model can constitute a crucial limiting factor for RIS-aided localization at high SNRs. Interestingly, we note that when the true model and the true values of $\bbetamin$, $\overline{\kappa}$ and $\overline{\phi}$ are known, the value of $\bbetamin$ does not influence the CRB values notably. In fact, as the CRB values for the Scenario-II and Scenario-III are almost the same, it can be inferred that once we know the true model, knowing the true values of $\bbetamin$, $\overline{\kappa}$ and $\overline{\phi}$ is not critical.

	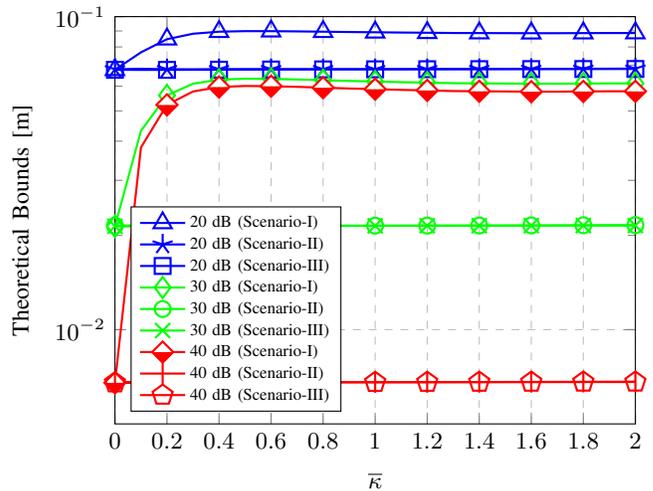
\begin{figure}%[H]
		\centering
		\begin{tikzpicture}
			\begin{semilogyaxis}[
				width=8.5cm,
                height=7cm,
				legend style={nodes={scale= 0.7, transform shape},at={(1,1)},anchor=north east}, 
				legend cell align={left},
				legend image post style={mark indices={}},
				xlabel={$\overline{\kappa}$},
				ylabel={Theoretical Bounds [m]},
				xmin=0, xmax=2,
				ymin=0.005, ymax=0.1,
				xtick={0, 0.2, 0.4, 0.6, 0.8, 1, 1.2, 1.4, 1.6, 1.8, 2},
				ytick={0.01,0.1,1},
				legend pos=south west,
				ymajorgrids=true,
				xmajorgrids=true,
				grid style=dashed,
				]
				
				\addplot[thick,
				color=blue,
				mark = triangle,
				mark indices={1,3,5,7,9,11,13,15,17,19,21},
				mark options={solid},
				mark size = 4pt,
				]
				coordinates {
              (0,0.06765)(0.1,0.0769)(0.2,0.084707)(0.3,0.088175)(0.4,0.089516)(0.5,0.089874)(0.6,0.089857)(0.7,0.089716)(0.8,0.089513)(0.9,0.089303)(1,0.089134)(1.1,0.088955)(1.2,0.088822)(1.3,0.088659)(1.4,0.088611)(1.5,0.08857)(1.6,0.088518)(1.7,0.088525)(1.8,0.088587)(1.9,0.088636)(2,0.088633)
				};

				\addplot[thick,
				color=blue,
				mark = star,
				mark indices={1,3,5,7,9,11,13,15,17,19,21},
				mark options={solid},
				mark size = 4pt,
				]
				coordinates {
                 (0,0.067923)(0.1,0.068136)(0.2,0.067951)(0.3,0.067901)(0.4,0.067967)(0.5,0.067985)(0.6,0.067981)(0.7,0.067971)(0.8,0.067963)(0.9,0.067958)(1,0.06796)(1.1,0.067968)(1.2,0.067983)(1.3,0.068003)(1.4,0.068027)(1.5,0.068055)(1.6,0.068086)(1.7,0.068119)(1.8,0.068153)(1.9,0.068188)(2,0.068223)
				};
				\addplot[thick,
				color=blue,
				mark = square,
				mark indices={1,3,5,7,9,11,13,15,17,19,21},
				mark options={solid},
				mark size = 3pt,
				]
				coordinates {
                 (0,0.06765)(0.1,0.067586)(0.2,0.067604)(0.3,0.067649)(0.4,0.067698)(0.5,0.067742)(0.6,0.06778)(0.7,0.067813)(0.8,0.06784)(0.9,0.067865)(1,0.067886)(1.1,0.067906)(1.2,0.067924)(1.3,0.067941)(1.4,0.067958)(1.5,0.067974)(1.6,0.067989)(1.7,0.068004)(1.8,0.068019)(1.9,0.068033)(2,0.068047)

				};

				\addplot[thick,
				color=green,
				mark = diamond,
				mark indices={1,3,5,7,9,11,13,15,17,19,21},
				mark options={solid},
				mark size = 4pt,
				]
				coordinates {
                 (0,0.021393)(0.1,0.04319)(0.2,0.056004)(0.3,0.061066)(0.4,0.062907)(0.5,0.063335)(0.6,0.063243)(0.7,0.062991)(0.8,0.062655)(0.9,0.062317)(1,0.062044)(1.1,0.061759)(1.2,0.061546)(1.3,0.061286)(1.4,0.061204)(1.5,0.061133)(1.6,0.061045)(1.7,0.061048)(1.8,0.061135)(1.9,0.061204)(2,0.061192)
				};

				\addplot[thick,
				color=green,
				mark = o,
				mark indices={1,3,5,7,9,11,13,15,17,19,21},
				mark options={solid},
			    mark size = 3pt,
				]
				coordinates {
                 (0,0.021479)(0.1,0.021546)(0.2,0.021488)(0.3,0.021472)(0.4,0.021493)(0.5,0.021499)(0.6,0.021497)(0.7,0.021494)(0.8,0.021492)(0.9,0.02149)(1,0.021491)(1.1,0.021493)(1.2,0.021498)(1.3,0.021504)(1.4,0.021512)(1.5,0.021521)(1.6,0.021531)(1.7,0.021541)(1.8,0.021552)(1.9,0.021563)(2,0.021574)
				};

			    \addplot[thick,
				color=green,
				mark = x,
				mark indices={1,3,5,7,9,11,13,15,17,19,21},
				mark options={solid},
				mark size = 4pt,
				]
				coordinates {
                 (0,0.021393)(0.1,0.021373)(0.2,0.021378)(0.3,0.021393)(0.4,0.021408)(0.5,0.021422)(0.6,0.021434)(0.7,0.021444)(0.8,0.021453)(0.9,0.021461)(1,0.021467)(1.1,0.021474)(1.2,0.02148)(1.3,0.021485)(1.4,0.02149)(1.5,0.021495)(1.6,0.0215)(1.7,0.021505)(1.8,0.021509)(1.9,0.021514)(2,0.021518)
				};
			
				\addplot[thick,
				color=red,
				mark = halfsquare*,
				mark indices={1,3,5,7,9,11,13,15,17,19,21},
				mark options={solid},
				mark size = 4pt,
				]
				coordinates {
               (0,0.006765)(0.1,0.038217)(0.2,0.052274)(0.3,0.057658)(0.4,0.059596)(0.5,0.060039)(0.6,0.059935)(0.7,0.059663)(0.8,0.059304)(0.9,0.058943)(1,0.058651)(1.1,0.058346)(1.2,0.058119)(1.3,0.057841)(1.4,0.057752)(1.5,0.057676)(1.6,0.057581)(1.7,0.057583)(1.8,0.057676)(1.9,0.057749)(2,0.057735)
				};

				\addplot[thick,
				color=red,
				mark = |,
				mark indices={1,3,5,7,9,11,13,15,17,19,21},
				mark options={solid},
				mark size = 4pt,
				]
				coordinates {
                 (0,0.0067923)(0.1,0.0068136)(0.2,0.0067951)(0.3,0.0067901)(0.4,0.0067967)(0.5,0.0067985)(0.6,0.0067981)(0.7,0.0067971)(0.8,0.0067963)(0.9,0.0067958)(1,0.006796)(1.1,0.0067968)(1.2,0.0067983)(1.3,0.0068003)(1.4,0.0068027)(1.5,0.0068055)(1.6,0.0068086)(1.7,0.0068119)(1.8,0.0068153)(1.9,0.0068188)(2,0.0068223)
				};
				
				\addplot[thick,
				color=red,
				mark = pentagon,
				mark indices={1,3,5,7,9,11,13,15,17,19,21},
				mark options={solid},
				mark size = 4pt,
				]
				coordinates {
                 (0,0.006765)(0.1,0.0067586)(0.2,0.0067604)(0.3,0.0067649)(0.4,0.0067698)(0.5,0.0067742)(0.6,0.006778)(0.7,0.0067813)(0.8,0.006784)(0.9,0.0067865)(1,0.0067886)(1.1,0.0067906)(1.2,0.0067924)(1.3,0.0067941)(1.4,0.0067958)(1.5,0.0067974)(1.6,0.0067989)(1.7,0.0068004)(1.8,0.0068019)(1.9,0.0068033)(2,0.0068047)
				};

				\legend{20 dB (Scenario-I), 20 dB (Scenario-II), 20 dB (Scenario-III), 30 dB (Scenario-I), 30 dB (Scenario-II), 30 dB (Scenario-III), 40 dB (Scenario-I), 40 dB (Scenario-II), 40 dB (Scenario-III) }
				
			\end{semilogyaxis}
		\end{tikzpicture}
		\caption{Theoretical bounds versus $\overline{\kappa}$ for SNR = $20\,$dB, $30\,$dB and $40\,$dB when the UE distance is $5$ meters, $\bbetamin = 0.7$ and $\overline{\phi} = 0$.}
		\label{fig:3}
	\end{figure}	
	
	In Fig.~\ref{fig:3}, for all the three scenarios, the theoretical bounds are plotted versus $\kappa$ for SNRs of $20$, $30$, and $40$ dB when the UE distance is $5$ meters from the center of the RIS, $\bbetamin = 0.7$, and $\overline{\phi} = 0$. Similar to Fig.~\ref{fig:2}, the CRB values for Scenario-II and Scenario-III are almost the same. We also observe that as $\kappa$ approaches $0$, i.e., as the mismatch between the true and assumed models decreases, the LB and the CRB values for Scenario-II and Scenario-III become closer to each other similarly to Fig.~\ref{fig:2}. In addition, as the SNR increases, the performance loss due to the mismatch becomes more significant. Moreover, increasing $\overline{\kappa}$ beyond $\overline{\kappa} = 0.4$ does not have any notable impacts on the LB values.

		\begin{figure}%[H]
		\centering
		\begin{tikzpicture}
			\begin{semilogyaxis}[
				width=8.5cm,
                height=7cm,
				legend style={nodes={scale= 0.7, transform shape},at={(1,1)},anchor=north east}, 
				legend cell align={left},
				legend image post style={mark indices={}},
				xlabel={Number of RIS elements},
				ylabel={Theoretical Bounds [m]},
				xmin=900, xmax=4225,
				ymin=0.01, ymax=1,
				xtick={1000, 1500, 2000, 2500, 3000,  3500, 4000},
				ytick={0.01,0.1, 1},
				legend pos=south west,
				ymajorgrids=true,
				xmajorgrids=true,
				grid style=dashed,
				]
				\addplot[thick, 
				color=blue,
				mark = triangle,
				mark options={solid},
				mark size = 4pt,
				]
				coordinates {
             (900,0.89716)(1225,0.51232)(1600,0.32949)(2025,0.26863)(2500,0.22298)(3025,0.20489)(3600,0.18468)(4225,0.16349)
				};
				
				\addplot[thick, 
				color=blue,
				mark = square,
				mark options={solid},
				mark size = 3pt,
				]
				coordinates {
            (900,0.18899)(1225,0.13817)(1600,0.1048)(2025,0.082871)(2500,0.066931)(3025,0.055725)(3600,0.046987)(4225,0.040188)
				};

				\addplot[thick, 
				color=blue,
				mark = star,
				mark options={solid},
				mark size = 4pt,
				]
				coordinates {
             (900,0.18503)(1225,0.13688)(1600,0.10438)(2025,0.082485)(2500,0.066612)(3025,0.055178)(3600,0.04626)(4225,0.039495)
				};

				\addplot[thick, 
				color=red,
				mark = halfsquare*,
				mark options={solid},
				mark size = 4pt,
				]
				coordinates {
         (900,0.30633)(1225,0.20112)(1600,0.14171)(2025,0.11463)(2500,0.093371)(3025,0.082461)(3600,0.072147)(4225,0.063122)
				};
				
			    \addplot[thick, 
				color=red,
				mark = |,
				mark options={solid},
				mark size = 4pt,
				]
				coordinates {
             (900,0.18781)(1225,0.13823)(1600,0.10496)(2025,0.082909)(2500,0.067027)(3025,0.055597)(3600,0.04669)(4225,0.039879)
				};

				 \addplot[thick, 
				color=red,
				mark = pentagon,
				mark options={solid},
				mark size = 4pt,
				]
				coordinates {
                (900,0.18669)(1225,0.13766)(1600,0.10461)(2025,0.08253)(2500,0.066777)(3025,0.055381)(3600,0.046465)(4225,0.039673)

				};

			\legend{$\bbetamin = 0.3$ (Scenario-I), $\bbetamin = 0.3$ (Scenario-II), $\bbetamin = 0.3$ (Scenario-III),
			$\bbetamin = 0.7$ (Scenario-I), $\bbetamin = 0.7$ (Scenario-II), $\bbetamin = 0.7$ (Scenario-III)
			}

			\end{semilogyaxis}
		\end{tikzpicture}
		\caption{Theoretical bounds versus number of RIS elements and  for $\bbetamin \in \{ 0.3,0.7\}$ when the UE distance is $5$ meters, SNR = 20 dB, $\overline{\kappa} = 1.5$ and $\overline{\phi} = 0$.}
		\label{fig:4}
	\end{figure}
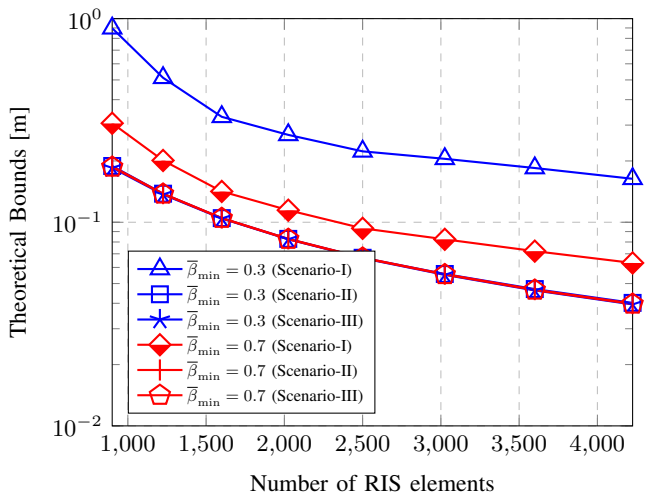
	
	\subsubsection{Effect of RIS Size}
	To investigate the effects of the number of RIS elements, the average LB and CRB values are plotted versus the RIS size in Fig.~\ref{fig:4} for all the three scenarios by averaging over $200$ different random phase profiles for the RIS elements, where the SNR is $20\,$dB, the UE distance is $5$ meters, $\bbetamin\in\{0.3, 0.7\}$, $\overline{\kappa} = 1.5$, and $\overline{\phi} = 0$. We observe that as the RIS size or $\bbetamin$ increases, we obtain lower LB values in general. In addition, the curves for different $\bbetamin$ values are almost parallel. We also note the significant price paid due to the model mismatch: With the perfect knowledge of the RIS model and with $1000$ elements, a similar performance can only be attained using a RIS with $4225$ elements when $\bbetamin=0.3$ under the model mismatch.
	
	\begin{figure}%[H]
		\centering
		\begin{tikzpicture}
			\begin{semilogyaxis}[
				width=8.5cm,
                height=7cm,
				legend style={nodes={scale= 0.7, transform shape},at={(1,1)},anchor=north east}, 
				legend cell align={left},
				legend image post style={mark indices={}},
				xlabel={SNR (dB)},
				ylabel={RMSE [m]},
				xmin=-10, xmax=40,
				ymin=0.005, ymax=25,
				xtick={-10, -5, 0, 5, 10, 15, 20, 25, 30, 35, 40},
				ytick={0.01,0.1,1,10},
				legend pos=south west,
				ymajorgrids=true,
				xmajorgrids=true,
				grid style=dashed,
				]
				
				\addplot[thick, dashed,
				color=blue,
				mark = triangle,
				mark options={solid},
				mark size = 4pt,
				]
				coordinates {
                  (-10,2.1561)(-5,1.2155)(-3,0.96761)(0,0.68897)(3,0.49328)(5,0.39692)(10,0.23928)(12,0.20035)(15,0.15982)(20,0.12459)(25,0.11114)(30,0.10654)(35,0.10504)(40,0.10456)
				};
				
				\addplot[thick, 
				color=blue,
				mark = square,
				mark options={solid},
				mark size = 3pt,
				]
				coordinates {
                  (-10,24.1254)(-5,15.7391)(-3,9.5384)(0,1.2433)(3,0.4648)(5,0.36527)(10,0.23871)(12,0.2027)(15,0.16171)(20,0.13411)(25,0.12283)(30,0.11902)(35,0.11755)(40,0.11778)
				};

				\addplot[thick, dashed,
				color=green,
				mark = o,
				mark options={solid},
				mark size = 3pt,
				]
				coordinates {
                (-10,2.1558)(-5,1.2123)(-3,0.96295)(0,0.68172)(3,0.48262)(5,0.38336)(10,0.21558)(12,0.17124)(15,0.12123)(20,0.068172)(25,0.038336)(30,0.021558)(35,0.012123)(40,0.0068172)
				};
				
				\addplot[thick, 
				color= green,
				mark = +,
				mark options={solid},
				mark size = 4pt,
				]
				coordinates {
                  (-10,23.6576)(-5,15.6283)(-3,9.4396)(0,1.0576)(3,0.51713)(5,0.39085)(10,0.21561)(12,0.17193)(15,0.1221)(20,0.067892)(25,0.038273)(30,0.021933)(35,0.012012)(40,0.0080385)
				};

				\addplot[thick, dashed,
				color=red,
				mark = diamond,
				mark options={solid},
				mark size = 4pt,
				]
				coordinates {
                 (-10,2.1546)(-5,1.2116)(-3,0.96241)(0,0.68134)(3,0.48235)(5,0.38314)(10,0.21546)(12,0.17114)(15,0.12116)(20,0.068134)(25,0.038314)(30,0.021546)(35,0.012116)(40,0.0068134)
				};

				\addplot[thick, 
				color=red,
				mark = x,
				mark options={solid},
				mark size = 4pt,
				]
				coordinates {
                 (-10,23.8779)(-5,17.1544)(-3,11.6126)(0,1.7145)(3,0.6018)(5,0.38118)(10,0.21659)(12,0.17134)(15,0.12262)(20,0.067708)(25,0.038212)(30,0.0218)(35,0.011591)(40,0.0076201)
				};

			\legend{LB [m] (Scenario-I), AMML [m] (Scenario-I), CRB [m] (Scenario-II), AML [m] (Scenario-II),  CRB [m] (Scenario-III), AML [m] (Scenario-III) }
				
			\end{semilogyaxis}

		\end{tikzpicture}
		\caption{Performance of the AMML and the AML algorithms along with the corresponding theoretical bounds versus SNR (dB)  when the UE distance is $5$ meters, $\overline{\kappa} = 1.5$, $\bbetamin = 0.5$  and $\overline{\phi} = 0$.}
		\label{fig:5}
	\end{figure}
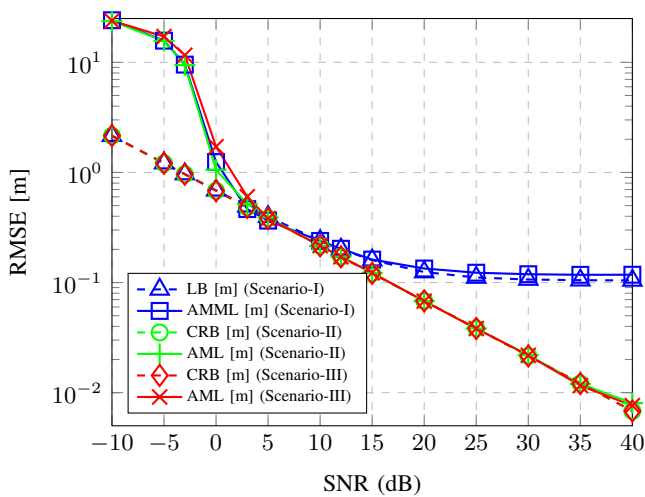

	\subsubsection{Performance of Algorithm~\ref{alg_jacobi} and Algorithm~\ref{alg_loc_est} vs. SNR}\label{sec_num_snr}
    To provide a comparative analysis of the three scenarios, in Fig.~\ref{fig:5}, the performances of the AMML algorithm in Algorithm~\ref{alg_jacobi} and the AML algorithm in Algorithm~\ref{alg_loc_est} are plotted versus SNR, and compared with the corresponding LB and CRB values when $\bbetamin = 0.5$. It is noted that the AMML and AML algorithms achieve the theoretical bounds in the high SNR regime in all the three scenarios. This indicates that the proposed Jacobi-Anger expansion based low-complexity approach in Algorithm~\ref{alg_jacobi} can successfully provide near-optimal solutions to the near-field localization problem in \eqref{eq_mml_est}. Moreover, by comparing the AMML and the AML curves at high SNRs, we observe that the AML algorithm completely recovers the performance loss due to model misspecification, which verifies the effectiveness of the online RIS calibration in Step~1 and the UE location refinement in Step~2 of Algorithm~\ref{alg_loc_est}.
    
    	\begin{figure}%[H]
		\centering
		\begin{tikzpicture}
			\begin{semilogyaxis}[
				width=8.5cm,
                height=7cm,
				legend style={nodes={scale= 0.7, transform shape},at={(1,1)},anchor=north east}, 
				legend cell align={left},
				legend image post style={mark indices={}},
				xlabel={SNR (dB)},
				ylabel={RMSE [m]},
				xmin=-10, xmax=40,
				ymin=0.005, ymax=25,
				xtick={-10, -5, 0, 5, 10, 15, 20, 25, 30, 35, 40},
				ytick={0.01,0.1,1,10},
				legend pos=north east,
				ymajorgrids=true,
				xmajorgrids=true,
				grid style=dashed,
				]
				
				\addplot[thick, dashed,
				color=blue,
				mark = triangle,
				mark options={solid},
				mark size = 4pt,
				]
				coordinates {
                  (-10,2.1561)(-5,1.2155)(-3,0.96761)(0,0.68897)(3,0.49328)(5,0.39692)(10,0.23928)(12,0.20035)(15,0.15982)(20,0.12459)(25,0.11114)(30,0.10654)(35,0.10504)(40,0.10456)
				};
				
				\addplot[thick, 
				color=blue,
				mark = square,
				mark options={solid},
				mark size = 3pt,
				]
				coordinates {
                  (-10,24.1254)(-5,15.7391)(-3,9.5384)(0,1.2433)(3,0.4648)(5,0.36527)(10,0.23871)(12,0.2027)(15,0.16171)(20,0.13411)(25,0.12283)(30,0.11902)(35,0.11755)(40,0.11778)
				};
				
				\addplot[thick, 
				color=blue,
				mark = +,
				mark options={solid},
				mark size = 4pt,
				]
				coordinates {
                   (-10,2.1536)(-5,1.2111)(-3,0.96198)(0,0.68103)(3,0.48213)(5,0.38297)(10,0.21536)(12,0.17107)(15,0.12111)(20,0.068101)(25,0.038296)(30,0.021535)(35,0.012111)(40,0.0068107)
				};
				
				\addplot[thick,
				color=blue,
				mark = o,
				mark options={solid},
				mark size = 3pt,
				]
				coordinates {
                (-10,0.10429)(-5,0.10429)(-3,0.10429)(0,0.10429)(3,0.10429)(5,0.10429)(10,0.10429)(12,0.10429)(15,0.10429)(20,0.10434)(25,0.10434)(30,0.10434)(35,0.10434)(40,0.10434)
				};

				\addplot[thick, dashed, 
				color=red,
				mark = halfsquare*,
				mark options={solid},
				mark size = 4pt,
				]
				coordinates {
                   (-10,2.1337)(-5,1.2007)(-3,0.95433)(0,0.67671)(3,0.48061)(5,0.38319)(10,0.22014)(12,0.17796)(15,0.13173)(20,0.086743)(25,0.066438)(30,0.058569)(35,0.055851)(40,0.054963)
				};
				
				\addplot[thick, 
				color=red,
				mark = x,
				mark options={solid},
				mark size = 4pt,
				]
				coordinates {
                   (-10,24.0045)(-5,14.5693)(-3,9.0685)(0,1.0625)(3,0.45209)(5,0.3524)(10,0.21869)(12,0.17745)(15,0.13389)(20,0.091605)(25,0.073672)(30,0.066818)(35,0.064019)(40,0.063431)
				};
				
				\addplot[thick, 
				color=red,
				mark = pentagon,
				mark options={solid},
				mark size = 4pt,
				]
				coordinates {
                 (-10,2.133)(-5,1.1995)(-3,0.95277)(0,0.67451)(3,0.47752)(5,0.37931)(10,0.2133)(12,0.16943)(15,0.11995)(20,0.067446)(25,0.037928)(30,0.021328)(35,0.011994)(40,0.0067454)
				};
				
				\addplot[thick,  dashed,
				color=red,
				mark options={solid},
				mark size = 4pt,
				]
				coordinates {
                  (-10,0.054446)(-5,0.054446)(-3,0.054446)(0,0.054446)(3,0.054446)(5,0.054446)(10,0.054446)(12,0.054446)(15,0.054446)(20,0.054548)(25,0.054548)(30,0.054548)(35,0.054548)(40,0.054548)
				};
				
			\legend{LB [m] ($\bbetamin = 0.5$), AMML [m] ($\bbetamin = 0.5$), MCRB [m] ($\bbetamin = 0.5$), Bias [m]($\bbetamin = 0.5$),
			LB [m] ($\bbetamin = 0.7$), AMML [m] ($\bbetamin = 0.7$), MCRB [m] ($\bbetamin = 0.7$), Bias [m] ($\bbetamin = 0.7$),
			}
				
			\end{semilogyaxis}

		\end{tikzpicture}
	\caption{AMML, LB, MCRB, and Bias term versus SNR (dB) for $\bbetamin \in \{ 0.5,0.7\}$ when the UE distance is $5$ meters, $\overline{\kappa} = 1.5$ and $\overline{\phi} = 0$.}
		\label{fig:6}
	\end{figure}
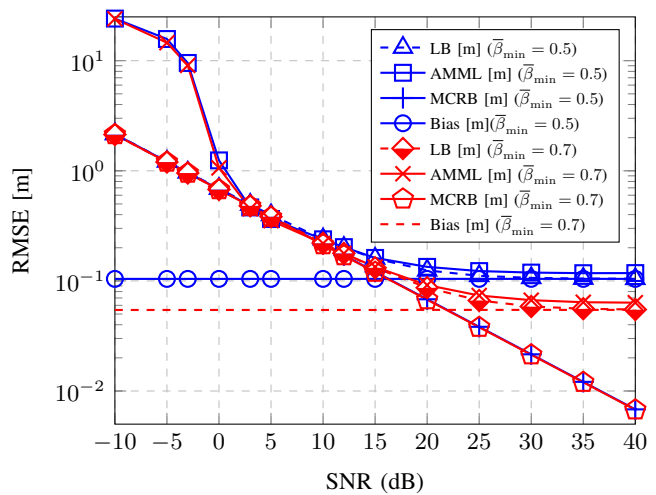

     To explore the asymptotic behavior of the AMML algorithm in Scenario I, its performance with respect to SNR is illustrated in Fig.~\ref{fig:6}   for $\bbetamin = 0.5$ and $0.7$ when the UE distance is $5$ meters. In addition to the performance of the AMML estimator, the LB, the MCRB, and the bias term values are also plotted. We observe that the AMML estimator exhibits three distinct regimes: a low-SNR regime where the AMML is limited by noise peaks and thus far away from the LB; a medium-SNR regime where the AMML is close to the LB, which itself is dominated by the MCRB; and a high-SNR regime, where the AMML and LB are limited by the bias term $(\bet_0)$. 
	
	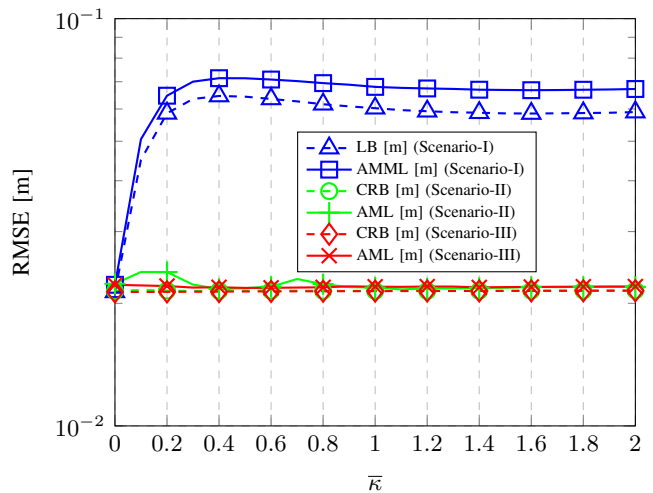
\begin{figure}%[H]
		\centering
		\begin{tikzpicture}
			\begin{semilogyaxis}[
				width=8.5cm,
                height=7cm,
                legend style={at={(0.35,0.55)},anchor=west, nodes = {scale = 0.7}},
				legend cell align={left},
				legend image post style={mark indices={}},				xlabel={$\overline{\kappa}$},
				ylabel={RMSE [m]},
				xmin=0, xmax=2,
				ymin=0.01, ymax=0.1,
                xtick={0, 0.2, 0.4, 0.6, 0.8, 1, 1.2, 1.4, 1.6, 1.8, 2},	ytick={0.01,0.1},
				ymajorgrids=true,
				xmajorgrids=true,
				grid style=dashed,
				]
				
				\addplot[thick, dashed,
				color=blue,
				mark = triangle,
				mark options={solid},
				mark indices = {1, 3, 5, 7, 9, 11, 13, 15, 17, 19, 21},
				mark size = 4pt,
				]
				coordinates {
                       (0,0.021338)(0.1,0.045053)(0.2,0.058549)(0.3,0.063507)(0.4,0.064585)(0.5,0.064337)(0.6,0.063521)(0.7,0.062644)(0.8,0.06164)(0.9,0.060901)(1,0.060273)(1.1,0.059624)(1.2,0.059254)(1.3,0.058918)(1.4,0.058752)(1.5,0.058531)(1.6,0.05849)(1.7,0.058552)(1.8,0.058649)(1.9,0.058741)(2,0.058965)
				};
				
				\addplot[thick, 
				color=blue,
				mark = square,
				mark options={solid},
				mark indices = {1, 3, 5, 7, 9, 11, 13, 15, 17, 19, 21},
				mark size = 3pt,
				]
				coordinates {
                     (0,0.022208)(0.1,0.050598)(0.2,0.064685)(0.3,0.07004)(0.4,0.071418)(0.5,0.071393)(0.6,0.070874)(0.7,0.070259)(0.8,0.069458)(0.9,0.06884)(1,0.067997)(1.1,0.067616)(1.2,0.067398)(1.3,0.067215)(1.4,0.066922)(1.5,0.066818)(1.6,0.066731)(1.7,0.066796)(1.8,0.066903)(1.9,0.067004)(2,0.067183)

				};
				
				\addplot[thick,  dashed,
				color=green,
				mark = o,
				mark options={solid},
				mark indices = {1, 3, 5, 7, 9, 11, 13, 15, 17, 19, 21},
				mark size = 3pt,
				]
				coordinates {
                    (0,0.021437)(0.1,0.021559)(0.2,0.021492)(0.3,0.021442)(0.4,0.021466)(0.5,0.021475)(0.6,0.021472)(0.7,0.021466)(0.8,0.02146)(0.9,0.021455)(1,0.021453)(1.1,0.021453)(1.2,0.021456)(1.3,0.02146)(1.4,0.021467)(1.5,0.021475)(1.6,0.021484)(1.7,0.021494)(1.8,0.021504)(1.9,0.021515)(2,0.021525)

				};
				
				\addplot[thick,  
				color=green,
				mark = +,
				mark options={solid},
				mark indices = {1, 3, 5, 7, 9, 11, 13, 15, 17, 19, 21},
				mark size = 4pt,
				]
				coordinates {
                     (0,0.022336)(0.1,0.023881)(0.2,0.023867)(0.3,0.022228)(0.4,0.021697)(0.5,0.02179)(0.6,0.021967)(0.7,0.022948)(0.8,0.022293)(0.9,0.021917)(1,0.021866)(1.1,0.02164)(1.2,0.02176)(1.3,0.021759)(1.4,0.021823)(1.5,0.021836)(1.6,0.021943)(1.7,0.021913)(1.8,0.021951)(1.9,0.021974)(2,0.021943)

				};

				\addplot[thick,  dashed,
				color=red,
				mark = diamond,
				mark options={solid},
				mark indices = {1, 3, 5, 7, 9, 11, 13, 15, 17, 19, 21},
				mark size = 4pt,
				]
				coordinates {
                     (0,0.021338)(0.1,0.021326)(0.2,0.021336)(0.3,0.021352)(0.4,0.021368)(0.5,0.021382)(0.6,0.021394)(0.7,0.021404)(0.8,0.021413)(0.9,0.02142)(1,0.021427)(1.1,0.021432)(1.2,0.021438)(1.3,0.021442)(1.4,0.021447)(1.5,0.021451)(1.6,0.021456)(1.7,0.02146)(1.8,0.021464)(1.9,0.021468)(2,0.021471)

				};
				
				\addplot[thick,  
				color=red,
				mark = x,
				mark options={solid},
				mark indices = {1, 3, 5, 7, 9, 11, 13, 15, 17, 19, 21},
				mark size = 4pt,
				]
				coordinates {
                 (0,0.022208)(0.1,0.022129)(0.2,0.022051)(0.3,0.021862)(0.4,0.021896)(0.5,0.021821)(0.6,0.021832)(0.7,0.021864)(0.8,0.021894)(0.9,0.02201)(1,0.02197)(1.1,0.021951)(1.2,0.021986)(1.3,0.021987)(1.4,0.021884)(1.5,0.021923)(1.6,0.021937)(1.7,0.021961)(1.8,0.021969)(1.9,0.02198)(2,0.021991)

				};

			\legend{LB [m] (Scenario-I), AMML [m] (Scenario-I), CRB [m] (Scenario-II), AML [m] (Scenario-II),  CRB [m] (Scenario-III), AML [m] (Scenario-III) }
			\end{semilogyaxis}

		\end{tikzpicture}
	\caption{Performance of the AMML and the AML algorithms along with the corresponding theoretical bounds versus $\overline{\kappa}$  when the UE distance is $5$ meters, SNR = 30 dB, $\bbetamin = 0.7$, and $\overline{\phi} = 0$.}
	\label{fig:7}
	\end{figure}
	
	\subsubsection{Performance of Algorithm~\ref{alg_jacobi} and Algorithm~\ref{alg_loc_est} vs. RIS Model Parameters}\label{sec_num_ris}
	To investigate the performance of the proposed localization methods under varying values of RIS model parameters, in Fig.~\ref{fig:7}, the RMSEs of the AMML and AML algorithms are evaluated versus $\overline{\kappa}$ when SNR = 30 dB, $\bbetamin = 0.7$, $\overline{\phi}=0$, and the UE distance is $5$ meters.  Similarly, 
	in Fig.~\ref{fig:8}, the performances of the AMML and AML algorithms versus $\bbetamin$ are shown when SNR = 30 dB, $\overline{\kappa} = 1.5$, $\overline{\phi}=0$, and the UE distance is $5$ meters. From Figs.~\ref{fig:7} and \ref{fig:8}, it is noted that for Scenario-II and Scenario-III, the AML algorithm achieves the CRB, which is insensitive to the values of $\overline{\kappa}$ and $\bbetamin$. In addition, a combined evaluation of Fig.~\ref{fig:5}, Fig.~\ref{fig:7} and Fig.~\ref{fig:8} demonstrates that both of the proposed algorithms can attain the corresponding bounds under a wide variety of settings concerning different SNR levels and RIS model parameters.

    	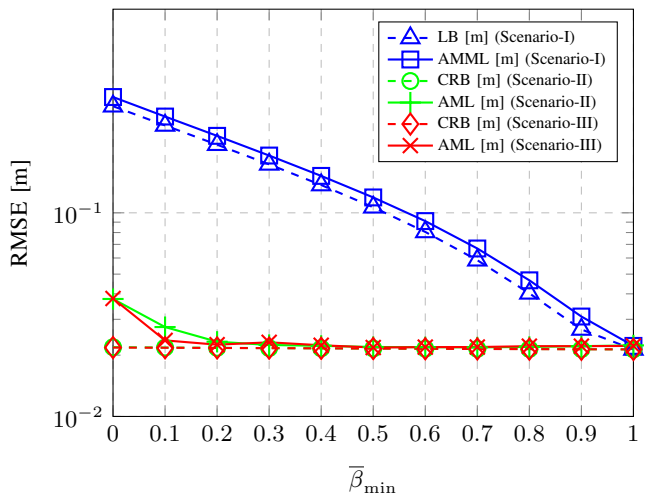
\begin{figure}%[H]
		\centering
		\begin{tikzpicture}
			\begin{semilogyaxis}[
				width=8.5cm,
                height=7cm,
                legend style={at={(1,1)},anchor=west, nodes = {scale = 0.7}},
				legend cell align={left},
				legend image post style={mark indices={}},	legend pos = north east,		xlabel={$\bbetamin$},
				ylabel={RMSE [m]},
				xmin=0, xmax=1,
				ymin=0.01, ymax=1,
                xtick={0, 0.1, 0.2, 0.3, 0.4, 0.5, 0.6, 0.7, 0.8, 0.9, 1},	ytick={0.01,0.1},
				ymajorgrids=true,
				xmajorgrids=true,
				grid style=dashed,
				]
				
				\addplot[thick, dashed,
				color=blue,
				mark = triangle,
				mark options={solid},
				mark size = 4pt,
				]
				coordinates {
                           (0,0.33389)(0.1,0.26864)(0.2,0.2161)(0.3,0.17301)(0.4,0.13696)(0.5,0.10652)(0.6,0.080588)(0.7,0.058515)(0.8,0.040247)(0.9,0.026748)(1,0.021338)

				};
				
				\addplot[thick, 
				color=blue,
				mark = square,
				mark options={solid},
				mark size = 3pt,
				]
				coordinates {
                         (0,0.37045)(0.1,0.29705)(0.2,0.23878)(0.3,0.19115)(0.4,0.15182)(0.5,0.11902)(0.6,0.091088)(0.7,0.066818)(0.8,0.046623)(0.9,0.030955)(1,0.022208)

				};
				
				\addplot[thick,  dashed,
				color=green,
				mark = o,
				mark options={solid},
				mark size = 3pt,
				]
				coordinates {
                        (0,0.021884)(0.1,0.021796)(0.2,0.021724)(0.3,0.021662)(0.4,0.021607)(0.5,0.021558)(0.6,0.021514)(0.7,0.021475)(0.8,0.021441)(0.9,0.021411)(1,0.021354)

				};
				
				\addplot[thick,  
				color=green,
				mark = +,
				mark options={solid},
				mark size = 4pt,
				]
				coordinates {
                         (0,0.037754)(0.1,0.027483)(0.2,0.023282)(0.3,0.022442)(0.4,0.022225)(0.5,0.021933)(0.6,0.021889)(0.7,0.021836)(0.8,0.021956)(0.9,0.022143)(1,0.022336)

				};

				\addplot[thick,  dashed,
				color=red,
				mark = diamond,
				mark options={solid},
				mark size = 4pt,
				]
				coordinates {
                         (0,0.02179)(0.1,0.021747)(0.2,0.021699)(0.3,0.021648)(0.4,0.021597)(0.5,0.021546)(0.6,0.021497)(0.7,0.021451)(0.8,0.02141)(0.9,0.021372)(1,0.021338)

				};
				
				\addplot[thick,  
				color=red,
				mark = x,
				mark options={solid},
				mark size = 4pt,
				]
				coordinates {
                     (0,0.037966)(0.1,0.023634)(0.2,0.022518)(0.3,0.023123)(0.4,0.022368)(0.5,0.0218)(0.6,0.021903)(0.7,0.021923)(0.8,0.022137)(0.9,0.022143)(1,0.022208)

				};

			\legend{LB [m] (Scenario-I), AMML [m] (Scenario-I), CRB [m] (Scenario-II), AML [m] (Scenario-II),  CRB [m] (Scenario-III), AML [m] (Scenario-III) }
			\end{semilogyaxis}

		\end{tikzpicture}
	\caption{Performance of the AMML and the AML algorithms along with the corresponding theoretical bounds versus $\bbetamin$  when the UE distance is $5$ meters, SNR = 30 dB, $\overline{\kappa} = 1.5$, and $\overline{\phi} = 0$.}
	\label{fig:8}
	\end{figure}

	\begin{figure}%[H]
		\centering
		\begin{tikzpicture}
			\begin{semilogyaxis}[
				width=8.5cm,
                height=7cm,
                legend style={at={(1,1)},anchor=west, nodes = {scale = 0.7}},
				legend cell align={left},
				legend image post style={mark indices={}},	legend pos = north east,		xlabel={Iterations, $k$},
				ylabel={Error Values [m]},
				xmin=0, xmax=10,
				ymin=0.005, ymax=1,
                xtick={0, 1, 2, 3, 4, 5, 6, 7, 8, 9, 10},	ytick={0.001,0.01, 0.1,1},
				ymajorgrids=true,
				xmajorgrids=true,
				grid style=dashed,
				]
				
				\addplot[thick,
				color=blue,
				mark = triangle,
				mark options={solid},
				mark size = 4pt,
				]
				coordinates {
                   (0,0.30801)(1,0.29202)(2,0.26402)(3,0.24002)(4,0.23202)(5,0.22402)(6,0.22003)(7,0.22003)(8,0.22402)(9,0.22402)(10,0.22402)

				};

				\addplot[thick, 
				color=black,
				mark = square,
				mark options={solid},
				mark size = 3pt,
				]
				coordinates {
                                     (0,0.17403)(1,0.15404)(2,0.11805)(3,0.094074)(4,0.07809)(5,0.070125)(6,0.066133)(7,0.066133)(8,0.066133)(9,0.066133)(10,0.066133)

				};

				\addplot[thick, 
				color=red,
				mark = o,
				mark options={solid},
				mark size = 3pt,
				]
				coordinates {
                   (0,0.13404)(1,0.11406)(2,0.07809)(3,0.050176)(4,0.030293)(5,0.022399)(6,0.018485)(7,0.018485)(8,0.018485)(9,0.018485)(10,0.018485)

				};

				\addplot[thick, 
				color=green,
				mark = diamond,
				mark options={solid},
				mark size = 3pt,
				]
				coordinates {
                     (0,0.12406)(1,0.10007)(2,0.06411)(3,0.036244)(4,0.016544)(5,0.0058105)(6,0.0061908)(7,0.0061908)(8,0.0061908)(9,0.0058105)(10,0.0058105)

				};

			\legend{$\norm{\hat{\bp}(k)- \overline{\bp}}$ (SNR = 10 dB),
		$\norm{\hat{\bp}(k)- \overline{\bp}}$ (SNR = 20 dB),$\norm{\hat{\bp}(k)- \overline{\bp}}$ (SNR = 30 dB),
		$\norm{\hat{\bp}(k)- \overline{\bp}}$ (SNR = 40 dB)
		}
			\end{semilogyaxis}
		\end{tikzpicture}
	\caption{Errors of estimates of $\bbp$ at each iteration in Step 1 of  Algorithm 2 when $\bbetamin = 0.5$, $\overline{\kappa} = 1.5$, $\overline{\phi} = 0$, and UE distance is 5 meters.}
	\label{fig:9}
	\end{figure}
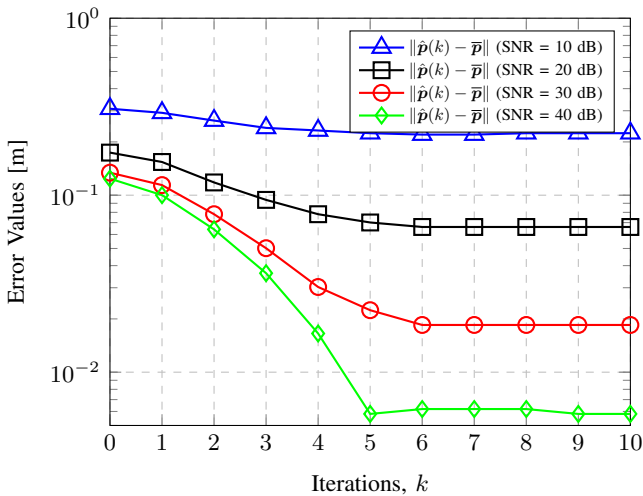
	
	\subsubsection{Convergence Behavior of Algorithm~\ref{alg_loc_est}}
	Finally, for a single realization, Fig.~\ref{fig:9} illustrates the errors of the position estimates at each alternating iteration in Step 1 of Algorithm 2. In the figure, 
	$\hat{\bp}(k)$ denotes the estimate of the position $\overline{\bp}$ obtained from the Jacobi-Anger approach by using 
	the estimate of $\overline{\zetab}$ at the $k$-th iteration.

	We observe that, starting from the initial location estimate given by Algorithm~\ref{alg_jacobi}, Algorithm~\ref{alg_loc_est} provides significant performance gains through the iterations of the online RIS calibration in Step~1. In particular, at an SNR of 40 dB, Algorithm~\ref{alg_loc_est} achieves the error value of $0.0058\,$m, while the error corresponding to Algorithm~\ref{alg_jacobi} is $0.1241\,$m.
	
	\subsubsection{\rev{Numerical Results for Small Number of Observations}} \label{subsec:SmallNumb}
	\rev{In this part, we use the same simulation setup described in Sec.~\ref{sec_sim} by only changing $T$ from $200$ to $10$. In Fig.~\ref{fig:10}, performances of the AMML and AML algorithms versus SNR are presented when $\bbetamin = 0.7$, $\bkappa = 1.5$, $\bphi = 0$. It can be observed that when the number of transmissions is reduced to $10$, the proposed algorithms are still capable of achieving the theoretical performance bounds in all three scenarios.}
	
		\begin{figure}%[H]
		\centering
		\begin{tikzpicture}
			\begin{semilogyaxis}[
				width=8.5cm,
                height=7cm,
				legend style={nodes={scale= 0.7, transform shape},at={(1,1)},anchor=north east}, 
				legend cell align={left},
				legend image post style={mark indices={}},
				xlabel={SNR (dB)},
				ylabel={RMSE [m]},
				xmin=-10, xmax=40,
				ymin=0.01, ymax=35,
				xtick={-10, -5, 0, 5, 10, 15, 20, 25, 30, 35, 40},
				ytick={0.01,0.1,1,10},
				legend pos=south west,
				ymajorgrids=true,
				xmajorgrids=true,
				grid style=dashed,
				]
				
				\addplot[thick, dashed,
				color=blue,
				mark = triangle,
				mark options={solid},
				mark size = 4pt,
				]
				coordinates {
                      (-10,9.0684)(-5,5.1029)(0,2.8755)(5,1.6275)(10,0.93371)(15,0.55665)(20,0.36352)(25,0.2756)(30,0.24121)(35,0.229267605490095)(40,0.2254)
				};
				
				\addplot[thick, 
				color=blue,
				mark = square,
				mark options={solid},
				mark size = 3pt,
				]
				coordinates {
                     (-10,33.9747)(-5,34.1358)(0,28.9072)(5,15.8562)(10,1.126)(15,0.59226)(20,0.37629)(25,0.28537)(30,0.25802)(35,0.254159234202347)(40,0.2557)
				};

				\addplot[thick, dashed,
				color=green,
				mark = o,
				mark options={solid},
				mark size = 3pt,
				]
				coordinates {
                   (-10,9.0523)(-5,5.0905)(0,2.8626)(5,1.6098)(10,0.90523)(15,0.50905)(20,0.28626)(25,0.16098)(30,0.090523)(35,0.050904857535015)(40,0.0286)
				};
				
				\addplot[thick, 
				color= green,
				mark = +,
				mark options={solid},
				mark size = 4pt,
				]
				coordinates {
                  (-10,33.7683)(-5,33.7391)(0,28.3357)(5,15.7571)(10,0.99497)(15,0.50849)(20,0.28528)(25,0.16107)(30,0.092801)(35,0.052905529423252)(40,0.0296)
				};

				\addplot[thick, dashed,
				color=red,
				mark = diamond,
				mark options={solid},
				mark size = 4pt,
				]
				coordinates {
                   (-10,8.0776)(-5,4.5424)(0,2.5544)(5,1.4364)(10,0.80776)(15,0.45424)(20,0.25544)(25,0.14364)(30,0.080776)(35,0.045423948768876)(40,0.0255)
				};

				\addplot[thick, 
				color=red,
				mark = x,
				mark options={solid},
				mark size = 4pt,
				]
				coordinates {
                     (-10,33.2465)(-5,33.0533)(0,30.2038)(5,14.5013)(10,0.93341)(15,0.47684)(20,0.26104)(25,0.14615)(30,0.084479)(35, 0.049953577777559)(40,0.0297)
				};

			\legend{LB [m] (Scenario-I), AMML [m] (Scenario-I), CRB [m] (Scenario-II), AML [m] (Scenario-II),  CRB [m] (Scenario-III), AML [m] (Scenario-III) }
				
			\end{semilogyaxis}

		\end{tikzpicture}
		\caption{\rev{Performance of the AMML and the AML algorithms along with the corresponding theoretical bounds versus SNR (dB)  when the UE distance is $5$ meters, $\overline{\kappa} = 1.5$, $\bbetamin = 0.7$, $\overline{\phi} = 0$, and $T = 10$.}}
		\label{fig:10}
	\end{figure}
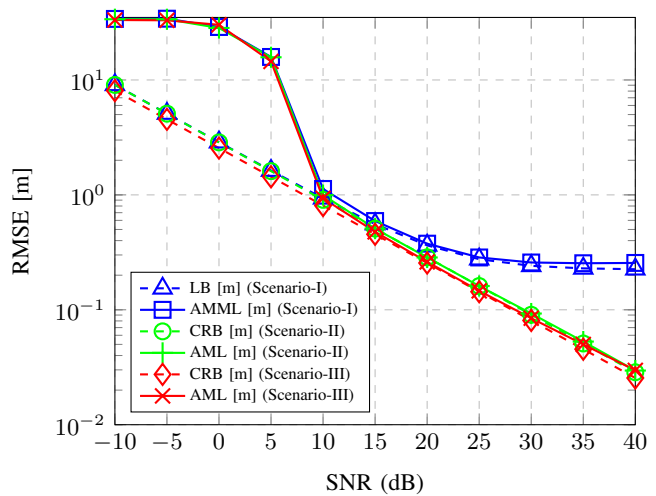

	\section{Concluding Remarks}\label{sec:Conc}
	We have studied the problem of RIS-aided near-field localization under amplitude variations of individual RIS elements as a function of the applied phase shifts, which is a practical model based on equivalent RIS circuit models of reflecting elements \cite{abeywickrama2020intelligent}. First, through the \ac{MCRB} analysis, we have quantified localization performance losses due to model misspecification when the UE is unaware of the RIS amplitude model, and developed an algorithm that achieves the corresponding \ac{LB}. Second, under a known RIS amplitude model, we have derived the corresponding CRB and proposed a low-complexity algorithm for joint UE localization and online calibration of RIS model parameters. Extensive simulations over a broad range of operating conditions demonstrate the following key results:
	\begin{itemize}
	    \item \textit{Significance of the Knowledge of RIS Amplitude Model:} Being unaware of the true RIS amplitude model and assuming conventional unit-amplitude RIS responses can severely degrade the localization accuracy, with the losses being more pronounced at higher SNRs and higher amplitude fluctuations (small $\bbetamin$ and large $\bkappa$ in \eqref{eq_beta_model}).
	    \item \textit{Localization under Model Mismatch:} Under the unknown RIS amplitude model, the proposed Jacobi-Anger expansion based low-complexity method in Algorithm~\ref{alg_jacobi} can provide near-optimal localization performance very close to the corresponding \ac{LB}.
	    \item \textit{Joint Localization and RIS Calibration:} Under the known RIS amplitude model, it is possible to recover the performance losses incurred by model misspecification using the proposed algorithm in Algorithm~\ref{alg_loc_est}, which can calibrate the RIS model online with the aid of an unknown-location UE and subsequently refine the UE location with an accuracy that asymptotically attains the CRB.
	\end{itemize}
	Based on these outcomes, future research will focus on localization-optimal passive beamforming at the RIS under the realistic RIS amplitude model in \eqref{eq_beta_model}, considering wideband signals and unobstructed LoS scenarios. \rev{In addition, data-driven approaches will be considered to learn the mapping from RIS phases to amplitudes in the absence of a specific functional form as in \eqref{eq_beta_model}, which would enable generalizing the proposed approach to any type of RIS.}

	\appendices
	
	%%%%%%%%%%%%%%%%%%%%%%%%%%%%%%%%%%%%%%%%%%%%%%%%%%%%%%%%%%%
	%%%%%%%%%%%%%%%%%%%%%%%%%%%%%%%%%%%%%%%%%%%%%%%%%%%%%%%%%%%
	\section{Proof of Lemma~\ref{lemma_pseudo}} \label{sec:AppA}
	Based on the definition of the KL divergence and the system model in Section~\ref{sec:System}, \eqref{eq:eta0} can be expressed as
	\begin{align}%\nonumber%\label{eq:KLder1}
		\bet_0 &= \argmin_{\bet\in\mathbb{R}^5} \int p(\by) \log \left( \frac{ p(\by)}{ \tilde{p}(\by|\bet)} \right) \, \text{d}\by  \\%\nonumber%\label{eq:KLder2}
		& =\argmin_{\bet\in\mathbb{R}^5} -\int p(\by) \log  \tilde{p}(\by|\bet)  \, \text{d}\by  \\
		& = \argmin_{\bet\in\mathbb{R}^5} \int p(\by)  \norm{\by-\tbmu(\bet)}^2   \, \text{d}\by  \label{eq:KLlast}
	\end{align}
	where the second equality is due to the independence of $p(\by)$ from $\bet$, and the last equality is obtained from \eqref{eq:assumedpdf}. Then, it can be shown that the following equations hold:
	\begin{align}	
		&\int p(\by)  \norm{\by-\tbmu(\bet)}^2  \, d\by  = \sum_{t=1}^{T}\int p(\by)  \abs{y_t -\tmu_t(\bet)}^2 \, \text{d}\by \nonumber \\
		%	& = \sum_{t=1}^{T} \int p(\by|\bbet )\abs{y_t -\tmu_t(\bet)}^2 \, d\by \nonumber \\
		& = \sum_{t=1}^{T} \underbrace{\left (\prod_{t'\neq t} \int p(y_{t'}) \, \text{d}y_{t'}\right)}_{ = 1} \left(\int p(y_t) \abs{y_t -\tmu_t(\bet)}^2 \, \text{d}y_t\right) \nonumber \\
		& = \sum_{t=1}^{T}\int p(y_t) \abs{y_t -\tmu_t(\bet)}^2 \, \text{d} y_t. \label{eq:chainlast}
	\end{align}
	We now introduce $\epsilon_t(\bet)=\mu_t-\tmu_t(\bet)$, so that  $\abs{y_t -\tmu_t(\bet)}^2=\abs{y_t -\mu_t + \epsilon_t(\bet) }^2$, and the integral expression in 
	\eqref{eq:chainlast} can be manipulated as follows:
	\begin{align}\nonumber
		&\int p(y_t) \abs{y_t -\tmu_t(\bet)}^2 \, \text{d} y_t   
		= \int p(y_t) \abs{y_t-\mu_t}^2\, \text{d} y_t + \\
		&\abs{\epsilon_t(\bet)}^2\int p(y_t)  \, \text{d}y_t  \label{eq:derivLem3}
		 + 2 \int p(y_t)\text{Re}\left((y_t-\mu_t)^{*}\epsilon_t(\bet)\right) \, \text{d} y_t. 
		%\\		&  \int p(y_t|\bbet)\left((y_t-\mu_t(\bbet))^{*}\epsilon_t(\bet) + \epsilon_t(\bet)^{*} (y_t-\mu_t(\bbet))\right) \, dy_t 
	\end{align}
	Since $y_t\sim\mathcal{CN}(\mu_t,N_0)$ and $\int p(y_t)  \, \text{d}y_t = 1$, the expression in \eqref{eq:derivLem3} can be simplified as
	\begin{equation}
		\int p(y_t) \abs{y_t -\tmu_t(\bet)}^2 \, \text{d} y_t  = N_0 +  \abs{\epsilon_t(\bet)}^2. \label{eq:integraleq}
	\end{equation}
	By combining \eqref{eq:KLlast}, \eqref{eq:chainlast} and \eqref{eq:integraleq}, we finally obtain that
	\begin{align*}
		\bar{\bet} &= \argmin_{\bet\in\mathbb{R}^5} \sum_{t=1}^{T} \left(N_0 + \abs{\epsilon_t(\bet)}^2\right)
		 = \argmin_{\bet\in\mathbb{R}^5} \sum_{t=1}^{T} \abs{\epsilon_t(\bet)}^2,
		%\argmin_{\bet\in\mathbb{R}^5} \norm{\bep(\bet)}
	\end{align*}
	which completes the proof.

	\section{\rev{Proof of Lemma~\ref{lemma:MCRB}}} \label{sec:AppMCRB}
\rev{
	Let $\bmu^{(k)} = \{\mu_t^{(k)}\}_{t=1}^{T^{(k)}}$ and $\tbmu^{(k)}(\bet) = \{\tilde{\mu}_t^{(k)}(\bet)\}_{t=1}^{T^{(k)}}$ be the set of the noiseless transmitted signals for the $k$-th set of observations under the true model and assumed model, respectively for a given parameter $\bet$.
\eqref{eq:phaseprofile_r} implies that $\mu_t^{(2)} = \mu_{\left( t\text{ mod } T^{(1)}\right)}^{(1)}$ for any $1\leq t\leq T^{(2)}$.
Hence, it can be stated that
\begin{align}
    \norm{\bmu^{(2)}-\tbmu^{(2)}(\bet)}^2 &= \sum_{t=1}^{T^{(2)}} \abs{\mu_t^{(2)}-\tilde{\mu}_t^{(2)}(\bet)}^2 ~, \\
    &= K \norm{\bmu^{(1)}-\tbmu^{(1)}(\bet)}^2 ~,
\end{align}
which implies $\bet_0^{(1)} = \bet_0^{(2)}$. This means that the Bias term in \eqref{eq_lb} remains the same even though the number of transmissions is increased from $T^{(1)}$ to $T^{(2)}$. That is, if $\bias\big(T^{(k)},\bet_0^{(k)}\big)$  denotes the corresponding Bias for the $k$-th set of observations, we have 
$\bias\big(T^{(1)},\bet_0^{(1)}\big) = \bias\big(T^{(2)},\bet_0^{(2)}\big)$.}

\rev{Now, let us compare the MCRB expressions for both cases. To stress the dependence on the number of observations, we will replace $\bA_{\bet_0}$ and $\bB_{\bet_0}$ in \eqref{eq:Aeta0} and \eqref{eq:Beta0} with $\bA_{T,\bet_0}$ and $\bB_{T,\bet_0}$ when the number of transmissions is equal to $T$. By using \eqref{eq:Aeta0_1}, for any $i, j$, it can be seen that
\begin{equation}
    [\bA_{T^{(2)},\bet_0^{(2)}}]_{i,j} = K [\bA_{T^{(1)},\bet_0^{(1)}}]_{i,j} ~. \label{eq:Aeta0_2}
\end{equation}
For $k\in\{1,2\}$ and any $i,j$, we define $\bR_{T^{(k)},\bet_0^{(k)}}$ and $\bS_{T^{(k)},\bet_0^{(k)}}$ as follows:
\begin{align}
  [\bR_{T^{(k)},\bet_0^{(k)}}]_{i,j}  &\triangleq \frac{4}{N_0^2} \Bigg[ \text{Re}\left\{\bep^{(k)}(\bet)^{\mathsf{H}}\frac{\partial\tbmu^{(k)}(\bet)}{\partial\eta_i}\right\} \notag\\
  &\text{Re}\left\{\bep^{(k)}(\bet)^{\mathsf{H}}\frac{\partial\tbmu^{(k)}(\bet)}{\partial\eta_j}\right\} \Bigg]\Bigg|_{\bet = \bet_0^{(k)}},\label{eq:R_T}\\
  [\bS_{T^{(k)},\bet_0^{(k)}}]_{i,j}  &\triangleq \frac{2}{N_0} \Bigg[ \text{Re} \left\{\left(\frac{\partial\tbmu^{(k)}(\bet)}{\partial\eta_i}\right)^{\mathsf{H}}\frac{\partial\tbmu^{(k)}(\bet)}{\partial\eta_j}\right\}\Bigg]\Bigg|_{\bet = \bet_0^{(k)}}.\label{eq:S_T}
\end{align}
It can easily be verified that both $ \bR_{T^{(k)},\bet_0^{(k)}}$ and $\bS_{T^{(k)},\bet_0^{(k)}}$ are positive-definite matrices for any $k\in\{1,2\}$. Based on \eqref{eq:R_T} and \eqref{eq:S_T}, we can write 
\begin{align}
    [\bB_{T^{(2)},\bet_0^{(2)}}] &= [\bR_{T^{(2)},\bet_0^{(2)}}] + [\bS_{T^{(2)},\bet_0^{(2)}}] \notag \\
    & = K^2 [\bR_{T^{(1)},\bet_0^{(1)}}] + K [\bS_{T^{(1)},\bet_0^{(1)}}] \label{eq:Beta0_2}
\end{align}
}
\rev{
By using \eqref{eq:Aeta0_2} and \eqref{eq:Beta0_2}, it is possible to express $\mcrb\big(T^{(2)},\bet_0^{(2)}\big)$ as follows:
\begin{align}
    &\mcrb\big(T^{(2)},\bet_0^{(2)}\big) = \frac{1}{K^2} \big(\bA_{T^{(1)},\bet_0^{(1)}}\big)^{-1} \\
    &\times\big( K^2 [\bR_{T^{(1)},\bet_0^{(1)}}] + K [\bS_{T^{(1)},\bet_0^{(1)}}]\big)  \big(\bA_{T^{(1)},\bet_0^{(1)}}\big)^{-1} \label{eq:MCRBT2}.
\end{align}
Since $ \bR_{T^{(1)},\bet_0^{(1)}}$, $\bS_{T^{(1)},\bet_0^{(1)}}$ and $\bA_{T^{(1)},\bet_0^{(1)}}$ are positive-definite matrices, $\Tr\big\{ \big(  [\bR_{T^{(1)},\bet_0^{(1)}}]\big)  \big(\bA_{T^{(1)},\bet_0^{(1)}}\big)^{-2}\big\}\geq 0$ and $\Tr\big\{ \big(  [\bS_{T^{(1)},\bet_0^{(1)}}]\big)  \big(\bA_{T^{(1)},\bet_0^{(1)}}\big)^{-2}\big\}\geq 0$. This implies that for $K\geq 2$, 
\begin{equation}
    \frac{1}{K} <   \frac{\Tr\big\{\mcrb\big(T^{(2)},\bet_0^{(2)}\big)\big\}}{\Tr\big\{\mcrb\big(T^{(1)},\bet_0^{(1)}\big)\big\}}  < 1 ~. \label{eq:traceineqMCRB}
\end{equation}
}

\rev{Furthermore, since $\bias\big(T^{(1)},\bet_0^{(1)}\big) = \bias\big(T^{(2)},\bet_0^{(2)}\big)$, the latter conclusion for the ratio of the trace of the LBs, is evident from \eqref{eq:traceineqMCRB}.}
\rev{
For the CRB calculation, as stated in Remark~\ref{rem:MCRB}, when there is no mismatch between the assumed model and the true model, we can write
\begin{equation} \label{eq_CRB_A}
    \crb\big(T^{(k)}\big) = \big(\bA_{T^{(k)},\bbet}\big)^{-1}
\end{equation}
for $k\in\{1,2\}$. It is evident from \eqref{eq:Aeta0_2} and \eqref{eq_CRB_A} that
\begin{equation}
    \frac{\Tr\big\{\crb\big(T^{(2)}\big)\big\}}{\Tr\big\{\crb\big(T^{(1)}\big)\big\}} = \frac{\Tr\big\{\big(\bA_{T^{(2)},\bbet}\big)^{-1}\big\}}{\Tr\big\{\big(\bA_{T^{(1)},\bbet}\big)^{-1}\big\}} = \frac{1}{K}
\end{equation}
as we desired to prove.
}
\section{\rev{Computational Complexity of Algorithm~\ref{alg_jacobi}}}
\label{sec:AppCompCostAlg1}
\rev{First assume that $T< \Tthr$. It is clear that complexity of Algorithm~\ref{alg_jacobi} is dominated by the 2-D search for estimating azimuth and elevation angles. In the algorithm we equate $I_{\text{max}}$ to $5$. That is, for Algorithm~\ref{alg_jacobi}, it is sufficient to analyze the computational complexity of the 2-D search given in \eqref{eq_2D_varthetaphi}. Let us define  $\bx(\vartheta,\varphi)\triangleq \tbQ \ba(\vartheta,\varphi) \sqrt{E_s}$. The computational cost of $\bx(\vartheta,\varphi)$ is simply equal to $\mathcal{O}(TM)$. One should note that,  \eqref{eq_2D_varthetaphi} is equivalent to the following problem}

\rev{\begin{equation}
      (\widehat{\vartheta},\widehat{\varphi}) = \argmax_{\vartheta, \varphi} \frac{\bx(\vartheta,\varphi)^{\mathsf{H}} \by}{ \norm{\bx(\vartheta,\varphi)}^2} ~.
\end{equation}
Hence after computing $\bx(\vartheta,\varphi)$, we need to search over $\vartheta$ and $\varphi$. Therefore, the overall cost of Algorithm~\ref{alg_jacobi} for $T\leq \Tthr$ is $\mathcal{O}(TK^2M)$.}

\rev{When $T\geq \Tthr$, the complexity of Algorithm~\ref{alg_jacobi} is dominated by the estimation of $\vartheta$. To estimate $\vartheta$, we need to compute $\bX (\vartheta)\triangleq \tbQ \bG^\trpose(\vartheta)$ first. The computational cost of $\bX (\vartheta)$ is $\mathcal{O}\left(T M (2N+1)\right)$. Then, we need to compute the pseudo-inverse of  $\bX (\vartheta)$, whose computational cost is given by $\mathcal{O}\left(T(2N+1)^2\right)$ since $T\geq 2N+1$. After computing the pseudo-inverse of $\bX(\vartheta)$, we need to search $\vartheta$ over $[0, \pi/2]$ to find the minimum of $\norm{\by-\bX(\vartheta)\bX^{\dagger}(\vartheta) \by }$. That is, the overall cost of Algorithm~\ref{alg_jacobi} when $T\geq \Tthr$ is given by $\mathcal{O}\left(T^2 (2N+1)K\right)  + \mathcal{O}\left(T M (2N+1) K\right) = \mathcal{O} \left( T K (2N+1) \text{max}\{T, M\}\right) = \mathcal{O} \left( T M (2N+1)K\right)$.}

\section{\rev{Computational Complexity of Algorithm~\ref{alg_loc_est}}}
\label{sec:AppCompCostAlg2}	
\rev{Similar to Algorithm~\ref{alg_jacobi}, the computational complexity of Algorithm~\ref{alg_loc_est} is dominated by the 2-D search for $\kappa$ and $\phi$. We need to compute $\norm{\yy -  \alpha\left( \betamin \Gammabt_1(\kappa, \phi) + \Gammabt_2(\kappa, \phi) \right)^\trpose \ba(\pphat)\sqrt{E_s}}$, whose computational cost is $\mathcal{O}(TM)$, for a given $\zetab$ and $\alpha$. Since we search over $\kappa$ and $\phi$'s, the total computational complexity becomes $\mathcal{O}(TML^2)$.}

	%%%%%%%%%%%%%%%%%%%%%%%%%%%%%%%%%%%%%%%%%%%%%%%%%%%%%%%%%%%
	%%%%%%%%%%%%%%%%%%%%%%%%%%%%%%%%%%%%%%%%%%%%%%%%%%%%%%%%%%%
	\section{Derivation of Entries in the MCRB} \label{sec:AppB}
Let $\bet$ be given by $\bet = [\alpha_r\, \alpha_i\, \rmx\, \rmy\, \rmz]^{\trpose}$. Also, define $\bp\triangleq [\rmx\, \rmy\, \rmz]^{\trpose}$, $b_m \triangleq [\bb(\bp)]_m$, and $\alpha \triangleq \alpha_r + j \alpha_i$. We also introduce $\bu=\frac{\bp-\bp_{\text{RIS}}}{\norm{\bp-\bp_{\text{RIS}}}}$and for any $1\leq m\leq M$, $\bu_m=\frac{\bp-\bp_m}{\norm{\bp-\bp_m}}$, where $\bu = [u_x\, u_y \, u_z]^{\trpose}$ and $\bu_m = [u_{m,x} \, u_{m,y} \, u_{m,z}]^{\trpose}$.Then, the first and second derivatives of $\tmu_t(\bet)$ with respect to $\bet$ are given as follows:
\begin{align*}
	\frac{\partial\tmu_t(\bet) }{\partial \alpha_r} &= \sum_{m=1}^{M} b_m \tilde{w}_{t,m} s_t, \, \frac{\partial\tmu_t(\bet) }{\partial \alpha_i} = j\sum_{m=1}^{M} b_m \tilde{w}_{t,m} s_t.
\end{align*}
For $\nu \in\{\rmx\, \rmy\, \rmz\}$, we can write
\begin{align*}
	\frac{\partial\tmu_t(\bet) }{\partial \nu} &= -j \frac{2\pi}{\lambda}\alpha \sum_{m=1}^{M} b_m  \left(u_{m,\nu}-u_{x}\right) \tilde{w}_{t,m} s_t, 
\end{align*}
\begin{align*}
	\frac{\partial^2\tmu_t(\bet) }{\partial \alpha_r \partial \nu} &=-j \frac{2\pi}{\lambda} \sum_{m=1}^{M} b_m  \left(u_{m,\nu}-u_\nu\right) \tilde{w}_{t,m} s_t,
\end{align*}
\begin{align*}
	\frac{\partial^2\tmu_t(\bet) }{\partial \alpha_i \partial \nu} &= j 	\frac{\partial^2\tmu_t(\bet) }{\partial \alpha_r \partial \nu},
\end{align*}
\begin{align*}
	&\frac{\partial^2 \tmu_t(\bet) }{\partial \nu \partial \nu} = -\alpha \frac{4\pi^2}{\lambda^2}\sum_{m=1}^{M} b_m \left(u_{m,\nu}-u_\nu\right)^2  \tilde{w}_{t,m} s_t \\
	&-j\frac{2\pi}{\lambda} \alpha\sum_{m=1}^{M} b_m \left(\frac{1-u^2_{m,\nu}}{\norm{\bp-\bp_m}} -\frac{1-u^2_{\nu}}{\norm{\bp-\bp_{\text{RIS}}}}\right)\tilde{w}_{t,m} s_t.
\end{align*}
Moreover, if $\nu_1, \nu_2\in\{\rmx\, \rmy\, \rmz\}$ and they correspond to different coordinates, it is possible to express  
\begin{align*}
	&\frac{\partial^2 \tmu_t(\bet) }{\partial \nu_1\partial \nu_2} = -\alpha \frac{4\pi^2}{\lambda^2}\sum_{m=1}^{M} b_m  \left(u_{m,\nu_1}-u_{\nu_1}\right)\left(u_{m,\nu_2}-u_{\nu_2}\right) \tilde{w}_{t,m} s_t  \\
	&+j\frac{2\pi}{\lambda} \alpha\sum_{m=1}^{M} b_m \left(\frac{u_{m,\nu_1} u_{m,\nu_2}}{\norm{\bp-\bp_m}}-\frac{u_{\nu_1} u_{\nu_2}}{\norm{\bp-\bp_\text{RIS}}}\right)\tilde{w}_{t,m} s_t.
\end{align*}

%%%%%%%%%%%%%%%%%%%%%%%%%%%%%%%%%%%%%%%%%%%%%%%%%%%%%%%%%%%
%%%%%%%%%%%%%%%%%%%%%%%%%%%%%%%%%%%%%%%%%%%%%%%%%%%%%%%%%%%
\section{\rev{Far-Field Approximation of the RIS Steering Vector in \eqref{eq_ap_nearfield}}}\label{sec_app_farfield}
\rev{In this part, we will show that the near-field steering vector in \eqref{eq_ap_nearfield} degenerates to its standard far-field counterpart when $\norm{\bp-\bpris} \gg \norm{\bp_m - \bpris}$. Specifically, using the geometrical relations based on the definitions before \eqref{eq:CM1}, we can write
\begin{align} \label{eq_dm}
    d_m^2 = d^2 + q_m^2 - 2 d q_m \sin(\vartheta) \cos(\varphi - \psi_m)~,
\end{align}
where $d_m = \norm{\bp - \bp_m} $, $d = \norm{\bp - \bpris} $, $q_m = \norm{\bp_m - \bpris} $ and  $\psi_m$ is the angle between $\bp_m - \bpris$ and the X-axis, i.e., 
\begin{align} \label{eq_pm_pris}
    \bp_m - \bpris = q_m [ \cos(\psi_m) ~ \sin(\psi_m) ~ 0
    ]^\trpose ~.
\end{align}
Re-arranging \eqref{eq_dm} and assuming $d \gg q_m$ yields
\begin{align} \nonumber
    d_m &= d \sqrt{ 1 + \frac{q_m^2}{d^2} - 2 \frac{q_m}{d} \sin(\vartheta) \cos(\varphi - \psi_m) }~,
    \\ \label{eq_fresnel_app}
    &\approx d \left( 1 - \frac{q_m}{d} \sin(\vartheta) \cos(\varphi - \psi_m) \right) ~,
\end{align}
where the approximation in \eqref{eq_fresnel_app} is obtained via first-order Taylor expansion of $f(x) = \sqrt{1 + x^2 - 2 x \kappa}$ around $x = 0$ for constant $\kappa$ \cite[Eq.~(6)]{Fresnel_2011}, \cite[Eq.~(3)]{nearfield_TSP_2005}, \cite[Eq.~(12)]{nearfield_Friedlander_2019}. Using the approximation in \eqref{eq_fresnel_app}, the phase term in \eqref{eq_ap_nearfield} can be written as
\begin{align} \nonumber
    \norm{\bp-\bp_m} - \norm{\bp-\bpris} &= d_m - d ~,
\\ \label{eq_dist_approx}
&\approx -q_m \sin(\vartheta) \cos(\varphi - \psi_m) ~,
\end{align}
which leads to the standard far-field steering vector \cite[Eq.~(9)]{nearfield_Friedlander_2019}:
\begin{align} \nonumber
			[\ba(\bp)]_{m} &= \exp\left(-j \frac{2\pi}{\lambda}\left(\norm{\bp-\bp_m}-\norm{\bp-\bp_{\text{RIS}}}\right)\right) ~,
			\\ \label{eq_ap_farfield}
			&\approx  \exp \left(-j (\pp_m - \ppris)^\trpose \boldsymbol{k}(\vartheta,\varphi)\right) ~,
\end{align}
where $\bk(\vartheta,\varphi)$ is defined in \eqref{eq_k_wavevector}.}
%Note that the far-field steering vector in \eqref{eq_ap_farfield} does not depend on the distance to the RIS, $d$; it is only a function of the angles $\vartheta$ and $\varphi$.}

%%%%%%%%%%%%%%%%%%%%%%%%%%%%%%%%%%%%%%%%%%%%%%%%%%%%%%%%%%%
%%%%%%%%%%%%%%%%%%%%%%%%%%%%%%%%%%%%%%%%%%%%%%%%%%%%%%%%%%%
	\section{\rev{Derivatives for FIM Computation in \eqref{eq:FIMCaseII}}} \label{sec:AppC}
	\rev{The derivatives of $\mu_t$ with respect to the RIS model parameters $\betamin$, $\kappa$ and $\phi$, used for FIM computation in \eqref{eq:FIMCaseII}, are given by}
		\begin{align*}
		\frac{\partial \mu_t (\bet) }{\partial \beta_{\text{min}}} = \alpha \sqrt{E_s} \sum_{m=1}^{M} &[\bb(\bp)]_{m} e^{j\theta_{t,m}}\\
		&\times\left(1-\left(\frac{\sin(\theta_{t,m}-\phi) + 1}{2}\right)^{\kappa}\right),
	\end{align*}
	\begin{align*}
	\frac{\partial \mu_t (\bet) }{\partial \kappa} &= \alpha \sqrt{E_s} \sum_{m=1}^{M} [\bb(\bp)]_{m} e^{j\theta_{t,m}}  (1-\betamin)\\
	&\times \left(\frac{\sin(\theta_{t,m}-\phi) + 1}{2}\right)^\kappa\log\left(\frac{\sin(\theta_{t,m}-\phi) + 1}{2}\right),
\end{align*}
	\begin{align*}
		\frac{\partial \mu_t (\bet) }{\partial \phi} = &-\alpha \sqrt{E_s} \sum_{m=1}^{M} [\bb(\bp)]_{m} e^{j\theta_{t,m}} (1-\beta_{\text{min}}) \\
		&\times
		 \kappa \left(\frac{\sin(\theta_{t,m}-\phi) + 1}{2}\right)^{\kappa-1} \left(\frac{\cos(\theta_{t,m}-\phi)}{2}\right).
	\end{align*}

	\bibliographystyle{IEEEtran}
	\bibliography{bibfile}
	
\end{document}